\documentclass[11pt,letter]{article}

\usepackage{amssymb}
\usepackage{amsmath}
\usepackage{graphicx}
\usepackage{color}
\usepackage{bigints}
\usepackage{array}
\usepackage[normalem]{ulem}
\usepackage{xspace}
\usepackage{enumerate}
\usepackage{hyperref}
\usepackage{numprint}
\usepackage{subcaption}
\usepackage{changepage}
\usepackage{enumerate}
\usepackage{amsthm}
\usepackage{mathrsfs}
\usepackage{algorithm}
\usepackage{algorithmicx}
\usepackage{algpseudocode}
\usepackage{cleveref}
\newcommand{\cancel}[1] {}
\newcommand{\+}[1]{\mathcal{#1}}
\newcommand{\reachVI}[2]{\mathcal{W}^{v_{#1}}_{w_{#2}}}
\newcommand{\reachWJ}[2]{\mathcal{W}^{w_{#1}}_{v_{#2}}}
\newcommand{\appReachVI}[3]{\mathcal{A}^{#1}_{#2,w_{#3}}}
\newcommand{\appReachWJ}[3]{\mathcal{A}^{#1}_{#2,v_{#3}}}
\newcommand{\appReachVIArray}[2]{\mathcal{A}^{#1}_{#2}}
\newcommand{\appReachWJArray}[2]{\mathcal{A}^{#1}_{#2}}
\newcommand{\suf}{\zeta^{\text{suf}}_k}
\newcommand{\pre}{\zeta^{\text{pre}}_k}
\newcommand{\Msuf}{\+M^{\text{suf}}_k}
\newcommand{\Mpre}{\+M^{\text{pre}}_k}

\newcommand{\HQ}[1]{{\color{red} Haoqiang: #1}}

\newtheorem{theorem}{\text{Theorem}}
\newtheorem{lemma}{\text{Lemma}}

\graphicspath{{./}{figures/}}

\begin{document}
	
	\begin{titlepage}
		
		\title{Constant Approximation of Fr\'echet Distance in Strongly Subquadratic Time\footnote{Research supported by Research Grants Council, Hong Kong, China (project no. 16208923).}}

		\author{Siu-Wing Cheng\footnote{Department~of~Computer~Science~and~Engineering,
				HKUST, Hong Kong. Email: {\tt scheng@cse.ust.hk}, {\tt haoqiang.huang@connect.ust.hk}}
			\and 
			Haoqiang Huang\footnotemark[2] \and Shuo Zhang\footnote{Gaoling~School~of~Artificial~Intelligence, 
				Renmin~University~of~China, Beijing, China. Email: {\tt zhangshuo1422@ruc.edu.cn} }}

		\date{}

		\maketitle
		
		\begin{abstract}
			Let $\tau$ and $\sigma$ be two polygonal curves in $\mathbb{R}^d$ for any fixed $d$. Suppose that $\tau$ and $\sigma$ have $n$ and $m$ vertices, respectively, and $m\le n$. While conditional lower bounds prevent approximating the Fr\'echet distance between $\tau$ and $\sigma$ within a factor of 3 in strongly subquadratic time, the current best approximation algorithm attains a ratio of $n^c$ in strongly subquadratic time, for some constant $c\in(0,1)$. We present a randomized algorithm with running time $O(nm^{0.99}\log(n/\varepsilon))$ that approximates the Fr\'echet distance within a factor of $7+\varepsilon$, with a success probability at least $1-1/n^6$. We also adapt our techniques to develop a randomized algorithm that approximates the \emph{discrete} Fr\'echet distance within a factor of $7+\varepsilon$ in strongly subquadratic time. They are the first algorithms to approximate the Fr\'echet distance and the discrete Fr\'echet distance  within constant factors in strongly subquadratic time.
		\end{abstract}
		
		\thispagestyle{empty}
	\end{titlepage}

	\section{Introduction}
	\emph{Fr\'echet distance} is a popular and natural distance metric to measure the similarity between two curves~\cite{Alt2009}. It finds many applications in spatio-temporal data mining~\cite{10.1145/3423334.3431451,Buchin2008DetectingCP}.

	
	Consider a polygonal curve $\tau$ in $\mathbb{R}^d$, with vertices $(v_1,v_2,\ldots,v_n)$. A \emph{parameterization} of $\tau$ is a continuous function $\rho : [0,1] \rightarrow \tau$ such that $\rho(0) = v_1$, $\rho(1) = v_n$, and for all $t, t' \in [0,1]$, $\rho[t]$ is not behind $\rho[t']$ from $v_1$ to $v_n$ along $\tau$ if $t\le t'$. Let $\varrho$ be a parameterization for another polygonal curve $\sigma$. The pair $(\rho, \varrho)$ define a \emph{matching} between $\tau$ and $\sigma$ such that $\rho(t)$ is matched to $\varrho(t)$ for all $t\in [0,1]$.
	The distance realized by $\+M$ is $d_{\+M}(\tau,\sigma) = \max_{t \in [0,1]} d(\rho(t),\varrho(t))$, where $d(\rho(t),\varrho(t))$ is the Euclidean distance between $\rho(t)$ and $\varrho(t)$. The Fr\'{e}chet distance of $\sigma$ and $\tau$ is $d_F(\sigma,\tau) = \inf_{\+M} d_{\+M}(\tau,\sigma)$. If we restrict $\+M$ to match every vertex of $\sigma$ to a vertex of $\tau$, and vice versa, the resulting $\inf_{\+M}d_{\+M}(\tau, \sigma)$ is the \emph{discrete Fr\'echet distance}, which we denote by $\tilde{d}_F(\tau,\sigma)$. The research on computing and approximating the Fr\'echet and discrete Fr\'echet distances continues for many years. 
	
	\vspace{2pt}
	
	\noindent\textbf{Exact algorithms.} Suppose that $\tau$ and $\sigma$ have $n$ and $m$ vertices, respectively, and $m\le n$. Since the first algorithm~\cite{Godau1991ANM} for computing $d_F(\tau, \sigma)$ in $O((n^2m+nm^2)\log(mn))$ time in 1991, there is ongoing research effort towards understanding the computational complexities of $d_F(\tau, \sigma)$ and $\tilde{d}_F(\tau, \sigma)$. The first algorithm~\cite{eiter1994computing} for computing $\tilde{d}_F(\tau, \sigma)$ runs in $O(nm)$ time using dynamic programming. One year later, the running time for computing $d_F(\tau, \sigma)$ was improved to $O(nm\log(mn))$~\cite{AG1995}. For roughly two decades, these remain to be the state-of-the-art. 
	
	In 2013, an asymptotic faster algorithm for computing $\tilde{d}_F(\tau,\sigma)$ in $\mathbb{R}^2$ was eventually presented~\cite{AAKS2013}. It runs in $O(mn\log\log n/\log n)$ time on a word RAM machine of $\Theta(\log n)$ word size with constant-time table lookup capability. Then a faster \emph{randomized} algorithm~\cite{buchin2014four} for computing $d_F(\tau, \sigma)$ in $\mathbb{R}^2$ was shown. It runs in $O(mn\sqrt{\log n}(\log\log n)^{3/2})$ expected time on a pointer machine and in $O(mn(\log\log n)^2)$ expected time on a word RAM machine of $\Theta(\log n)$ word size. The quadratic time barrier was broken recently.
	In $\mathbb{R}^1$, $d_F(\tau, \sigma)$ can be computed in $O(m^2\log^2n+n\log n)$ time~\cite{blank2024faster}. It runs in subquadratic time when $m=o(n)$. In $\mathbb{R}^d$ for $d\ge2$, $d_F(\tau, \sigma)$ can be computed in $O(mn(\log\log n)^{2+c}\log n/\log^{1+c}m)$ time for some fixed $c\in(0,1)$~\cite{Cheng2024FrchetDI}. This is the first $o(mn)$-time algorithm when $m=\Omega(n^\varepsilon)$ for some fixed $\varepsilon\in (0,1)$.

	
	It is unlikely that these running times can be improved significantly due to the famous negative result~\cite{bringmann2014walking}: for any fixed $c\in (0,1)$ and $d\ge 2$, an $O((mn)^{1-c})$-time algorithm for computing $d_F(\tau, \sigma)$ or $\tilde{d}_F(\tau, \sigma)$ in $\mathbb{R}^d$ would refute the widely accepted \emph{Strong Exponential Time Hypothesis} (SETH). Abboud~and~Bringmann~\cite{abboud_et_al:LIPIcs.ICALP.2018.8} further refined the result on $\tilde{d}_F(\tau, \sigma)$. They showed that $\tilde{d}_F(\tau, \sigma)$ is unlikely to be computed in $O(n^2/\log^{7+c}n)$ time for any $c>0$ assuming $m=n$. 
	\vspace{2pt}
	
	\noindent\textbf{Approximation algorithms.} Since the revelation of conditional lower bounds above, several studies focused on improving the running time by allowing approximation. For any $\alpha\in [1,n]$, there is an $O(\alpha)$-approximation of $\tilde{d}_F(\tau, \sigma)$ in $O(n\log n+n^2/\alpha)$ time~\cite{bringmann2016approximability} assuming that $m=n$.
	Later, the running time was improved to $O(n\log n+n^2/\alpha^2)$~\cite{chan2018improved} for any $\alpha\in [1,\sqrt{n/\log n}]$. 
	
	Similar trade-off results have been obtained for Fr\'echet distance as well. Assuming that $m=n$, there is an $O(\alpha)$-approximation algorithm for computing $d_F(\tau, \sigma)$ in $O((n^3/\alpha^2)\log^3n)$ time for any $\alpha\in [\sqrt{n},n]$~\cite{colombe2021approximating}. The result was improved to an $O(\alpha)$-approximation algorithm for computing $d_F(\tau, \sigma)$ in $O((n+mn/\alpha)\log^3n)$ time for any $\alpha\in [1,n]$~\cite{van2023subquadratic}. The running time was further improved to $O((n+mn/\alpha)\log^2n)$~\cite{vanderhorst_et_al:LIPIcs.SoCG.2024.63}. However, for any fixed $c\in (0,1)$, all these known approximation schemes can only achieve approximation ratios of $n^c$ in $O\left((mn)^{1-c}\right)$ time for both Fr\'echet and discrete Fr\'echet distances. Refer to Table~\ref{tab:previous works} for a summary. There are better results for some restricted classes of input curves~\cite{AKW2004,aronov2006frechet,driemel2012approximating,bringmann2015improved,gudmundsson2018fast}.
	
	For (conditional) lower bounds, Bringmann~\cite{bringmann2014walking} proved that no strongly subquadratic algorithm can approximate $d_F(\tau, \sigma)$ within a factor of 1.001 unless SETH fails. As for $\tilde{d}_F(\tau, \sigma)$, there is no strongly subquadratic approximation algorithm of ratio less than 1.399 even in one dimension unless SETH fails~\cite{bringmann2016approximability}. Later, it was shown that no strongly subquadratic algorithm can approximate $d_F(\tau, \sigma)$ or $\tilde{d}_F(\tau, \sigma)$ within a factor of 3 even in one dimension unless SETH fails~\cite{buchin2019seth}. 

	\begin{table}[h!]
		\centering
		\resizebox{\textwidth}{!}{ 
			\begin{tabular}{c|c|c|c}
				Setting & Solution & Running time &Reference\\ \hline
				$d_F$ & Exact &  $O((n^2m+nm^2)\log(mn))$ & \cite{Godau1991ANM}\\
				$\tilde{d}_F$ & Exact& $O(nm)$ & \cite{eiter1994computing}\\
				$d_F$ & Exact& $O(nm\log(mn))$ & \cite{AG1995}\\
				$\tilde{d}_F$ in $\mathbb{R}^2$& Exact & $O(mn\log\log n/\log n)$ & \cite{AAKS2013}\\
				$d_F$ in $\mathbb{R}^2$, pointer machine & Exact& $O(mn\sqrt{\log n}(\log\log n)^{3/2})$ expected & \cite{buchin2014four}\\ 
				$d_F$ in $\mathbb{R}^2$, wordRAM machine& Exact& $O(mn(\log\log n)^2)$ expected & \cite{buchin2014four}\\
				$d_F$ in $\mathbb{R}$ & Exact & $O(m^2\log^2n+n\log n)$ & \cite{blank2024faster}\\
				$d_F$ & Exact & $O(mn(\log\log n)^{2+\mu}\log n/\log^{1+\mu}m)$ expected & \cite{Cheng2024FrchetDI}\\ 
				$\tilde{d}_F$, $m=n$ & $O(\alpha)$-approx. & $O(n\log n+n^2/\alpha)$ & \cite{bringmann2016approximability}\\
				$\tilde{d}_F$, $m=n$ & $O(\alpha)$-approx.& $O(n\log n+n^2/\alpha^2)$ & \cite{chan2018improved} \\
				$d_F$, $m=n$ & $O(\alpha)$-approx.& $O((n^3/\alpha^2)\log^3n)$ & \cite{colombe2021approximating}\\
				$d_F$ &$O(\alpha)$-approx.&  $O((n+mn/\alpha)\log^3n)$ & \cite{van2023subquadratic}\\
				$d_F$ & $O(\alpha)$-approx.&  $O((n+mn/\alpha)\log^2n)$ & \cite{vanderhorst_et_al:LIPIcs.SoCG.2024.63}\\
				$d_F$& $(7+\varepsilon)$-approx. & $O(nm^{0.99}\log n)$ &Theorem \ref{thm:approx_Frechet}\\
				$\tilde{d}_F$ & $(7+\varepsilon)$-approx. & $O(nm^{0.99}\log n)$ & Theorem  \ref{thm:approx_Dis_Frechet}
				
			\end{tabular}
		}
		\caption{Previous results}
		\label{tab:previous works}
	\end{table}
		
    \vspace{2pt}
	
    \noindent\textbf{Our results.} Our main result is a \emph{randomized} $(7+\varepsilon)$-approximation decision procedure for the Fr\'echet distance running in strongly subquadratic time (Theorem~\ref{thm:Frechet}). Specifically, given two polygonal curves $\tau$, $\sigma$ and a fixed value $\delta$, for any $\varepsilon\in (0,1)$, if our decision procedure returns yes, then $d_F(\tau, \sigma)\le (7+\varepsilon)\delta$; if it returns no, then $d_F(\tau, \sigma)>\delta$ with probability at least $1-n^{-7}$. The technical ideas can be adapted to get a \emph{randomized} strongly subquadratic $(7+\varepsilon)$-approximation decision procedure for the discrete Fr\'echet distance as well (Theorem~\ref{thm:Dis Frechet}). 
	
	We follow the approaches in~\cite{colombe2021approximating}~and~\cite{bringmann2016approximability} to turn our decision procedures into $(7+\varepsilon)$-approximation algorithms for computing $d_F(\tau, \sigma)$ and $\tilde{d}_F(\tau, \sigma)$, respectively. Since these approaches only introduce extra polylog factors in the running times, the algorithms still run in strongly subquadratic time.
	They are the first strongly subquadratic algorithms that can achieve constant approximation ratios, which improves the previous ratio of $n^c$ significantly. 
	
	\section{Background}

	Given a polygonal curve $\tau=(v_1, v_2,\ldots, v_n)$ in $\mathbb{R}^d$, let $|\tau|=n$ be its size. Take two points $x,y\in \tau$. We say $x\le_{\tau} y$ if $x$ is not behind $y$ along $\tau$, and we use $\tau[x, y]$ to denote the subcurve of $\tau$ between $x$ and $y$. We use $xy$ to denote an oriented line segment from $x$ to $y$. An array $\+S$ of $n-1$ items is \emph{induced} by $\tau$ if $\+S[i]$ is either a line segment on the edge $v_iv_{i+1}$ or empty for all $i\in [n-1]$. 
	We will also denote $\+S[i]$ by $\+S_{v_i}$. A point $x\in v_iv_{i+1}$ is \emph{covered} by $\+S$ if $x\in \+S_{v_i}$.
	For two points $p, q\in \mathbb{R}^d$, let $d(p,q)$ be the Euclidean distance between $p$ and $q$. For any value $\delta$, we use $\+B(p, \delta)$ to denote the ball of radius $\delta$ centered at $p$.  We present basic knowledge of the Fr\'echet distance in this section.
	
	\vspace{4pt}
 	\noindent{\textbf{Reachability propagation}.} The key concept in determining whether $d_F(\tau, \sigma)\le\delta$ is the \emph{reachability interval}. Take points $x, y\in \tau$ and points $p, q\in \sigma$. For any value $r\ge 0$, the pair $(y,q)$ is \emph{$r$-reachable} from $(x,p)$ if $x\le_\tau y$, $p\le_\sigma q$, and $d_F(\tau[x,y], \sigma[p,q])\le r$. When a pair $(x, p)$ is $r$-reachable from $(v_1, w_1)$, we call $(x,p)$ a \emph{$r$-reachable pair}. All points on the edge $w_jw_{j+1}$ that can form $r$-reachable pairs with $x$ constitute the reachability interval of $x$ on $w_jw_{j+1}$ with respect to $r$. 
 	It is possible that $x$ may not have a reachability interval on some edge of $\sigma$. The reachability intervals of a point $p$ of $\sigma$ can be defined analogously. 
	
	Alt~and~Godau~\cite{AG1995} introduced the free space diagram for the visualization of reachability information. See Figure.~\ref{fig:FS} for an illustration. Specifically, the free space diagram of $\tau$ and $\sigma$ with respect to $r$ is a parameteric space spanned by $\tau$ and $\sigma$ such that each point $(x,p)$ in the diagram corresponds to a pair of points with $x\in \tau$ and $p\in \sigma$. The point $(x,p)$ locates inside the free space with respect to $r$ if and only if $d(x, p)\le r$. Call a path inside the free space diagram bi-monotone if it moves either upwards or to the right. Given two pairs $(x,p)$ and $(y,q)$, $(y,q)$ is $r$-reachable from $(x, p)$ if and only if there is a bi-monotone path from $(x,p)$ to $(y,q)$ inside the free space.

	The planarity of free space diagram provides us with the following useful lemma.

	\begin{lemma}\label{lem:planarity}
		Take $x_1, x_2, y_1, y_2\in \tau$ and $p, q\in \sigma$. Suppose that $x_1\le_\tau x_2\le_\tau y_2\le_\tau y_1$, and $p\le_\sigma q$. If $(y_1, q)$ is $r$-reachable from $(x_1, p)$, and $(y_2, q)$ is $r$-reachable from $(x_2, p)$, then $(y_2,q)$ is $r$-reachable from $(x_1,p)$ as well.
	\end{lemma}

	\begin{proof}
		As illustrated in Figure.~\ref{fig:planarity}. The bi-monotone path between $(x_1, p)$ and $(y_1,q)$ intersects the bi-monotone path between $(x_2, p)$ and $(y_2, q)$. It implies that there is a bi-monotone path from $(x_1, p)$ to $(y_2,q)$ as well. Hence, $(y_2,q)$ is $r$-reachable from $(x_1,p)$.
	\end{proof}

	\begin{figure}
		\centering
		\includegraphics[scale=1]{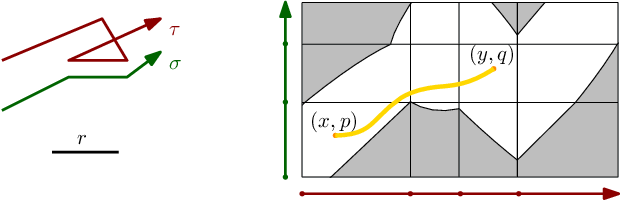}
		\caption{The free space diagram with respect to $r$. The free space is the region in white. The pair $(y,q)$ is $r$-reachable from $(x,p)$, which implies $d_F(\tau[x,y], \sigma[p,q])\le r$.}\label{fig:FS}
	\end{figure}
	
	\begin{figure}
		\centering
		\includegraphics[scale=1.3]{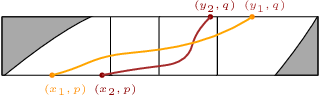}
		\caption{Illustration of proof of Lemma~\ref{lem:planarity}}\label{fig:planarity}
	\end{figure}
	Fix $r=\delta$. Any vertex $v_i$ has at most one reachability interval on every edge of $\sigma$. So does $w_j$. All reachability intervals of $v_i$ can be organized into an array $\+R$ induced by $\sigma$ such that $\+R_{w_j}$ is the reachability interval of $v_i$ on $w_jw_{j+1}$. If $v_i$ does not have a reachability interval on $w_jw_{j+1}$, $\+R_{w_j}$ is empty. We can also organize all reachability intervals of $w_j$ into an array $\+R'$ induced by $\tau$. 
	Given $\tau$, $\sigma$ and $\delta$, Alt~and~Godau~\cite{AG1995} designed an $O(nm)$-time algorithm to compute the reachability intervals for all vertices of $\tau$ and $\sigma$. We will use a variant of this algorithm. We call it \pmb{\sc WaveFront}. Besides $\tau$, $\sigma$ and $\delta$, {\sc WaveFront} accepts two more inputs $\+S$ and $\+S'$, i.e., we will invoke it as \pmb{{\sc WaveFront}$(\tau, \sigma, \delta, \+S, \+S')$}, where $\+S$ and $\+S'$ are arrays induced by $\tau$ and $\sigma$, respectively.

	Let $\+{W}^{v_i}$ be an array induced by $\sigma$ for all $i\in[n]$, and 
	let $\+{W}^{w_j}$ be an array induced by $\tau$ for all $j\in [m]$. Calling {\sc WaveFront}$(\tau, \sigma, \delta,\+S, \+S')$ returns $(\+{W}^{v_i})_{i\in [n]}$ and $(\+{W}^{w_j})_{j\in[m]}$. For every $\+{W}^{v_i}$, a point $q\in \sigma$ is covered by it if and only there is a point $p$ covered by $\+S'$ with $(v_i, q)$ being $\delta$-reachable from $(v_1, p)$ or a point $x$ covered by $\+S$ with $(v_i, q)$ being $\delta$-reachable from $(x, w_1)$. For every $\+{W}^{w_j}$, a point $y\in \tau$ is covered by it if and only if there is a point $p$ covered by $\+S'$ with $(y, w_j)$ being $\delta$-reachable from $(v_1, p)$ or a point $x$ covered by $\+S$ with $(y, w_j)$ being $\delta$-reachable from $(x, w_1)$. 
\cancel{
    \begin{lemma}\label{lem: wave}
		Given two polygonal curves $\tau$ and $\sigma$ in $\mathbb{R}^d$, a value $\delta>0$, an array $\+S$ induced by $\tau$, and an array $\+S'$ induced by $\sigma$, calling {\sc WaveFront}$(\tau, \sigma, \delta, \+S, \+S')$ returns $(\+{W}^{v_i})_{i\in [n]}$ and $(\+{W}^{w_j})_{j\in[m]}$: 
		
	
		\begin{itemize}
			\item $\+{W}^{v_i}$ is an array induced by $\sigma$ for all $i\in [1, n]$. A point $p\in w_jw_{j+1}$ belongs to  $\+W\reachVI{i}{j}$ {\bf if and only if} there is a point $q$ covered by $\+S'$ such that $q\le_{\sigma} p$ and $d_F(\tau[v_1, v_i], \sigma[q, p])\le \delta$ or there is a point $x$ covered by $\+S$ such that $x\le_{\tau} v_i$ and $d_F(\tau[x, v_i], \sigma[w_1, p])\le \delta$.
			\item $\+{W}^{w_j}$ is an array induced by $\tau$ for all $j\in [1, m]$. A point $x\in v_iv_{i+1}$ belongs to $\+W\reachWJ{j}{i}$ {\bf if and only if} there is a point $y$ covered by $\+S$ such that $y\le_{\tau} x$ and $d_F(\tau[y,x], \sigma[w_1, w_j])\le \delta$ or there is a point $p$ covered by $\+S'$ such that $p\le_{\sigma}w_j$ and $d_F(\tau[v_1, x], \sigma[p, w_j])\le \delta$.
		\end{itemize} 
	\end{lemma}
		}


        We present {\sc WaveFront} that runs in $O(nm)$ time. It is based on dynamic programming.

    \vspace{2pt}
	\noindent\pmb{{\sc WaveFront}$(\tau, \sigma, \delta, \+S, \+S')$.} We first initialize $\+{W}^{v_1}$ and $\+{W}^{w_1}$. Set $\+{W}^{v_1}_{w_1}=w_1w_2\cap \+B(v_1, \delta)$ if $w_1\in \+B(v_1, \delta)$ and $v_1\in S_{v_1}$. Otherwise, if $\+S'_{w_1}\not=\emptyset$, set $\+{W}^{v_1}_{w_1}= s_1w_2\cap \+B(v_1, \delta)$, where $s_1$ is the start of $\+S'_{w_1}$, and set $\+{W}^{v_1}_{w_1}=\emptyset$ if $\+S'_{w_1}=\emptyset$. For all $j\in [2, m-1]$, if $w_j\in \+{W}^{v_1}_{w_{j-1}}$, set $\+{W}^{v_1}_{w_j}=w_jw_{j+1}\cap \+B(v_1, \delta)$; otherwise,  if $\+S'_{w_j}\not=\emptyset$, set $\+{W}^{v_1}_{w_j}= s_jw_{j+1}\cap \+B(v_1, \delta)$, where $s_j$ is the start of $\+S'_{w_j}$, and set $\+{W}^{v_1}_{w_j}=\emptyset$ if $\+S'_{w_j}=\emptyset$. We compute $\+{W}^{w_1}_{v_i}$ for all $i\in [n-1]$ similarly.  
	
	We compute the remaining outputs by the following recurrence. Suppose that we have computed $\+{W}^{v_i}_{w_j}$ and $\+{W}^{w_j}_{v_i}$. We can proceed to compute $\+{W}^{v_{i+1}}_{w_j}$ and $\+{W}^{w_{j+1}}_{v_i}$ as follows. Set $\+{W}^{v_{i+1}}_{w_j}$ and $\+{W}^{w_{j+1}}_{v_i}$ to be empty if $\+{W}^{v_i}_{w_j}$ and $\+{W}^{w_j}_{v_i}$ are empty. The reason is as follows. 
    For any point $q\in w_jw_{j+1}$, take a point $p\in \sigma$ with $p\le_\sigma q$, it is either the case that a matching between $\tau[v_1, v_{i+1}]$ and $\sigma[p, q]$ matches $v_i$ to some point in $w_jw_{j+1}$ or $q\le_\sigma w_j$ and $w_j$ is matched to some point in $v_iv_{i+1}$. So are a point $x\in \tau$ and the matching between $\tau[x,v_i]$ and $\sigma[w_1, q]$. Hence, we should set $\+{W}^{v_{i+1}}_{w_j}=\emptyset$ if $\+{W}^{v_i}_{w_j}$ and $\+{W}^{w_j}_{v_i}$ are empty. Similarly, we should also set $\+{W}^{w_{j+1}}_{v_i}=\emptyset$ in this case.
	

	Now suppose that $\+{W}^{w_j}_{v_i}\not=\emptyset$. Pick an arbitrary point $y\in \+{W}^{w_j}_{v_i}$, there is either a point $x$ covered by $\+S'$ such that $(y, w_j)$ is $\delta$-reachable from $(x, w_1)$ or a point $p$ covered by $\+S$ such that $(y, w_j)$ is $\delta$-reachable from $(v_1, p)$. According to the definition of $\delta$-reachable, $y\in \+B(w_j, \delta)\cap v_iv_{i+1}$. Hence, for any point $q\in \+B(v_{i+1}, \delta)\cap w_jw_{j+1}$, the pair $(v_{i+1}, q)$ is $\delta$-reachable from the pair $(y, w_j)$. Because the Fr\'echet distance between $yv_{i+1}$ and $w_jq$ is at most $\delta$. It implies that $q$ is covered by $\+{W}^{v_{i+1}}_{w_j}$. Hence, we set $\+{W}^{v_{i+1}}_{w_j}=\+B(v_{i+1}, \delta)\cap w_jw_{j+1}$.

	
	In the case where $\+{W}^{w_j}_{v_i}=\emptyset$ and $\+{W}^{v_i}_{w_j}\not=\emptyset$, for any point $q\in \+{W}^{v_{i+1}}_{w_j}$, suppose that $(v_{i+1}, q)$ is $\delta$-reachable from $(v_1, p)$ or $(x, w_1)$, where $q$ and $x$ are points covered by $\+S'$ and $\+S$, respectively. Since $\+{W}^{w_j}_{v_i}=\emptyset$, no point $x'\in v_iv_{i+1}$ satisfies that $(x', w_j)$ is $\delta$-reachable from $(v_1, p)$ or $(x, w_1)$. It implies that $(v_{i+1}, q)$ must be $\delta$-reachable from $(v_i, q')$, where $q'$ is some point in $\+{W}^{v_i}_{w_j}$. Hence, all points in $\+{W}^{v_{i+1}}_{w_j}$ must not be in front of the start of $\+{W}^{v_i}_{w_j}$ along $\sigma$. Let $\ell^i_{w_j}$ be the start of $\+{W}^{v_i}_{w_j}$. We set $\+{W}^{v_{i+1}}_{w_j}=\ell^i_{w_j}w_{j+1}\cap \+B(v_{i+1},\delta)$.
	
	
	We can compute $\+{W}^{w_{j+1}}_{v_i}$ in a similar way. By executing the recurrence for all $i\in [n]$ and $j\in [m]$, we get $(\+{W}^{v_i})_{i\in [n]}$ and $(\+{W}^{w_j})_{j\in[m]}$. This completes the description of {\sc WaveFront} which runs in $O(|\tau||\sigma|)$ time.

	\vspace{4pt}

	\noindent{\textbf{Matching construction}.} Suppose that $d_F(\tau, \sigma)\le \delta$. In our decision procedure, we will compute a matching $\+M$ between $\tau$ and $\sigma$ with $d_{\+M}(\tau, \sigma)\le \delta$ explicitly. We need to store $\+M$ such that for any point $p$ in $\tau$ or $\sigma$, $\+M(p)$ can be accessed efficiently, where $\+M(p)$ is a point matched to $p$ by $\+M$. We use the following lemma. 

	\begin{lemma}\label{lem:matching}
		Given two curves $\tau$ and $\sigma$ in $\mathbb{R}^d$, suppose that $d_F(\tau, \sigma)\le \delta$. There is an $O(mn)$-time algorithm for computing a matching $\+M$ between $\tau$ and $\sigma$ such that $d_{\+M}(\tau, \sigma)\le \delta$. The matching $\+M$ can be stored in $O(n+m)$ space such that for any point $x\in \tau$, we can retrieve a point $\+M(x)\in \sigma$ in $O(\log m)$ time, and for any point $p\in \sigma$, we can retrieve a point $\+M(p)\in \tau$ in $O(\log n)$ time.
	\end{lemma}
	
	\begin{proof}
		Set $\+S$ and $\+S'$ to cover only $v_1$ and $w_1$, respectively. We can construct $\+M$ by calling {\sc WaveFront}$(\tau, \sigma, \delta,\+S,\+S')$. Then we go through the output $\+{W}^{v_i}$ and $\+{W}^{w_j}$ in decreasing order of $i$ and $j$. We first determine $\+M(v_i)$ and $\+M(w_j)$ for all vertices $v_i$'s and $w_j$'s.
		
		Since $d_F(\tau, \sigma)\le \delta$, $d(v_n, w_m)\le \delta$ and both $\reachVI{n}{m-1}$ and $\reachWJ{m}{n-1}$ are non-empty. We initialize $\+M(v_n)=w_m$ and $\+M(w_m)=v_n$. We then present a recurrence for determining the matching partners for the remaining vertices. Our recurrence always guarantees that $\+M(v_i)$ belongs to $v_i$'s reachability interval for all $i\in [n]$, and so does $\+M(w_j)$ for all $j\in [m]$.
		
		Now suppose that we have determined $\+M(v_i)$ and $\+M(w_j)$ for some $i> 1$ and some $j>1$, but $\+M(v_{i-1})$ and $\+M(w_{j-1})$ have not been determined yet. The recurrence maintains an invariant that $\+M(w_j)\in v_{i-1}v_i$ or $\+M(v_i)\in w_{j-1}w_{j}$. In the case that $\+M(w_j)\in v_{i-1}v_i$, it implies that $\reachWJ{j}{i-1}\not=\emptyset$. Hence, either $\reachVI{i-1}{j-1}$ or $\reachWJ{j-1}{i-1}$ is non-empty. If $\reachVI{i-1}{j-1}\not=\emptyset$, we pick an arbitrary point $p$ in it and set $\+M(v_{i-1})=p$. Because $d_F(\tau[v_1, v_{i-1}],\sigma[w_1, p])\le \delta$ by the definition of $\reachVI{i-1}{j-1}$, and we can match $pw_j$ to $v_{i-1}\+M(w_j)$ by a linear interpolation as $d(p, v_{i-1})\le \delta$. It is clear that the invariant still holds. In the case where $\reachVI{i-1}{j-1}=\emptyset$ and $\reachWJ{j-1}{i-1}\not=\emptyset$. By the procedure {\sc WaveFront}$(\tau, \sigma, \delta, \+S, \+S')$, the start of $\reachWJ{j-1}{i-1}$ cannot be behind $\+M(w_j)$ along $v_{i-1}v_i$. We set $\+M(w_{j-1})$ to be the start of $\reachWJ{j-1}{i-1}$ and match $\+M(w_{j-1})\+M(w_j)$ to $w_{j-1}w_{j}$ by linear interpolation. The invariant is persevered as well.
		
		In the case that $\+M(v_i)\in w_{j-1}w_j$, it implies that $\reachVI{i}{j-1}\not=\emptyset$. We can proceed to determine $\+M(w_{j-1})$ or $\+M(v_{i-1})$ in a similar way. We can repeat the above procedure until we have determined $\+M(v_i)$ for all $i\in [n]$ or $\+M(w_j)$ for all $j\in [m]$. Suppose we have matched $v_1$ and there are still some vertices $w_1, w_2,\ldots,w_j$ of $\sigma$ remaining to be matched. Provided that $\+M(v_1)$ belongs to $v_1$'s reachability interval in $\sigma[w_j, w_m]$, it implies that the entire subcurve $\sigma[w_1, w_j]$ locates inside $\+B(v_1, \delta)$. We set $\+M(w_{j'})=v_1$ for all $j'\in [1,j]$. In the case where we have matched $w_1$ and there are still some vertices $v_1, v_2,\ldots, v_i$ of $\tau$ remaining to be matched. We can set $\+M(v_{i'})=w_1$ according to the same analysis. It takes $O(nm)$ time to deal with all vertices of $\tau$ and $\sigma$.
		
		We store $\+M(v_i)$ and $\+M(w_j)$ for all $i\in [n]$ and $j\in [m]$. It takes $O(mn)$ time and $O(n+m)$ space, and $d_{\+M}(\tau, \sigma)\le \delta$ according to the procedure. 
		
		Next, we show how to retrieve $\+M(x)$ for any point $x\in \tau$. If $x$ happens to be some vertex $v_i$, we can access $\+M(v_i)$ in $O(1)$ as it is stored explicitly. Otherwise, suppose that $x$ belongs to the edge $v_iv_{i+1}$. We first get $\+M(v_i)$ and $\+M(v_{i+1})$. The entire subcurve $\sigma[\+M(v_i), \+M(v_{i+1})]$ is matched to the edge $v_iv_{i+1}$. If $\sigma[\+M(v_i), \+M(v_{i+1})]$ is a line segment and does not contain any vertices of $\sigma$, it means that the matching between $v_iv_{i+1}$ and $\sigma[\+M(v_i), \+M(v_{i+1})]$ is a linear interpolation between them. We can calculate $\+M(x)$ in $O(1)$ time. 
		
		If $\sigma[\+M(v_i), \+M(v_{i+1})]$ contains vertices $w_j, w_{j+1}, \ldots, w_{j_1}$ of $\sigma$, then $\+M(w_{j'})\in v_iv_{i+1}$ for all $j'\in [j, j_1]$. The points $\+M(w_{j'})$'s partition $v_iv_{i+1}$ into $j_1-j+2 = O(|\sigma|)$ disjoint segments. By a binary search, we can find out $x$ belongs to which segment in $O(\log|\sigma|)$ time. Suppose that $x\in \+M(w_{j'})\+M(w_{j'+1})$. Given that the matching between $\+M(w_{j'})\+M(w_{j'+1})$ and $w_{j'}w_{j'+1}$ is a linear interpolation, we can proceed to calculate $\+M(x)$ in $O(1)$ time.
		
		For any point $p\in \sigma$, we can retrieve $\+M(p)$ in $O(\log|\tau|)$ time similarly.
		
	\end{proof}

	\noindent{\textbf{Curve simplification}.} We will need to simplify $\tau$ to another curve of fewer vertices. 
	The \emph{curve simplification problem} under the Fr\'echet distance has been studied a lot in the literature~\cite{agarwal2005near,bringmann2019polyline,cheng2022curve,cheng2023solving,SHJ,van2018optimal,van2019global,guibas1993approximating}. We employ an almost linear-time algorithm developed in~\cite{SHJ} recently. Any other polynomial time algorithms that can achieve the same performance guarantee can also fit into our framework.
	
	\begin{lemma}[\cite{SHJ}]\label{lem: simp}
		Given a curve $\tau$ in $\mathbb{R}^d$, a fixed value $\delta$, for any $\epsilon\in (0,1)$, there is an $O\left(\epsilon^{-\alpha}n\log (1/\varepsilon)\right)$-time algorithm that computes a curve $\tau'$ of size at most $2(k^*-1)$ 
		with $d_F(\tau, \tau')\le (1+\varepsilon)\delta$, where $\alpha=2(d-1)\lfloor d/2\rfloor^2+d$ and $k^*=\min \{|\tau''|: d_F(\tau, \tau'')\le \delta\}$.
	\end{lemma}

	\section{Overview of the decision procedure}\label{sec: overview}
	We focus on describing the decision procedure for the Fr\'echet distance at a high level. The ideas can be adapted to the discrete Fr\'eche distance. 
	
	Given $\tau$ and $\sigma$, 
    we carefully select a set of vertices and approximate their reachability intervals. For any $i\in [n]$ and $j\in [m-1]$, a line segment on $w_jw_{j+1}$ is an \emph{$\alpha$-approximate reachability interval} of $v_i$ if it lies inside $\+B(v_i, \delta)$, it contains $v_i$'s reachability interval on $w_jw_{j+1}$, and every point in it forms an $(\alpha\delta)$-reachable pair with $v_i$. 
	The vertex $v_i$ can have infinitely many $\alpha$-approximate reachability intervals on $w_jw_{j+1}$. For any subcurve $\sigma[w_{j_1}, w_{j_2}]$ of $\sigma$, an array $\+S$ induced by $\sigma[w_{j_1}, w_{j_2}]$ is \emph{$\alpha$-approximate reachable} for $v_i$ 
	if $\+S_{w_{j}}$ is an $\alpha$-approximate reachability interval for all $j\in [j_1, j_2-1]$. We can define approximate reachability intervals for a vertex $w_j$ of $\sigma$ similarly.

	
	For any fixed $\varepsilon\in (0,1)$, our decision procedure aims to compute a $(7+\varepsilon)$-approximate reachability interval for $w_m$ on the edge $v_{n-1}v_n$. 
	We first divide $\tau$ and $\sigma$ into short subcurves. Let $\mu_1$ and $\mu_2$ be two integers whose values will be specified later with $\mu_2=\mu_1^{1-c}$ for some constant $c\in (0,1)$. We assume that $n-1$ and $m-1$ are multiples of $\mu_1$ and $\mu_2$, respectively. Define $a_k=(k-1)\mu_1+1$ for $k\in [(n-1)/\mu_1]$. We divide $\tau$ into a collection $(\tau_1, \tau_2,\ldots, \tau_{(n-1)/\mu_1})$ of $(n-1)/\mu_1$ subcurves such that $\tau_k=\tau[v_{a_k}, v_{a_{k+1}}]$. Every subcurve $\tau_k$ contains $\mu_1$ edges, and $\tau_{k-1}\cap \tau_k=v_{a_k}$. 
	
	We divide $\sigma$ into shorter subcurves. 
    Define $b_l=(l-1)\mu_2+1$ for $l\in [(m-1)/\mu_2]$. We divide $\sigma$ into a collection $(\sigma_1, \sigma_2,\ldots, \sigma_{(m-1)/\mu_2})$ of $(m-1)/\mu_2$ subcurves such that $\sigma_l=\sigma[w_{b_l}, w_{b_{l+1}}]$. Every subcurve $\sigma_l$ contains $\mu_2$ edges, and $\sigma_{l-1}\cap \sigma_l=w_{b_l}$.
	The heart of our decision procedure is a subroutine {\sc Reach} that works for any $k$ and $l$ and a constant $\beta>1$. 
	
	\vspace{2pt}
	
	\noindent\fbox{\parbox{\dimexpr\linewidth-2\fboxsep-2\fboxrule\relax}{\textbf{Procedure {\sc Reach}}

	\vspace{2pt}

	\textbf{Input:} a subcurve $\tau_k$, a subcurve $\sigma_l$, an array $\appReachVIArray{a_k}{l}$ induced by $\sigma_l$ that is $\beta$-approximate reachable for $v_{a_k}$, an array $\appReachWJArray{b_l}{k}$ induced by $\tau_k$ that is $\beta$-approximate reachable for $w_{b_l}$. 
	
	\vspace{2pt}
	
	\textbf{Output:} an array $\appReachVIArray{a_{k+1}}{l}$ induced by $\sigma_l$ that is $\beta$-approximate reachable for $v_{a_{k+1}}$ and an array $\appReachWJArray{b_{l+1}}{k}$ induced by $\tau_k$ that is $\beta$-approximate reachable for $w_{b_{l+1}}$.
	
	
	}	 
	}
	
	\vspace{2pt}
	
	We will design {\sc Reach} to run in $O(\left(\mu_1\mu_2\right)^{1-c})$ time with $\beta=7+\varepsilon$ for some fixed $c\in (0,1)$ and $\varepsilon\in (0,1)$. 
	To use {\sc Reach}, we first compute all reachability intervals of $v_1$ and $w_1$ in $O(n+m)$ time. We then generate $\appReachVIArray{1}{l}$ and $\appReachWJArray{1}{k}$ as follows. For every $l\in [(m-1)/\mu_2]$ and every $j\in [b_l, b_{l+1}-1]$, set $\appReachVI{1}{l}{j}=\reachVI{1}{j}$. By definition, $\appReachVIArray{1}{l}$ is induced by $\sigma_l$ and is $(7+\varepsilon)$-approximate reachable for $v_1$ for all $l\in [(m-1)/\mu_2]$. We can generate $\appReachWJArray{1}{k}$ for $w_1$ and all $k\in [(n-1)/\mu_1]$ similarly. It takes $O(n+m)$ time so far.
	
	We are now ready to invoke {\sc Reach}$(\tau_1, \sigma_1, \appReachVIArray{a_1}{1}, \appReachWJArray{b_1}{1})$ to get $\appReachVIArray{a_2}{1}$ induced by $\sigma_1$ that is $(7+\varepsilon)$-approximate reachable for $v_{a_2}$ and $\appReachWJArray{b_2}{1}$ induced by $\tau_1$ that is $(7+\varepsilon)$-approximate reachable for $w_{b_2}$. We proceed to invoke the procedure {\sc Reach}$(\tau_2, \sigma_1, \appReachVIArray{a_2}{ 1}, \appReachWJArray{b_1}{2})$.  We can repeat the process for all $k\in [(n-1)/\mu_1]$ to get $\appReachWJArray{b_2}{k}$ of $w_{b_2}$ for all $k\in [(n-1)/\mu_1]$.
	
	We proceed to invoke {\sc Reach}$(\tau_1, \sigma_2, \appReachVIArray{a_1}{2}, \appReachWJArray{b_2}{1})$. We can repeat the process above to get $\appReachWJArray{b_3}{k}$ of $w_{b_3}$ for all $k\in [(n-1)/\mu_1]$. In this way, we can get the $(7+\varepsilon)$-approximate reachability intervals for $w_m$ after calling {\sc Reach} for $mn/(\mu_1\mu_2)$ times. In the end, we finish the decision by checking whether $v_n$ is covered by $\appReachWJArray{m}{(n-1)/\mu_1}$ in $O(1)$ time. If so, return yes; otherwise, return no.
	
	The correctness of our decision procedure is guaranteed by the definition of the $(7+\varepsilon)$-approximate reachability interval. The running time is dominated by invocations of {\sc Reach}. If {\sc Reach} runs in $O\left((\mu_1\mu_2)^{1-c} \right)$ time, the total time is $mn/(\mu_1\mu_2)\cdot O\left((\mu_1\mu_2)^{1-c} \right) =O(mn/(\mu_1\mu_2)^c)$. By choosing the values for $\mu_1$ and $\mu_2$ to make $(\mu_1\mu_2)^c=\Omega\left(m^{c_1}\right)$ for some constant $c_1\in (0,1)$, we can realize an $O\left(nm^{1-c_1}\right)$-time decision procedure in the end.
	
	
	\subsection{Technical overview of {\sc Reach}} 

	Given $\tau_k$, $\sigma_l$, and arrays $\appReachVIArray{a_k}{l}$ and $\appReachWJArray{b_l}{k}$ that are $\alpha$-approximate reachable for $v_{a_k}$ and $w_{b_l}$, respectively, the key to implement the procedure {\sc Reach} is to cover all points in the reachability intervals of $v_{a_{k+1}}$ and $w_{b_{l+1}}$ on $\tau_k$ and $\sigma_l$. We first introduce four types of points in $\tau_k$ and $\sigma_l$ to be covered. The underlying intuition that guides the classification can be found in Figure.~\ref{fig:classification}. We will show that all points in the reachability intervals of $v_{a_{k+1}}$ and $w_{b_{l+1}}$ belong to at least one of these four types. We then describe how to cover points of every type fast.

	\begin{figure}
		\centering
		\includegraphics[scale=0.7]{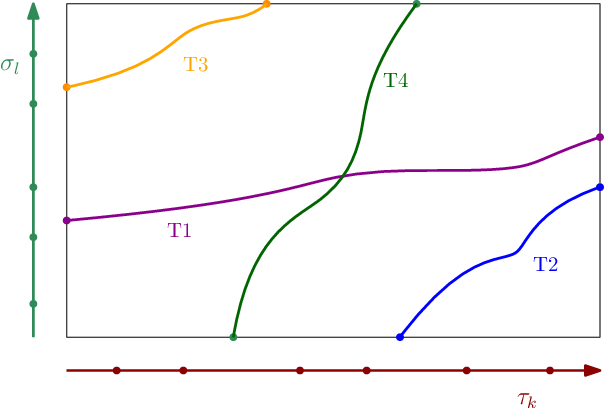}
		\caption{Free space diagram induced by $\tau_k$ and $\sigma_l$ with respect to $\delta$. The input arrays $\appReachVIArray{a_k}{l}$ and $\appReachWJArray{b_l}{k}$ correspond to intervals on the left and bottom boundaries. The targeted output arrays correspond to intervals on the right and upper boundaries where we can reach from $\appReachVIArray{a_k}{l}$ and $\appReachWJArray{b_l}{k}$ via bi-monotone paths in the free space. These paths can be classified into four types depending on where they start and end.}\label{fig:classification}
	\end{figure}
	
	\subsubsection{Classification of target points}
	The points in in $\tau_k$ and $\sigma_l$ that we are going to cover can be classified into the following four types.

	\vspace{4pt}
	
	\noindent\fbox{\parbox{\dimexpr\linewidth-2\fboxsep-2\fboxrule\relax}{
		\noindent\textbf{T1.} $q\in \sigma_l$: There is a point $p$ covered by $\appReachVIArray{a_k}{l}$ such that $(v_{a_{k+1}},q)$ is $\delta$-reachable from $(v_{a_k},p)$. 
		
		\noindent\textbf{T2.} $q\in \sigma_l$: There is a point $x$ covered by $\appReachWJArray{b_l}{k}$ such that $(v_{a_{k+1}},q)$ is $\delta$-reachable from $(x, w_{b_l})$. 
		
		\noindent\textbf{T3.} $y\in \tau_k$: There is a point $p$ covered by $\appReachVIArray{a_k}{l}$ such that $(y, w_{b_{l+1}})$ is $\delta$-reachable from $(v_{a_k}, p)$. 
		
		\noindent\textbf{T4.} $y\in \tau_k$: There is a point $x$ covered by $\appReachWJArray{b_l}{k}$ such that $(y, w_{b_{l+1}})$ is $\delta$-reachable from $(x, w_{b_l})$. 
}}
	
	\vspace{4pt}

	The classification of points is useful due to the following lemma.

	\begin{lemma}\label{lem:cover1}
		Every point in the reachablity intervals of $v_{a_{k+1}}$ on $\sigma_l$ and $w_{b_{l+1}}$ on $\tau_k$ belongs to at least one type.
	\end{lemma}

	\begin{proof}
		Take a point $q\in\sigma_l$ that belongs to $v_{a_{k+1}}$'s reachability interval. The pair $(v_{a_{k+1}}, q)$ is $\delta$-reachable from $(v_1, w_1)$. The Fr\'echet matching between $\tau[v_1, v_{a_{k+1}}]$ and $\sigma[w_1, q]$ either matches $v_{a_k}$ to some point $p\in \sigma_l$ or matches $w_{b_l}$ to some point $x\in \tau_k$. Since $d_F(\tau[v_1, v_{a_{k+1}}], \sigma[w_1,q])\le\delta$, the points $p$ and $x$ must belong to the reachability intervals of $v_{a_k}$ and $w_{b_l}$, respectively, if they exist. It implies that $p$ and $x$ are covered by $\appReachVIArray{a_k}{l}$ and $\appReachWJArray{b_l}{k}$, respectively. In addition, $(v_{a_{k+1}},q)$ is $\delta$-reachable from either $(v_{a_k}, p)$ or $(x, w_{b_l})$. Hence, $q$ belongs to either type 1 or 2. We can prove that points in $w_{b_{l+1}}$'s reachability intervals on $\tau_k$ belong to type 3 and 4 similarly.
	\end{proof}

	\subsubsection{Covering points of each type}
	While points of each type require separate handling, the speedup of covering the points in T1-T3 is achieved by the use of curve simplification. A simple motivating example is as follows. Given two curves $\zeta_1$ and $\zeta_2$ together with a non-negative value $r$, suppose that $\zeta_1$ can be simplified to another curve $\zeta'_1$ such that $|\zeta'_1|\ll|\zeta_1|$ and $d_F(\zeta_1, \zeta'_1)=O(r)$. We can use $\zeta'_1$ as a surrogate of $\zeta_1$ and compute $d_F(\zeta'_1, \zeta_2)$ to determine whether $d_F(\zeta_1, \zeta_2)\le r$ approximately. It takes $O(|\zeta'_1||\zeta_2|)=o(|\zeta_1||\zeta_2|)$ time excluding the time spent on curve simplification. 
    A similar idea has been used by Blank~and~Driemel~\cite{blank2024faster} to achieve a faster algorithm in computing Fr\'echet distance in one-dimensional space.
    The remaining type (T4) is the most challenging to deal with. We resort to sampling and preprocessing. The high level idea of handling each type is as follows. 
	
	\vspace{2pt}
	
	\noindent\textbf{Type 1.} We construct an array $\+I^1$ induced by $\sigma_l$ such that $\+I^1$ covers all points of type 1, and every point covered by $\+I^1$ forms a $((7+\varepsilon)\delta)$-reachable pair with $v_{a_{k+1}}$. 
	For any point $q$ of type 1, there is a point $p\in \sigma_l$ such that $d_F(\tau_k, \sigma[p, q])\le \delta$. 
	Note that $\sigma[p, q]$ contains at most $\mu_2+1$ vertices. Hence, we can simplify $\tau_k$ to a curve $\zeta_k$ of $O(\mu_2)$ vertices in $\tilde{O}(\mu_1^2)$ time\footnote{We hide $\epsilon^{-O(d)}$ and polylog factors in $\tilde{O}$.} such that $d_F(\tau_k, \zeta_k)\le (1+\epsilon)\delta$ by Lemma~\ref{lem: simp}. Then $d_F(\zeta_k,\sigma[p,q])\le(2+\epsilon)\delta$ by the triangle inequality. We use $\zeta_k$ as a surrogate of $\tau_k$. Specifically, let $\+S^k$ be an array induced by $\zeta_k$ with all elements being empty, and let $v'_a$ be the last vertex of $\zeta_k$. We invoke {\sc WaveFront}$(\zeta_k, \sigma_l, (2+\epsilon)\delta, \+S^k, \appReachVIArray{a_k}{l})$ to get the output array $\+{W}^{v'_a}$ for $v'_a$ in $O(\mu_2^2)$ time. We then set $\+I^1=\+{W}^{v'_a}$. 
    Since all elements in $\+S^k$ are empty, $\+{W}^{v'_a}$ covers all points $p'\in \sigma_l$ such that there is a point $p$ covered by $\appReachVIArray{a_k}{l}$ with $p\le_\sigma p'$ and $d_F(\zeta_k, \sigma[p, p'])\le (2+\epsilon)\delta$. Hence, $q$ is covered by $\+I^1$. 
    It may happen that $\tau_k$ cannot be simplified to a curve of $O(\mu_2)$ vertices. In this case, the points of type 1 do not exist, and we do not need to worry about it.

	\vspace{2pt}
	
	\noindent\textbf{Type 2.} We construct an array $\+I^2$ induced by $\sigma_l$ such that $\+I^2$ covers all points of type 2, and every point covered by $\+I^2$ forms a $((7+\varepsilon)\delta)$-reachable pair with $v_{a_{k+1}}$. 
	For any point $q$ of type 2, there is a point $x\in \tau_k$ with $d_F(\tau[x, v_{a_{k+1}}], \sigma[w_{b_l}, q])\le \delta$. Note that $\tau[x, v_{a_{k+1}}]$ is a suffix of $\tau_k$, and $\sigma[w_{b_l}, q]$ has at most $\mu_2+1$ vertices. 
    Hence, $x$ lies in some suffix of $\tau_k$ that can be simplified to a curve of $\mu_2+1$ vertices with respect to $\delta$. 
	Therefore, we identify the longest suffix $\tau[v_{i_{\text{suf}}},v_{a_{k+1}}]$ of $\tau_k$ such that using Lemma~\ref{lem: simp} on it returns a curve $\suf$ of at most $2\mu_2$ vertices. Since the point $x$ must belong to $\tau[v_{i_{\text{suf}}},v_{a_{k+1}}]$, we use $\suf$ as a surrogate of the suffix to compute $\+I^2$. Specifically, let $\+S'$ be an array induced by $\sigma_l$ with all elements being empty. Let $\bar{v}_a$ be the last vertex of $\suf$. We also construct an array $\+S$ induced by $\suf$ based on $\appReachWJArray{b_l}{k}$ and a matching $\Msuf$ with $d_{\Msuf}(\tau[v_{i_{\text{suf}}},v_{a_{k+1}}],\suf)\le(1+\epsilon)\delta$. We ensure that $\Msuf(x)$ is covered by $\+S$. The detailed construction of $\+S$ in $O(\mu_1)$ time will be described in the next section. We invoke {\sc WaveFront}$(\suf, \sigma_l, (2+\epsilon)\delta, \+S, \+S')$ to get the output array $\+{W}^{\bar{v}_a}$ for $\bar{v}_a$ in $O(\mu_2^2)$ time. We then set $\+I^2=\+{W}^{\bar{v}_a}$. 
    By the triangle inequality, $d_F(\suf[\Msuf(x),\bar{v}_a],\sigma[w_{b_l}, q])\le(2+\epsilon)\delta$. Hence, $q$ is covered by $\+I^2$. 

	\vspace{2pt}
	
	\noindent\textbf{Type 3.} We construct an array $\+I^3$ induced by $\tau_k$ such that $\+I^3$ covers all points of type 3, and every point covered by $\+I^3$ forms a $((7+\varepsilon)\delta)$-reachable pair with $w_{b_{l+1}}$. 
	The construction of $\+I^3$ is analogous to that of $\+I^2$. For any point $y$ of type 3, there is a point $p\in \sigma_l$ such that $d_F(\tau[v_{a_k},y], \sigma[p,w_{b_{l+1}}])\le \delta$. Note that $\tau[v_{a_k},y]$ is a prefix of $\tau_k$. We identify the longest prefix of $\tau_k$ such that using Lemma~\ref{lem: simp} on that prefix returns a curve $\pre$ of at most $2\mu_2$ vertices. Since $y$ must belong to this prefix, we proceed to use $\pre$ as a surrogate to invoke {\sc WaveFront}$(\pre, \sigma_l, (2+\epsilon)\delta, \+S, \appReachVIArray{a_k}{l})$ to construct $\+I_3$ in $O(\mu_2^2)=o(\mu_1\mu_2)$ time. The array $\+S$ is induced by $\pre$, and all elements in it are empty. 
	
	
	
	\vspace{2pt}
	
	\noindent\textbf{Type 4.}  We construct an array $\+I^4$ such that $\+I^4$ covers all points of type 4, and every point covered by $\+I^4$ forms a $((7+\varepsilon)\delta)$-reachable pair with $w_{b_{l+1}}$. 
    It is unlikely that the use of curve simplification can accelerate the handling of points of type 4. Because these points may be scattered over $\tau_k$, and the entire $\tau_k$ may not be simplified significantly.
	We achieve the speedup via several data structures. The key data structure is the one that allows us to answer the following query fast.

    \vspace{2pt}
        
        \noindent\pmb{{\sc Cover}$(\tau', \delta',\+S)$.} Given a subcurve $\tau'$ of $\tau_k$, an array $\+S$ induced by $\tau_k$, and $\delta'>0$. For any $\epsilon\in(0,1)$, the query output is an array $\bar{\+S}$ induced by $\tau_k$ that satisfies the following properties:
	
	\begin{itemize}
		\item for any point $y\in\tau_k$, if there is a point $x$ covered by $\+S$ such that $x\le_{\tau} y$ and $d_F(\tau[x,y], \tau')\le \delta'$, then $y$ is covered by $\bar{\+S}$;
		\item for any point $y$ covered by $\bar{\+S}$, there is a point $x$ covered by $\+S$ such that $x\le_{\tau} y$ and $d_F(\tau[x,y], \tau')\le (1+\epsilon)\delta'$.
	\end{itemize}

    We present in Lemma~\ref{lem:cover} how to preprocess $\tau_k$ to answer a query {\sc Cover} in $O(\mu_1)$ time. To use it, note that if there is a $y\in \tau_k$ of type 4, there is a subcurve $\tau''=\tau[x,y]$ of $\tau_k$ such that $d_F(\tau',\sigma_l)\le \delta$. 
    Assume that we can find $\tau''$ efficiently. We use $\tau''$ as a surrogate of $\sigma_l$ and execute a query {\sc Cover}$(\tau'', 2\delta, \appReachWJArray{b_l}{k})$. We then assign the output to $\+I^4$. 
    For any point $y'$ of type 4, there is a point $x'$ covered by $\appReachWJArray{b_l}{k}$ such that $d_F(\tau[x', y'],\sigma_l)\le \delta$ by definition. By the triangle inequality, $d_F(\tau[x', y'], \tau'')\le 2\delta$. Hence, $y'$ is covered by $\+I^4$ by the definition of {\sc Cover}. 
	
    The bottleneck is how to find $\tau''$ for $\sigma_l$ efficiently. The nearest neighbor data structures~\cite{BDNP2022,CH2023,cheng2023solving,mirzanezhad2020approximate} under the Fr\'echet distance may help. However, all these data structures have exponential dependency on the size of the query curve, i.e., $|\sigma_l|=\mu_2$, in either the space complexity or the query time. We cannot afford to use them.
	
	We use sampling.
	Take a curve $\zeta$ with $|\zeta|\le \mu_1$. An edge $v_iv_{i+1}$ of $\tau_k$ is defined to be \emph{marked} by $\zeta$ if there is some subcurve $\tau'$ of $\tau_k$ such that $\tau'\cap v_iv_{i+1}\not=\emptyset$ and $d_F(\tau',\zeta)\le\delta$. Intuitively, it is easier to find a subcurve of $\tau_k$ close to $\zeta$ given an edge marked by $\zeta$, because we can focus on the subcurves around this edge. We present in Lemma~\ref{lem: marked-edge} how to find a subcurve $\tau''$ of $\tau_k$ such that $d_F(\tau'',\zeta)\le(3+2\epsilon)\delta$ in $O(|\zeta|^4)$ time given an edge marked by $\zeta$. For small enough $|\zeta|$, $O(|\zeta|^4)=o(\mu_1)$.
	
	
    A curve $\zeta$ is called \emph{$\omega$-dense} if it marks at least $\omega$ edges of $\tau_k$. By choosing a suitable value for $\omega$, for any $\omega$-dense $\zeta$, we can sample a small set of $\tau_k$'s edges to include at least one edge marked by $\zeta$ with high probability. We employ the observation in the following way.
	
	Take an integer $\mu_3<\mu_2$. We divide $\sigma_l$ into a collection $(\sigma_{l,1}, \sigma_{l,2},\ldots,\sigma_{l, (\mu_2-1)/\mu_3})$ of subcurves of $\mu_3$ edges such that $\sigma_{l,r}=\sigma[w_{b_{l,r}}, w_{b_{l,r+1}}]$, where $b_{l,r}=b_l+(r-1)\mu_3+1$ for $r\in [(\mu_2-1)/\mu_3]$. 
	Suppose that all $\sigma_{l,r}$'s are $\omega$-dense. For every $\sigma_{l,r}$, we can find a subcurve $\tau'_{l,r}$ of $\tau_k$ with $d_F(\tau'_{l,r},\sigma_{l,r})\le (3+2\epsilon)\delta$ with high probability via sampling. We then use the sequence $(\tau'_{l,1},\tau'_{l,2},\ldots,\tau'_{l,(\mu_2-1)/\mu_3})$ as a surrogate of $\sigma_l$. We do not require that the subcurve sequence follows the ordering along $\tau_k$. We use it as follows. We first execute a query {\sc Cover}$(\tau'_{l,1}, (4+2\epsilon)\delta, \epsilon, \appReachWJArray{b_l}{k})$ to get an array $\+S^1$. For any $r\ge 2$, suppose that we have gotten $\+S^{r-1}$, we proceed to execute a query {\sc Cover}$(\tau'_{l,r}, (4+2\epsilon)\delta, \+S^{r-1})$ to get $\+S^r$. In the end, we set $\+I^4=\+S^{(\mu_2-1)/\mu_3}$. It takes $O(\mu_1\mu_2/\mu_3)$ time.
 
    If some $\sigma_{l,r}$ is not $\omega$-dense, we find all edges marked by it by calling {\sc WaveFront}$(\tau_k, \sigma_{l,r},\delta,\+S, \+S')$, where $\+S$ is an array induced by $\tau_k$ with $\+S_{v_i}=v_iv_{i+1}\cap\+B(w_{b_l,r},\delta)$ and $\+S'$ is an array induced by $\sigma_{l,r}$ with all elements being empty. It takes $O(\mu_1\mu_3)$ time. There must be an edge marked by both $\sigma_{l,r}$ and $\sigma_l$ as $\sigma_l$ contains $\sigma_{l,r}$. Hence, we can use Lemma~\ref{lem: marked-edge} with every edge marked by $\sigma_{l,r}$ to find a subcurve $\tau'$ of $\tau_k$ with $d_F(\tau',\sigma_l)\le(3+2\epsilon)\delta$ in $O(\omega\mu_2^4)$ time. We execute {\sc Cover}$(\tau', (4+2\epsilon)\delta, \appReachWJArray{b_l}{k})$ and assign the output to $\+I^4$ in $O(\mu_1)$ time.
    We will use the following Chernoff bound. 
	
	\begin{lemma}\label{lem:Chernoff}
		Let $X=\sum_{i=1}^{k} X_i$, where $X_i=1$ with probability $p_i$, and $X_i=0$ with probability $1-p_i$, and all $X_i$'s are independent. Let $\beta=E(X)$. Then $\mathbb{P}(X\le (1-\alpha)\beta)\le e^{-\beta\alpha^2/2}$ for all $\alpha\in (0,1)$. 
	\end{lemma}
	
	\section{Subquadratic decision algorithm for Fr\'echet distance}
	
	\subsection{Preprocessing}\label{sec:preprocessing}
	Fix $\tau_k$ and $\varepsilon\in(0,1)$. We preprocess $\tau_k$ to construct $\zeta_k$, $\pre$, $\suf$, and several data structures. 
	
	\vspace{2pt}
	
	\noindent\textbf{Preprocessing for $\+I^1$, $\+I^2$ and $\+I^3$.} We generate $\zeta_k$, $\pre$ and $\suf$ first. Set $\epsilon=\varepsilon/10$ in Lemma~\ref{lem: simp}. We first call the algorithm in Lemma~\ref{lem: simp} with $\tau_k$ and $\delta$. If the algorithm returns a curve $\zeta_k$ of at most $2\mu_2$ vertices with $d_F(\tau_k, \zeta_k)\le (1+\epsilon)\delta$, we are done; otherwise, set $\zeta_k$ to be null. Let $\alpha=2(d-1)\lfloor d/2\rfloor^2+d$. It takes $O\left(\epsilon^{-\alpha}\mu_1\log(1/\varepsilon)\right)$ time. To generate $\pre$, we try to invoke the curve simplification algorithm with $\tau[v_{a_k}, i]$ and $\delta$ for all $i\in [a_k+1, a_{k+1}]$. We then identify the maximum $i$ such that the algorithm returns a curve of at most $2\mu_2$ vertices and set $\pre$ to be the corresponding output. It takes $O\left(\epsilon^{-\alpha}\mu_1^2\log(1/\varepsilon)\right)$ time as there are $O(\mu_1)$ prefixes to try. We can generate $\suf$ by trying all suffixes of $\tau_k$ in $O\left(\epsilon^{-\alpha}\mu_1^2\log(1/\varepsilon)\right)$ time in the same way. Let $\tau[v_{a_k}, v_{i_{\text{pre}}}]$ and $\tau[v_{i_{\text{suf}}}, v_{a_{k+1}}]$ be the corresponding prefix and suffix of $\tau_k$ for $\pre$ and $\suf$, respectively.
	
	We also need to compute and store a matching $\Mpre$ between $\tau[v_{a_k}, v_{i_{\text{pre}}}]$ and $\pre$ and a matching $\Msuf$ between $\tau[v_{i_{\text{suf}}}, v_{a_{k+1}}]$ and $\suf$ such that $\Mpre$ and $\Msuf$ realize distances at most $(1+\epsilon)\delta$. Given that $d_F(\tau[v_{a_k}, v_{i_{\text{pre}}}], \pre)\le (1+\epsilon)\delta$ and $d_F(\tau[v_{i_{\text{suf}}}, v_{a_{k+1}}], \suf)\le (1+\epsilon)\delta$ by Lemma~\ref{lem: simp}, we use the algorithm in Lemma~\ref{lem:matching} which runs in $O(\mu_1\mu_2)$ time.

	\vspace{2pt}
	
	\noindent\textbf{Preprocessing for $\+I^4$.} We first preprocess $\tau_k$ so that for any $l\in [(m-1)/\mu_2]$ and any subcurve $\sigma'=\sigma[w_j, w_{j_1}]$ of $\sigma_l$, given an edge $v_iv_{i+1}$ of $\tau_k$ marked by $\sigma'$, we can find a subcurve of $\tau_k$ that is close to $\sigma'$ efficiently. Recall that an edge $v_iv_{i+1}$ of $\tau_k$ being marked by $\sigma'$ means that there is a subcurve $\tau'$ of $\tau_k$ such that $\tau'\cap v_iv_{i+1}\not=\emptyset$ and $d_F(\tau',\sigma')\le \delta$.
	
	For every vertex $v_i$ of $\tau_k$, we identify the longest suffix of $\tau[v_{a_k}, v_i]$ and the longest prefix of $\tau[v_i, v_{a_{k+1}}]$ such that running the algorithm in Lemma~\ref{lem: simp} on them returns a curve of at most $2\mu_2$ vertices with respect to $\delta$, respectively. It takes $O\left(\epsilon^{-\alpha}\mu_1^2\log(1/\varepsilon)\right)$ time as there are $O(\mu_1)$ prefixes and suffixes to try. Let $\tau[\bar{v}_i, v_i]$ and $\tau[v_i, \tilde{v}_i]$ be the corresponding suffix and prefix, respectively. Let $\bar{\zeta}_i$ and $\tilde{\zeta}_i$ be the simplified curves for $\tau[\bar{v}_i, v_i]$ and $\tau[v_i, \tilde{v}_i]$, respectively. 
	By Lemma~\ref{lem:matching}, we further construct a matching $\bar{\+M}_i$ between $\tau[\bar{v}_i, v_i]$ and $\bar{\zeta}_i$ and a matching $\tilde{\+M}_i$ between $\tau[v_i, \tilde{v}_i]$ and $\tilde{\zeta}_i$ in $O(\mu_1\mu_2)$ time. It takes $O\left(\epsilon^{-\alpha}\mu_1^3\log(1/\varepsilon)\right)$ time to process all $v_i$'s.
	
	Next, we present how to find a subcurve of $\tau_k$ that is close to $\sigma'$ based on the above preprocessing. Suppose that $v_iv_{i+1}$ is marked by $\sigma'$. There is a subcurve $\tau[x,y]$ of $\tau_k$ such that $d_F(\tau[x,y], \sigma')\le \delta$ and $\tau[x,y]\cap v_iv_{i+1}\not=\emptyset$. It holds that $\bar{v}_i\le_\tau x$. Otherwise, there is a vertex $v_{i'}\le_\tau \bar{v}_i$ such that calling the algorithm in Lemma~\ref{lem: simp} on $\tau[v_{i'}, v_i]$ returns a curve of at most $2\mu_2$ vertices, which is a contradiction. We have $y\le_\tau \tilde{v}_{i+1}$ via the similar analysis. Hence, $\tau[x,y]$ is a subcurve of $\tau[\bar{v}_i, \tilde{v}_{i+1}]$. 
	
	Intuitively, we can concatenate $\bar{\zeta}_i$ and $\tilde{\zeta}_{i+1}$ to get a new curve $\zeta'$ of at most $4\mu_2$ vertices as the surrogate of $\tau[\bar{v}_i, \tilde{v}_{i+1}]$. It is sufficient to find a subcurve of $\zeta'$ that is close to $\sigma'$, and then retrieve a subcurve of $\tau[\bar{v}_i, \tilde{v}_{i+1}]$ that is close to $\sigma'$ quickly using $\bar{\+M}_i$ and $\tilde{\+M}_{i+1}$. We formalize the intuition in the following lemma. 
	
	\begin{lemma}\label{lem: marked-edge}
		We can preprocess $\tau_k$ in $O\left(\epsilon^{-\alpha}\mu_1^3\log(1/\varepsilon)\right)$ time such that given any subcurve $\sigma'$ of $\sigma_l$ and an edge $v_iv_{i+1}$ of $\tau_k$, there is an $O(\mu_2^4)$-time algorithm that returns null or a subcurve $\tau'$ of $\tau_k$ with $d_F(\tau', \sigma')\le (3+2\epsilon)\delta$. If the algorithms returns null, $v_iv_{i+1}$ is not marked by $\sigma'$.
	\end{lemma} 
	
	\begin{proof}
			The preprocessing time follows the construction of $\bar{\zeta}_i$, $\tilde{\zeta}_i$,$\bar{\+M}_i$, and $\tilde{\+M}_i$ directly. We proceed to show how to find $\tau'$.
			
		Suppose that we are given a subcurve $\sigma'=\sigma[w_j, w_{j_1}]$ of $\sigma_l$ and an edge $v_iv_{i+1}$. If there is subcurve of $\tau[\bar{v}_i,\tilde{v}_{i+1}]$ within a Fr\'echet distance $\delta$ to $\sigma'$, we present how to find a subcurve $\tau'$ of $\tau[\bar{v}_i,\tilde{v}_{i+1}]$ with $d_F(\tau', \delta')\le (3+2\epsilon)\delta$ based on $\bar{\zeta}_i$, $\tilde{\zeta}_{i+1}$, $\bar{\+M}_i$, and $\tilde{\+M}_{i+1}$. We first construct a new curve $\zeta'$ by appending $\tilde{\zeta}_{i+1}$ to $\bar{\zeta}_i$. That is, we join the last vertex of $\bar{\zeta}_i$ and the first vertex of $\tilde{\zeta}_{i+1}$ by a line segment to generate $\zeta'$. Since the new line segment is within a Fr\'echet distance $(1+\epsilon)\delta$ to the edge $v_iv_{i+1}$, we have $d_F(\tau[\bar{v}_i, \tilde{v}_{i+1}], \zeta')\le (1+\epsilon)\delta$. Given that there is a subcurve of $\tau[\bar{v}_i, \tilde{v}_{i+1}]$ within a Fr\'echet distance $\delta$ to  $\sigma'$, there is a subcurve $\zeta''$ of $\zeta'$ with $d_F(\zeta'', \sigma')\le (2+\epsilon)\delta$ by the triangle inequality. 
		
		We aim to find such a $\zeta''$. For every edge of $\zeta'$, if it intersects the ball $\+B(w_j, (2+\epsilon)\delta)$, we take the minimum point $x$ in the intersection with respect to $\le_{\zeta'}$. We insert $x$ into a set $X$. If this edge intersects the ball $\+B(w_{j_1}, (2+\epsilon)\delta)$, we take the maximum point $y$ in the intersection with respect to $\le_{\zeta'}$, and insert $y$ into another set $Y$. 
		
		The existence of $\zeta''$ implies the existence of some subcurve that starts from a point in $X$, ends at a point in $Y$, and locates within a Fr\'echet distance $(2+\epsilon)\delta$ to $\sigma'$. Let $x$ be the point in $X$ such that $x$ and the start of $\zeta''$ are on the same edge of $\zeta'$. Let $y$ be the point in $Y$ such that $y$ and the end of $\zeta''$ are on the same edge of $\zeta'$. By definition of $X$ and $Y$, $\zeta'[x,y]$ includes $\zeta''$. We can extend the Fr\'echet matching between $\zeta''$ and $\sigma'$ to a matching between $\zeta'[x,y]$ and $\sigma'$ by matching the line segment between $x$ and the start of $\zeta''$ to $w_j$, and matching the line segment between the end of $\zeta''$ and $y$ to $w_{j_1}$. The matching realizes a distance at most $(2+\epsilon)\delta$.
	
	We test all subcurves of $\zeta'$ starting from some point in $X$ and ending at some point in $Y$. There are $O(|\zeta'|^2)=O(\mu_2^2)$ subcurves to be tested. For every subcurve, we check whether the Fr\'echet distance between it and $\sigma'$ is at most $(2+\epsilon)\delta$. If so, we return it as $\zeta''$. If there is not a satisfactory $\zeta''$ after trying all $O(\mu_2^2)$ subcurves, it means that there is no subcurve of $\zeta'$ within a Fr\'echet distance $(2+\epsilon)\delta$ to $\sigma'$. It implies that no subcurve of $\tau[\bar{v}_i, \tilde{v}_{i+1}]$ is within a Fr\'echet distance $\delta$ to $\sigma'$. Hence, the edge $v_iv_{i+1}$ is not marked by $\sigma'$. We return null. It takes $O(\mu_2^4)$ time. 
	
	Suppose that we have found $\zeta''=\zeta'[x,y]$. We proceed to find a subcurve $\tau'$ of $\tau_k$ within a Fr\'echet distance $(1+\epsilon)\delta$ to $\zeta''$ based on $\bar{\+M}_i$ and $\tilde{\+M}_{i+1}$. Recall that $\zeta'$ is a concatenation of $\bar{\zeta}_i$ and $\tilde{\zeta}_{i+1}$. It means that $x$ must locate in $\bar{\zeta}_i$, or $\tilde{\zeta}_{i+1}$, or the new line segment that joins the last vertex of $\bar{\zeta}_i$ and the first vertex of $\tilde{\zeta}_{i+1}$. So does $y$. We define a point $x'$ in $\tau_k$ that corresponds to $x$ as follows. We set $x'=\bar{\+M}_i(x)$ if $x\in \bar{\zeta}_i$, set $x'=\tilde{\+M}_{i+1}(x)$ if $x\in \tilde{\zeta}_{i+1}$, and set $x'$ to be the point matched to $x$ by the Fr\'echet matching between $v_iv_{i+1}$ and the new line segment otherwise. We can define $y'$ for $y$ similarly. It is clear from the construction that $d_F(\tau[x',y'], \zeta'')\le (1+\varepsilon)\delta$. Hence, $d_F(\tau[x',y'], \sigma')\le (3+2\epsilon)\delta$ by the triangle inequality. We set $\tau'=\tau[x', y']$. It takes an extra time of $O(\log\mu_1)=O(\log m)$. We complete the proof.
	
\end{proof}

	Next, we preprocess $\tau_k$ to answer the query {\sc Cover}$(\tau',\delta',\+S)$. For any subcurve $\tau'=\tau[x, y]$ of $\tau_k$, if $\tau'$ has at most 2 edges, we can take an array $\+S'$ induced by $\tau'$ with all elements being empty, and call {\sc WaveFront}$(\tau_k, \tau', \delta', \+S, \+S')$ to answer the query {\sc Cover}$(\tau', \delta', \+S)$ in $O(\mu_1)$ time. Specifically, we use the output array $\+{W}^y$ returned by the invocation of {\sc WaveFront} as the answer. 
	
	When $\tau'$ has more than 2 edges, it must contain at least 2 vertices of $\tau_k$. That is, $\tau[x,y]=(x, v_i, v_{i+1},\ldots, v_{i_1}, y)$. It is natural to process $xv_i$, the vertex-to-vertex subcurve $\tau[v_i, v_{i_1}]$ and $v_{i_1}y$ progressively. Since both $xv_i$ and $v_{i_1}y$ are line segments, the bottleneck is the handling of $\tau[v_i, v_{i_1}]$. Given that $\tau_k$ has $O(\mu_1^2)$ vertex-to-vertex subcurves, we can afford to precompute all subcurves that are close to $\tau[v_i, v_{i_1}]$ to speed up the query process. The details are given in Appendix~\ref{sec:cover}.

	\begin{lemma}\label{lem:cover}
		Fix $\tau_k$ and $\delta'$. For any $\epsilon\in (0,1)$, there is a data structure of size $O(\mu_1^4/\epsilon)$ and preprocessing time $O(\mu_1^5/\epsilon)$ that answers in $O(\mu_1)$ time the query {\sc Cover} for any $\tau'$ and $\+S$. 
	\end{lemma}
	
	
	\subsection{Implementation of {\sc Reach}}
	Given $\tau_k$, $\sigma_l$, an array $\appReachVIArray{a_k}{l}$ induced by $\sigma_l$ that is $(7+\varepsilon)$-approximate reachable for $v_{a_k}$ and an array $\appReachWJArray{b_l}{k}$ induced by $\tau_k$ that is $(7+\varepsilon)$-approximate reachable for $w_{b_l}$, we implement {\sc Reach} to compute an array $\appReachVIArray{a_{k+1}}{l}$ induced by $\sigma_l$ that is $(7+\varepsilon)$-approximate reachable for $v_{a_{k+1}}$ and an array $\appReachWJArray{b_{l+1}}{k}$ induced by $\tau_k$ that is $(7+\varepsilon)$-approximate reachable for $w_{b_{l+1}}$.
	As discussed in the overview, we first construct $\+I^1$, $\+I^2$, $\+I^3$ and $\+I^4$ to cover points of types $1-4$, respectively.
	
	\vspace{2pt}
	
	
	\noindent\textbf{Construction of $\+I^1$.} 
	If $\zeta_k$ is not null, let $\+S^k$ be an array induced by $\zeta_k$ such that all its elements are empty. We invoke {\sc WaveFront}$(\zeta_k, \sigma_l, (2+\epsilon)\delta, \+S^k,\appReachVIArray{a_k}{l})$. Let $\+{W}^{v'_a}$ be the output array for the last vertex $v'_a$ of $\zeta_k$.
	Set $\+I^1=\+{W}^{v'_a}$. If $\zeta_k$ is null, set all elements in $\+I^1$ to be empty. It runs in $O(\mu_2^2)$ time. The array $\+I^1$ satisfies the following property. 
	
	
	\begin{lemma}~\label{lem: I_1}
		We can construct $\+I^1$ in $O(\mu_2^2)$ time such that $\+I^1$ covers all points of type 1, and every point covered by $\+I^1$ forms a $((7+\varepsilon)\delta)$-reachable pair with $v_{a_{k+1}}$. 
	\end{lemma}

	\begin{proof}
		The running time follows the construction procedure. We focus on proving $\+I^1$'s property.
	
		If $\zeta_k$ is null, it means that $\tau_k$ is at a Fr\'echet distance more than $\delta$ to all curves of at most $\mu_2$ vertices, which implies that points of type 1 do not exist. According to the construction procedure, all elements in $\+I^1$ are set to be empty. There is nothing to prove.
	
		When $\zeta_k$ is not null, suppose that there is a point $q$ of type 1. By definition, $q\in \sigma_l$, and there is a point $p\in \sigma_l$ covered by $\appReachVIArray{a_k}{l}$ such that $d_F(\tau_k, \sigma[p, q])\le \delta$. 
		Given that $d_F(\tau_k, \zeta_k)\le (1+\epsilon)\delta$,  $d_F(\zeta_k, \sigma[p, q])\le(2+\epsilon)\delta$ by the triangle inequality. Recall that $\+I^1$ equals to $\+{W}^{v'_a}$, which is the output array for the last vertex of $\zeta_k$ returned by {\sc WaveFront}$(\zeta_k,\sigma_l, (2+\epsilon)\delta, \+S^k, \appReachVIArray{a_k}{l})$, where all elements in $\+S^k$ are empty. Hence, $q$ is covered by $\+I^1$ according to the definition of {\sc WaveFront}.
		
	
		Next, by the definition of {\sc WaveFront}, for any point $q'$ covered by $\+I^1$, there is a point $p'$ covered by $\appReachVIArray{a_k}{l}$ such that $p'\le_{\sigma_l} q'$ and $d_F(\zeta_k, \sigma[p',q'])\le (2+\epsilon)\delta$. By the triangle inequality, $d_F(\tau_k, \sigma[p',q'])\le (3+2\epsilon)\delta$ as $d_F(\tau_k, \zeta_k)\le (1+\epsilon)\delta$. Recall that $\epsilon=\varepsilon/10$. Provided that $\appReachVIArray{a_k}{l}$ is $(7+\varepsilon)$-approximate reachable for $v_{a_k}$, it holds that $d_F(\tau[v_1, v_{a_k}], \sigma[w_1, p'])\le(7+\varepsilon)\delta$. It implies that we can construct a matching between $\tau[v_1, v_{a_{k+1}}]$ and $\sigma[w_1, q']$ by concatenating the Fr\'echet matching between $\tau_k$ and $\sigma[p',q']$ to the Fr\'echet matching between $\tau[v_1, v_{a_k}]$ and $\sigma[w_1, p']$. The matching realizes a distance at most $(7+\varepsilon)\delta$. Hence, all points in $\+I^1_{w_j}$ can form $((7+\varepsilon)\delta)$-reachable pairs with $v_{a_{k+1}}$.
	\end{proof}

	\noindent\textbf{Construction of $\+I^2$.} 
	We invoke {\sc WaveFront}$(\suf, \sigma_l, (2+\epsilon)\delta, \+S, \+S')$ to construct $\+I^2$. As discussed in the overview, $\+S'$ is an array induced by $\sigma_l$ with all elements being empty. The array $\+S$ is induced by $\suf$ such that for any point $x\in \tau[v_{i_{\text{suf}}},v_{a_{k+1}}]$ covered by $\appReachWJArray{b_l}{k}$, $\Msuf(x)$ is covered by $\+S$. We present how to construct $\+S$ based on $\appReachWJArray{b_l}{k}$ and $\Msuf$. 
	For every edge $\bar{v}_e\bar{v}_{e+1}$ of $\suf$, we find the first point $x'$ and the last point $y'$ in $\bar{v}_e\bar{v}_{e+1}$ such that $\Msuf(x')$ and $\Msuf(y')$ are covered by $\appReachWJArray{b_l}{k}$. We then set $\+S_{\bar{v}_e}$ to be $x'y'$. 
	We set $\+S_{\bar{v}_e}$ to be empty if $x'$ and $y'$ do not exist.
	
	We present how to find $x'$ and $y'$. Let $\ell_i$ and $r_i$ be the start and end of $\appReachWJ{b_l}{k}{i}$, respectively. We set $\ell_i$ and $r_i$ to be null if $\appReachWJ{b_l}{k}{i}$ is empty. For ease of presentation, we assume that $\ell_i$ and $r_i$ are non-empty for all $i\in [i_{\text{suf}}, a_{k+1}-1]$ as we can exclude all elements of null easily. By Lemma~\ref{lem:matching}, we can get $\Msuf(\ell_i)$ and $\Msuf(r_i)$ for all $i\in [i_{\text{suf}}, a_{k+1}-1]$ in $O(\mu_1\log\mu_2)$ time. Since \cancel{$\ell_{i_{\text{suf}}}\le_{\tau}r_{i_{\text{suf}}}\le_{\tau}\ell_{i_{\text{suf}}+1}\le_{\tau}r_{i_{\text{suf}}+1}\le_{\tau}\ldots\le_{\tau}\ell_{a_{k+1}-1}\le_{\tau}r_{a_{k+1}-1}$}$\ell_{i_{\text{suf}}},r_{i_{\text{suf}}},\ell_{i_{\text{suf}}+1},\ldots,\ell_{a_{k+1}-1},r_{a_{k+1}-1}$ are sorted along $\tau$, the points matched to them by $\Msuf$ are sorted along $\suf$ as well. 
     Note that for every edge $\bar{v}_e\bar{v}_{e+1}$, $x'$ and $y'$ belong to either $\{\bar{v}_e, \bar{v}_{e+1}\}$ or the points matched to some $\ell_{i}$ or $r_{i}$ by $\Msuf$. We first check whether $\Msuf(\bar{v}_e)$ is covered by $\appReachWJArray{b_l}{k}$. If so, we set $x'$ to be $\bar{v}_e$; otherwise, we traverse the sorted point sequence on $\suf$ 
	to identify the first point on $\bar{v}_e\bar{v}_{e+1}$ that is matched to some $\ell_i$ by $\Msuf$ and assign its value to $x'$. We can determine $y'$ by checking $\Msuf(\bar{v}_{e+1})$ and traversing the sorted sequence similarly. This finishes how to construct $\+S$ in $O(\mu_1\log\mu_2)$ time.

    We finally invoke {\sc WaveFront}$(\suf, \sigma_l, (2+\epsilon)\delta, \+S, \+S')$. Let $\+{W}^{\bar{v}_a}$ be the output array for the last vertex $\bar{v}_a$ of $\suf$. Set $\+I^2=\+{W}^{\bar{v}_a}$. 
	
	\begin{lemma}~\label{lem: I_2}
		We can construct $\+I^2$ in $O(\mu_1\log \mu_2+\mu_2^2)$ time such that $\+I^2$ covers all points of type 2, and every point covered by $\+I^2$ forms a $((7+\varepsilon)\delta)$-reachable pair with $v_{a_{k+1}}$. 
	\end{lemma}

	\begin{proof}
		The running time follows the construction procedure. We focus on proving $\+I^2$'s property.
		
		Suppose that there is a point $q$ of type 2. By definition, $q\in \sigma_l$, and there is a point $x$ covered by $\appReachWJArray{b_l}{k}$ with $d_F(\tau[x, v_{a_{k+1}}], \sigma[w_{b_l}, q])\le \delta$. 
		Let $\bar{v}_a$ be the last vertex of $\suf$. Note that $x\in \tau[v_{i_{\text{suf}}}, v_{a_{k+1}}]$. Hence, $\suf$ contains $\Msuf(x)$. Let $\bar{x}$ denote $\Msuf(x)$ for ease of notation. Given that the matching $\Msuf$ realizes a distance at most $(1+\epsilon)\delta$, we have $d_F(\tau[x, v_{a_{k+1}}], \suf[\bar{x}, \bar{v}_a])\le (1+\epsilon)\delta$. By the triangle inequality, $d_F(\suf[\bar{x}, \bar{v}_a], \sigma[w_{b_l}, q])\le(2+\epsilon)\delta$. 
		
		Suppose that $\Msuf(x)$ is on the edge $\bar{v}_e\bar{v}_{e+1}$ of $\suf$. Since $x$ is covered by $\appReachWJArray{b_l}{k}$, $\bar{x}$ must be covered by $\+S_{\bar{v}_e}=x'y'$ by the construction of $\+S$. 
	   Given that $\+I^2$ equals to the output array $\+{W}^{\bar{v}_a}$ returned by {\sc WaveFront}$(\suf,\sigma_l,(2+\epsilon)\delta,\+S,\+S')$, $q$ must be covered by $\+I^2$ by the definition of {\sc WaveFront}.
	
		Next, take any point $q'$ covered by $\+I^2$. Recall that $\+I^2$ equals to $\+{W}^{\bar{v}_a}$. By the definition of {\sc WaveFront}, there is a point $x''$ covered by $\+S$ such that the Fr\'echet distance between $\suf[x'', \bar{v}_a]$ and $\sigma[w_{b_l},q']$ is at most $(2+\epsilon)\delta$ as all elements in $\+S'$ are empty. Suppose that $x''$ is on some edge $\bar{v}_e\bar{v}_{e+1}$ of $\suf$. Take the start $x'$ of $\+S_{\bar{v}_e}$. We proceed to prove that $d_F(\suf[x', \bar{v}_a], \sigma[w_{b_l}, p])\le (2+\epsilon)\delta$. It is sufficient to prove that the entire line segment $x'x''$ is inside the ball $\+B(w_l, (2+\epsilon)\delta)$. Because we can generate a matching between $\suf[x', \bar{v}_a]$ and $\sigma[w_{b_l},p]$ by matching the line segment $x'x''$ to $w_{b_l}$ and matching $\suf[x'', \bar{v}_a]$ to $\sigma[w_{b_l}, q']$ by their Fr\'echet matching. The matching realizes a distance of at most $(2+\epsilon)\delta$. Note that $d(w_{b_l}, x'')\le (2+\epsilon)\delta$. Suppose that $x'=\Msuf(y)$ for some $y$ covered by $\appReachWJArray{b_l}{k}$. We have $d(y, x')\le (1+\epsilon)\delta$ and $d(y, w_{b_l})\le \delta$ as $\appReachWJArray{b_l}{k}$ is $(7+\varepsilon)$-approximate reachable for $w_{b_l}$. By the triangle inequality, $d(x', w_{b_l})\le(2+\epsilon)\delta$. Hence, the entire line segment $x'x''$ is inside the ball $\+B(w_{b_l}, (2+\epsilon)\delta)$, and $d_F(\suf[x', \bar{v}_a], \sigma[w_{b_l},q'])\le (2+\epsilon)\delta$.
	
		By the triangle inequality, $d_F(\tau[y, v_{a_{k+1}}], \sigma[w_{b_l},q'])\le (3+2\epsilon)\delta$ as $d_F(\tau[y, v_{a_{k+1}}], \suf[x', \bar{v}_a])\le (1+\epsilon)\delta$. Recall that $\epsilon=\varepsilon/10$.
		Provided that $\appReachWJArray{b_l}{k}$ is $(7+\varepsilon)$-approximate reachable for $w_{b_l}$, $d_F(\tau[v_1, y], \sigma[w_1, w_{b_l}])\le(7+\varepsilon)\delta$. It implies that we can construct a matching between $\tau[v_1, v_{a_{k+1}}]$ and $\sigma[w_1, q']$ by concatenating the Fr\'echet matching between $\tau[y, v_{a_{k+1}}]$ and $\sigma[w_{b_l},q']$ to the Fr\'echet matching between $\tau[v_1, y]$ and $\sigma[w_1, w_{b_l}]$. The matching realizes a distance at most $(7+\varepsilon)\delta$. Hence, all points in $\+I^2_{w_j}$ can form $((7+\varepsilon)\delta)$-reachable pairs with $v_{a_{k+1}}$.
	
		\end{proof}

	\noindent\textbf{Construction of $\+I^3$.} 
	 We invoke {\sc WaveFront}$(\pre, \sigma_l, (2+\epsilon)\delta, \+S, \appReachVIArray{a_k}{l})$, where $\+S$ is an array induced by $\pre$ with all elements being empty. Let $\+{W}^{w_{b_{l+1}}}$ be the output array for $w_{b_{l+1}}$. 

     Suppose that $\pre=(\tilde{v}_1,\tilde{v}_2,\ldots,\tilde{v}_b)$. Note that $\+{W}^{b_{l+1}}$ is induced by $\pre$. A point $\tilde{x}\in \pre$ is covered by $\+{W}^{b_{l+1}}$ if and only if there is a point $p$ covered by $\appReachVIArray{a_k}{l}$, and $d_F(\pre[\tilde{v}_1,\tilde{x}],\sigma[p,w_{b_{l+1}}])\le (2+\epsilon)\delta$ by the definition of {\sc WaveFront}. As discussed in the overview, for any point $y$ of type 3, there is a point $p$ covered by $\appReachVIArray{a_k}{l}$ such that $d_F(\tau[v_{a_k}, y], \sigma[p, w_{b_{l+1}}])\le \delta$. Note that the Fr\'echet distance between $\pre[\tilde{v}_1, \Mpre(y)]$ and $\sigma([p, w_{b_{l+1}}])$ is at most $(2+\epsilon)\delta$ by the triangle inequality. It implies that $\Mpre(y)$ is covered by $\+{W}^{w_{b_{l+1}}}$. 
	 
	 We compute $\+I^3$ such that all points covered by $\+{W}^{w_{b_{l+1}}}$ are matched to some points in $\+I^3$ by $\Mpre$.
    For $i\in [i_{\text{pre}}, a_{k+1}-1]$, we set $\+I^3_{v_i}$ to be empty as there cannot be any point of type 3 on $v_iv_{i+1}$. For every edge $v_iv_{i+1}$ of $\tau[v_{a_k}, v_{i_{\text{pre}}}]$, we find the first point $x'$ and last point $y'$ in $v_iv_{i+1}$ such that $\Mpre(x')$ and $\Mpre(y')$ are covered by $\+{W}^{w_{b_{l+1}}}$. We then set $\+I^3_{v_i}$ to be $x'y'$. If $x'$ and $y'$ do not exist, we set $\+I^3_{v_i}$ to be empty.
	
	We present how to find $x'$ and $y'$ for every edge $v_iv_{i+1}$. Let $s_e$ and $e_e$ be the start and end of $\+{W}^{w_{b_{l+1}}}_{\tilde{v}_e}$, respectively. We set $s_e$ and $e_e$ to be null if $\+{W}^{w_{b_{l+1}}}_{\tilde{v}_e}$ is empty. For ease of presentation, we assume that $s_e$ and $e_e$ are not null for all edges $v\tilde{v}_e\tilde{v}_{e+1}$ ofr $\pre$ as we can exclude all elements of null easily. By Lemma~\ref{lem:matching}, we can get $\Mpre(s_e)$ and $\Mpre(e_e)$ for all $b\in [b-1]$ in $O(\mu_2\log\mu_1)$ time. Since $s_1,e_1, s_2,\ldots, s_{b-1},e_{b-1}$ are sorted along $\pre$,\cancel{$s_1\le_{\pre} e_1\le_{\pre} s_2\le_{\pre}e_2\le_{\pre}\ldots\le_{\pre} s_{|\pre|-1}\le_{\pre}e_{|\pre|-1}$} the points matched to them by $\Mpre$ are sorted along $\tau[v_1, v_{i_{\text{pre}}}]$ as well.\cancel{$\Mpre(s_1)\le_{\tau}\Mpre(e_1)\le_{\tau}\Mpre(s_2)\le_{\tau}\Mpre(e_2)\le_{\tau}\ldots\le_{\tau}\Mpre(s_{|\pre|-1})\le_{\tau}\Mpre(e_{|\pre|-1})$} Note that $x'$ and $y'$ belong to either $\{v_i, v_{i+1}\}$ or the points matched to some $s_e$ or $e_e$ by $\Mpre$. We first check whether $\Mpre(v_i)$ is covered by $\+{W}^{w_{b_{l+1}}}$. If so, we set $x'=v_i$; otherwise, we traverse the sorted point sequence on $\tau[v_1, v_{i_\text{pre}}]$ to identify the first point on $v_iv_{i+1}$ that is matched to some $s_e$ by $\Mpre$ and assign its value to $x'$. We can determine $y'$ by checking $\Mpre(v_{i+1})$ and traversing the sorted sequence similarly. 
	The array $\+I^3$ satisfies the following property. 
	
	\begin{lemma}~\label{lem: I_3}
		We can construct $\+I^3$ in $O(\mu_1+\mu_2^2)$ time such that $\+I^3$ covers all points of type 3, and every point covered by $\+I^3$ forms a $((7+\varepsilon)\delta)$-reachable pair with $w_{b_{l+1}}$. 
	\end{lemma}
	
	\begin{proof}
		The running time follows the construction procedure. We focus on proving $\+I^3$'s property.
		
		Suppose that there is a point $y$ of type 3. By definition, 
		there is a point $p$ covered by $\appReachVIArray{a_k}{l}$ with $d_F(\tau[v_{a_k}, y], \sigma[p, w_{b_{l+1}}])\le \delta$. 
        Provided that $\sigma[p, w_{b_{l+1}}]$ has at most $\mu_2+1$ vertices, $y$ must belong to $\tau[v_{a_k}, v_{i_{\text{pre}}}]$.
		We first prove that $\Mpre(y)$ is covered by $\+{W}^{w_{b_{l+1}}}$. Recall that $\+{W}^{w_{b_{l+1}}}$ is the output array for $w_{b_{l+1}}$ returned by {\sc WaveFront}$(\pre,\sigma_l,(2+\epsilon)\delta, \+S, \appReachVIArray{a_k}{l})$. 
        Given that $\Mpre$ realizes a distance at most $(1+\epsilon)\delta$, it holds that $d_F(\tau[v_{a_k}, y], \pre[\tilde{v}_1, \Mpre(y)])\le(1+\epsilon)\delta$. Hence, $d_F(\pre[\tilde{v}_1, \Mpre(y)], \sigma[p, w_{b_{l+1}}])\le(2+\epsilon)\delta$ by the triangle inequality. Provided that $p$ is covered by $\appReachVIArray{a_k}{l}$, and $\+{W}^{w_{b_{l+1}}}$ is returned by the invocation of {\sc WaveFront}$(\pre, \sigma_l, (2+\epsilon)\delta, \+S, \appReachVIArray{a_k}{l})$, $\Mpre(y)$ must be covered by $\+{W}^{w_{b_{l+1}}}$. 
	
		Suppose that $y\in v_iv_{i+1}$. Given that $\Mpre(y)$ is covered by $\+{W}^{w_{b_{l+1}}}$, $\+I^3_{v_i}$ is not empty by construction. In addition, the start $x'$ and end $y'$ of $\+I^3_{v_i}$ satisfies that $x'\le_{\tau} x\le_{\tau} y'$. Hence, $y\in \+I^3_{v_i}$.
	
		Next, suppose that $\+I^3_{v_i}=x'y'$ is not empty. By construction, $i$ belongs to $[a_k, i_{\text{pre}}-1]$, and both $\Mpre(x')$ and $\Mpre(y')$ are covered by $\+{W}^{w_{b_{l+1}}}$. Hence, $\Mpre(x')$ satisfies that there is some point $p$ covered by $\appReachVIArray{a_k}{l}$ such that $d_F(\pre[\tilde{v}_1, \Mpre(x')], \sigma[p, w_{b_{l+1}}])\le (2+\epsilon)\delta$. So does $\Mpre(y')$. It implies that $d(w_{b_{l+1}}, \Mpre(x'))\le (2+\epsilon)\delta$ and $d(w_{b_{l+1}}, \Mpre(y'))\le (2+\epsilon)\delta$. Provided that $d(x', \Mpre(x'))\le (1+\epsilon)\delta$ and $d(y', \Mpre(y'))\le (1+\epsilon)\delta$, both $d(x', w_{b_{l+1}})$ and $d(y', w_{b_{l+1}})$ are at most $(3+2\epsilon)\delta$ by the triangle inequality. Hence, $\+I^3_{v_i}$ is inside the ball $\+B(w_{b_{l+1}}, (3+2\epsilon)\delta)$.
	
		We proceed to prove that $(x', w_{b_{l+1}})$ is $((7+\varepsilon)\delta)$-reachable. 
		Observe that by the triangle inequality,
		we have $d_F(\tau[v_{a_k}, x'], \sigma[p, w_{b_{l+1}}])\le (3+2\epsilon)\delta$. Recall that $\epsilon=\varepsilon/10$. Given that $\appReachVIArray{a_k}{l}$ is $(7+\varepsilon)$-approximate reachable for $v_{a_k}$, it holds that $d_F(\tau[v_1, v_{a_k}], \sigma[w_1, p])\le (7+\varepsilon)\delta$. Hence, we can construct a matching between $\tau[v_1, x']$ and $\sigma[w_1, w_{b_{l+1}}]$ by concatenating the Fr\'echet matching between $\tau[v_{a_k}, x']$ and $\sigma[p, w_{b_{l+1}}]$ and the Fr\'echet matching between $\tau[v_1, v_{a_k}]$ and $\sigma[w_1, p]$. The matching realizes a distance at most $(7+\varepsilon)\delta$. Therefore, $d_F(\tau[v_1, x'], \sigma[w_1, w_{b_{l+1}}])\le (7+\varepsilon)\delta$, it means that $(x', w_{b_{l+1}})$ is $((7+\varepsilon)\delta)$-reachable.
	
		For any point $y\in \+I^3_{v_i}$, we can extend the Fr\'echet matching between $\tau[v_1, x']$ and $\sigma[w_1, w_{b_{l+1}}]$ to a matching between $\tau[v_1, y]$ and $\sigma[w_1, w_{b_{l+1}}]$ by matching the line segment $x'y$ to $w_{b_{l+1}}$. The matching realizes a distance at most $(7+\varepsilon)\delta$. Therefore, $d_F(\tau[v_1, y], \sigma[w_1, w_{b_{l+1}}])\le (7+\varepsilon)\delta$, it means that $(y, w_{b_{l+1}})$ is $((7+\varepsilon)\delta)$-reachable.
	
	\end{proof}
	
	
	
	\noindent\textbf{Construction of $\+I^4$.} 
	As in the overview, we divide $\sigma_l$ into subcurves $(\sigma_{l,1}, \sigma_{l,2},\ldots,\sigma_{l, (\mu_2-1)/\mu_3})$ of $\mu_3$ edges. For every $\sigma_{l,r}$ and a constant $c>0$ whose value will be specified later, we sample a set of $2c\log n\cdot \mu_1/\omega$ edges of $\tau_k$ independently with replacement. We run the algorithm in Lemma~\ref{lem: marked-edge} on every $\sigma_{l,r}$ and every sampled edge. It takes $O(\mu_2^5\mu_1\cdot\log n/(\omega\mu_3))$ time. In the case that we succeed in finding $\tau'_{l,r}$ for every $\sigma_{l,r}$ via sampling, we get $(\tau'_{l,r})_{r\in [(\mu_2-1)/\mu_3]}$ as a surrogate of $\sigma_l$. 
	
	Otherwise, take some $\sigma_{l,r}$ for which we fail to get a $\tau'_{l,r}$. It implies that no sampled edge is marked by $\sigma_{l,r}$. By Lemma~\ref{lem:Chernoff}, $\sigma_{l,r}$ is not $\omega$-dense with probability at least $1-n^{-10}$ for sufficiently large $c$. Conditioned on that $\sigma_{l,r}$ is not $\omega$-dense. 
    Let $\+S$ be an array induced by $\tau_k$ such that $\+S_{v_i}=v_iv_{i+1}\cap \+B(w_{b_{l,r}}, \delta)$. Let $\+S'$ be an array induced by $\sigma_{l,r}$ with all elements being empty. We invoke {\sc WaveFront}$(\tau_k, \sigma_{l,r}, \delta, \+S, \+S')$ to get the output array $\+{W}^{w_{b_{l,r+1}}}$ for the vertex $w_{b_{l, r+1}}$ in $O(\mu_1\mu_3)$ time. Let $E$ contain all edges $v_iv_{i+1}$ such that $\+{W}^{w_{b_{l,r+1}}}_{v_i}\not=\emptyset$. Every edge in $E$ is marked by $\sigma_{l,r}$. Hence, $|E|<\omega$. If there is an edge of $\tau_k$ marked by $\sigma_l$, $E$ must contain one because $\sigma_l$ contains $\sigma_{l,r}$. We run the algorithm in Lemma~\ref{lem: marked-edge} on $\sigma_l$ and every edge in $E$ to find $\tau'$ with $d_F(\tau',\sigma_l)\le(3+2\epsilon)\delta$. It takes $O(\omega\mu_2^4)$ time.
	
	If there is a subcurve $\tau''$ of $\tau_k$ such that $d_F(\tau'',\sigma_l)\le\delta$, we find either $(\tau'_{l,r})_{r\in[(\mu_2-1)/\mu_3]}$ or $\tau'$ as a surrogate of $\sigma_l$ with probability at least $1-n^{-10}$. We then use the query {\sc Cover} as discussed in the overview to get $\+I^4$ in $O(\mu_1\mu_2/\mu_3)$ time. Otherwise, we set all elements in $\+I^4$ to be empty because points of type 4 do not exist.
	
	\begin{lemma}\label{lem: I_4}
		We can construct $\+I^4$ in $O(\mu_1(\mu_3+\mu_2/\mu_3+\mu_2^5\log n/(\omega\mu_3))+\omega\mu_2^4)$ time such that $\+I^4$ covers all points of type 4 with probability at least $1-n^{-10}$, and every point covered by $\+I^4$ forms a $((7+\varepsilon)\delta)$-reachable pair with $w_{b_{l+1}}$.
	\end{lemma}

	\begin{proof}
		The running time follows the construction procedure. We focus on proving $\+I^4$'s property.
		
		Suppose that there is a point $y$ of type 4. By definition, there is a point $x$ covered by $\+A^{b_l}_k$ such that $d_F(\tau[x,y],\sigma_l)\le \delta$. It implies that $\sigma_l$ is within a Fr\'echet distance $\delta$ to some subcurve of $\tau_k$. Hence, with probability at least $1-n^{-10}$, we can either find a single subcurve $\tau'$ of $\tau_k$ such that $d_F(\tau', \sigma_l)\le (3+2\epsilon)\delta$ or subcurves $\tau'_{l,r}$ of $\tau_k$ for every $r\in [(\mu_2-1)/\mu_3]$ such that $d_F(\tau'_{l,r}, \sigma_{l,r})\le (3+2\epsilon)\delta$.
		
		When we find $\tau'$, $\+I^4$ equals to the answer of {\sc Cover}$(\tau', (4+2\epsilon)\delta, \appReachWJArray{b_l}{k})$. Note that $y$ satisfies that there is a point $x$ covered by $\appReachWJArray{b_l}{k}$ such that $d_F(\tau[x, y], \sigma_l)\le \delta$. By the triangle inequality, $d_F(\tau[x,y], \tau')\le (4+2\epsilon)\delta$. Hence, $y$ must be covered by $\+I^4$ by the definition of the query {\sc Cover}.
		
		When we get $\tau'_{l,1}, \tau'_{l,2}, \ldots, \tau'_{l, (\mu_2-1)/\mu_3}$, we will construct a sequence $\+S^1, \+S^2,\ldots, \+S^{(\mu_2-1)/\mu_3}$ and set $\+I_4=\+S^{(\mu_2-1)/\mu_3}$. We prove that $\+S^r$ contains all points $y'\in \tau_k$ satisfying that there is a point $x$ covered by $\appReachWJArray{b_l}{k}$ with $x\le_\tau y'$ and $d_F(\tau[x, y'], \sigma[w_{b_l}, w_{b_l, r+1}])\le \delta$ by induction on $r$. 
		
		When $r=1$, the analysis is the same as the analysis in the case where we find $\tau'$. When $r\ge 2$, assume that $\+S^{r-1}$ satisfies the property. Take any point $y'\in \tau_k$ satisfying that there is a point $x$ covered by $\appReachWJArray{b_l}{k}$ with $x\le_\tau y'$ and $d_F(\tau[x,y'], \sigma[w_{b_l}, w_{b_l, r+1}])\le \delta$. The Fr\'echet matching between $\tau[x,y']$ and $\sigma[w_{b_l}, w_{b_l, r+1}]$ must matching the vertex $w_{b_l, r}$ to some point $z\in \tau[x, y']$. It implies that $d_F(\tau[x,z], \sigma[w_{b_l}, w_{b_l,r}])\le \delta$ and $d_F(\tau[z, y'], \sigma[w_{b_l, r}, w_{b_l, r+1}])\le \delta$. By the induction hypothesis, the point $z$ is covered by $\+S^{r-1}$. By the construction procedure, $\+S^r$ is the answer for {\sc Cover}$(\tau'_{l,r}, (4+2\epsilon)\delta, \+S^{r-1})$. Since $\sigma_{l,r}=\sigma[w_{b_l,r}, w_{b_l, r+1}]$ is within a Fr\'echet distance $(3+2\epsilon)\delta$ to $\tau'_{l,r}$, by the triangle inequality, $d_F(\tau'_{l,r}, \sigma_{l,r})\le (4+2\epsilon)\delta$. Hence, $y'$ must be covered by $\+S^r$ by the definition of the query {\sc Cover} as $z$ is covered by $\+S^{r-1}$. We finish proving the property for all $\+S^r$'s. Since $\+I^4=\+S^{(\mu_2-1)/\mu_3}$ and $w_{b_l, (\mu_2-1)/\mu_3+1}=w_{b_{l+1}}$, all points of type 4 are covered by $\+I^4$. 
		
		Next, we prove that all points covered by $\+I^4$ can form $((7+\varepsilon)\delta)$-reachable pairs with $w_{b_{l+1}}$. In the case where we find $\tau'$, $\+I^4$ equals to the answer of the query {\sc Cover}$(\tau', (4+2\epsilon)\delta, \appReachWJArray{b_l}{k})$. By definition of {\sc Cover}, all points $y$ covered by $\+I^4$ satisfied that there is a point $x$ covered by $\appReachWJArray{b_l}{k}$ such that $x\le_\tau y$ and $d_F(\tau[x,y], \tau')\le (1+\epsilon)\cdot (4+2\epsilon)\delta$. Given that $d_F(\tau', \sigma_l)\le (3+2\epsilon)\delta$, we have $d_F(\tau[x,y], \sigma_l)\le (7+8\epsilon+2\epsilon^2)\delta$ by the triangle inequality. Note that $\epsilon=\varepsilon/10$. Hence, $d_F(\tau[x,y], \sigma_l)\le (7+\varepsilon)\delta$. Given that $\appReachWJArray{b_l}{k}$ is $(7+\varepsilon)$-approximate reachable for $w_{b_l}$, it holds that $d_F(\tau[v_1, x], \sigma[w_1, w_{b_l}])\le (7+\varepsilon)\delta$. Hence, we can construct a matching between $\tau[v_1, y]$ and $\sigma[w_1, w_{b_{l+1}}]$ by concatenating the Fr\'echet matching between $\tau[v_1,x]$ and $\sigma[w_1, w_{b_l}]$ and the Fr\'echet matching between $\tau[x,y]$ and $\sigma_l$. The matching realizes a distance at most $(7+\varepsilon)\delta$. Therefore, $d_F(\tau[v_1, y], \sigma[w_1, w_{b_{l+1}}])\le (7+\varepsilon)\delta$, it means that $(y, w_{b_{l+1}})$ is $((7+\varepsilon)\delta)$-reachable.
		
		When we get $\tau'_{l,1}, \tau'_{l,2}, \ldots, \tau'_{l, (\mu_2-1)/\mu_3}$, we will construct a sequence $\+S^1, \+S^2,\ldots, \+S^{(\mu_2-1)/\mu_3}$ and set $\+I_4=\+S^{(\mu_2-1)/\mu_3}$. We prove that all points $y'$ covered by $\+S^r$ form $((7+\varepsilon)\delta)$-reachable pairs with $w_{b_l, r+1}$ by induction on $r$.
		
		When $r=1$, the analysis is the same as the analysis in the case where we find $\tau'$. When $r\ge 2$, assume that $\+S^{r-1}$ satisfies the property. Take any point $y'$ covered by $\+S^r$. Since $\+S^r$ is the answer for {\sc Cover}$(\tau'_{l,r},(4+2\epsilon)\delta,\+S^{r-1})$, by the definition of the query {\sc Cover}, we can find a point $x$ covered by $\+S^{r-1}$ such that $x\le_\tau y'$ and $d_F(\tau[x,y'], \tau'_{l,r})\le (1+\epsilon)\cdot(4+2\epsilon)\delta$. Given that $d_F(\tau'_{l,r}, \sigma_{l,r})\le (3+2\epsilon)\delta$, we have $d_F(\tau[x,y'], \sigma_{l,r})\le (7+8\epsilon+2\epsilon^2)\delta$ by the triangle inequality. Note that $\epsilon=\varepsilon/10$. Hence, $d_F(\tau[x,y'], \sigma_{l,r})\le (7+\varepsilon)\delta$. By the induction hypothesis, it holds that $d_F(\tau[v_1, x], \sigma[w_1, w_{b_l,r}])\le (7+\varepsilon)\delta$. Hence, we can construct a matching between $\tau[v_1, y']$ and $\sigma[w_1, w_{b_l, r+1}]$ by concatenating the Fr\'echet matching between $\tau[v_1, x]$ and $\sigma[w_1, w_{b_l, r}]$ and the Fr\'echet matching between $\tau[x,y']$ and $\sigma_{l,r}$. The matching realizes a distance at most $(7+\varepsilon)\delta$. Therefore, $d_F(\tau[v_1, y'], \sigma[w_1, w_{b_l, r+1}])\le (7+\varepsilon)\delta$. We finish proving the property for all $\+S^r$'s.  Since $\+I^4=\+S^{(\mu_2-1)/\mu_3}$ and $w_{b_l, (\mu_2-1)/\mu_3+1}=w_{b_{l+1}}$, all points covered by $\+I^4$ can form $((7+\varepsilon)\delta)$-reachable pairs with $w_{b_{l+1}}$.
	\end{proof}
	
	Next, we construct $\appReachVIArray{a_{k+1}}{l}$ based on $\+I^1$ and $\+I^2$ and construct $\appReachWJArray{b_{l+1}}{k}$ based on $\+I^3$ and $\+I^4$. For every $j\in [b_l, b_{l+1}-1]$, we set $\appReachVI{a_{k+1}}{l}{j}$ to be empty if both $\+I^1_{w_j}$ and $\+I^2_{w_j}$ are empty. Otherwise, take the minimum point $p$ with respect to $\le_\sigma$ in $\+I^1_{w_j}\cup\+I^2_{w_j}$. We set $\appReachVI{a_{k+1}}{l}{j}=pw_{j+1}\cap \+B(v_{a_{k+1}}, \delta)$. 
	For every $i\in [a_k, a_{k+1}-1]$, we set $\appReachWJ{b_{l+1}}{k}{i}$ to be empty if both $\+I^3_{v_i}$ and $\+I^4_{v_i}$ are empty. Otherwise, take the minimum point $x$ with respect to $\le_\tau$ in $\+I^3_{v_i}\cup\+I^4_{v_i}$. We set $\appReachWJ{b_{l+1}}{k}{i}=xv_{i+1}\cap \+B(w_{b_{l+1}}, \delta)$. We finish implementing {\sc Reach}. 
	
	Putting Lemma~\ref{lem:cover1},~\ref{lem: marked-edge},~\ref{lem:cover},~\ref{lem: I_1},~\ref{lem: I_2},~\ref{lem: I_3} and~\ref{lem: I_4} together, we have the following lemma. We treat both $\epsilon$ and $d$ as fixed constants.
	
	\begin{lemma}\label{lem:reach}
		There is an algorithm that preprocesses every $\tau_k$ in $O(\mu_1^5)$ time to implement the procedure {\sc Reach} in $O(\mu_1(\mu_3+\mu_2/\mu_3+\mu_2^5\log n/(\omega\mu_3))+\omega\mu_2^4)$ time with success probability at least $1-n^{-10}$.
	\end{lemma}

	We then carefully choose values for parameters $\mu_1, \mu_2, \mu_3$, and $\omega$ to get the running time of {\sc Reach} to be $o(\mu_1\mu_2)$. We set $\mu_1=m^{0.24}$, $\mu_2=m^{0.02}$, $\mu_3=m^{0.01}$, and $\omega=m^{0.12}$ so that $O(\mu_1(\mu_3+\mu_2/\mu_3+\mu_2^5\log n/(\omega\mu_3))+\omega\mu_2^4)=O(m^{0.25})$ and $\mu_1\mu_2=m^{0.26}$.
	As shown in the overview, we call {\sc Reach} for $mn/(\mu_1\mu_2)$ times to get a decision procedure. It succeeds if all these invocations of {\sc Reach} succeed. 
	
	
	\begin{theorem}\label{thm:Frechet}
		Given two polygonal curves $\tau$ and $\sigma$ in $\mathbb{R}^d$ for some fixed $d$, there is a randomized $(7+\varepsilon)$-approximate decision procedure for determining $d_F(\tau, \sigma)$ in $O(nm^{0.99})$ time with success probability as least $1-n^{-7}$. 
	\end{theorem}
	
	We use the approach developed in~\cite{colombe2021approximating} to get an approximation algorithm. For any $\varepsilon\in (0,1)$, it gets an $((1+\varepsilon)\alpha)$-approximate algorithm for computing $d_F(\tau, \sigma)$ by carrying out $O(\log(n/\varepsilon))$ instances of any $\alpha$-approximate decision procedure. In our case, the approximate algorithm succeeds if all these instances succeed.
	
	\begin{theorem}\label{thm:approx_Frechet}
		Given two polygonal curves $\tau$ and $\sigma$ in $\mathbb{R}^d$ for some fixed $d$, there is a randomized $(7+\varepsilon)$-approximate algorithm for computing $d_F(\tau, \sigma)$ in $O(nm^{0.99}\log{n/\varepsilon})$ time with success probability at least $1-n^{-6}$.
	\end{theorem}

	\section{Subquadratic decision algorithm for discrete Fr\'echet distance}\label{sec:discrete}

	We extend our technical framework to the discrete Fr\'echet distance.
	Recall that the  \emph{discrete} Fr\'echet distance of $\tau, \sigma$ is $\tilde{d}_F(\tau, \sigma) = \inf_{\+M}d_{\+M}(\tau, \sigma)$, where $\+M$ is a matching that matches at least one vertex of $\sigma$ to a vertex of $\tau$, and vice versa. We first present the discrete counterpart of {\sc WaveFront}, and the matching construction and curve simplification of the discrete Fr\'echet distance.
	

	We only need to consider vertices of $\tau$ and $\sigma$ for determining $\tilde{d}_F(\tau, \sigma)$. Take $i, i'\in [n]$ and $j, j'\in [m]$ with $i\le i'$ and $j\le j'$. For any value $r\ge 0$, we call $(v_{i'}, w_{j'})$ is \emph{$r$-reachable} from $(v_i, w_j)$ if $\tilde{d}_F(\tau[v_i, v_{i'}], \sigma[w_j, w_{j'}])\le r$. When a pair $(v_i, w_j)$ is $r$-reachable from $(v_1, w_1)$, we call it a \emph{$r$-reachable pair}.


	We name the discrete counterpart of {\sc WaveFront} as {\sc DisWave}. Fix $r=\delta$. Given $\tau$, $\sigma$, a set of some $\tau$'s vertices $\+S$, and a set of some $\sigma$'s vertices $\+S'$, {\sc DisWave}$(\tau, \sigma, \delta, \+S, \+S')$ aims to return $\bigl(\+{DW}^{v_i}\bigr)_{i\in[n]}$ and $\bigl(\+{DW}^{w_j}\bigr)_{j\in [m]}$. For every $i\in [n]$, $\+{DW}^{v_i}$ is a set of vertices of $\sigma$ such that $w_j$ belongs to $\+{DW}^{v_i}$ if and only if $(v_i, w_j)$ is $\delta$-reachable from $(v_1, w_{j'})$ for some $w_{j'}\in \+S'$ or $(v_i, w_j)$ if $\delta$-reachable from $(v_{i'}, w_1)$ for some $v_{i'}\in \+S$. For every $j\in [m]$, $\+{DW}^{w_j}$ is a set of vertices of $\tau$ such that $v_i$ belongs to $\+{DW}^{w_j}$ if and only if $(v_i, w_j)$ is $\delta$-reachable from $(v_1, w_{j'})$ for some $w_{j'}\in \+S'$ or $(v_i, w_j)$ if $\delta$-reachable from $(v_{i'}, w_1)$ for some $v_{i'}\in \+S$. 


	We present {\sc DisWave} that runs in $O(mn)$ time. It is also based on dynamic programming. We assume that given arbitrary sets $\+S$ and $\+S'$ of $\tau$'s vertices and $\sigma$'s vertices, respectively, we can determine whether $v_i\in \+S$ and $w_j\in \+S$ in $O(1)$ time. Because we can always represent $\+S$ by a binary array $\+A$ of size $n$ such that $\+A[i]=1$ if and only if $v_i\in \+S$, so does $\+S'$.

	\vspace{4pt}

	\noindent \pmb{\sc DisWave$(\tau, \sigma, \delta, \+S, \+S')$.} We first initialize $\+{DW}^{v_1}$ as $\emptyset$. For $w_1$, insert $w_1$ into $\+{DW}^{v_1}$ if $w_1\in\+{S}'$ and $d(v_1, w_1) \leq \delta$. For $j \in [2, m]$, insert $w_j$ into $\+{DW}^{v_1}$ if either $w_j\in \+{S}'$ and $d(v_1, w_j) \leq \delta$ or $w_{j-1}\in \+{DW}^{v_1}$ and $d(v_1, w_j) \leq \delta$. 
	


	We then use the following recurrence to construct $\+{DW}^{v_i}$ for $i \in [2, n]$. We insert $w_j$ into $\+{DW}^{v_i}$ if and only if $d(v_i, w_j) \leq \delta$ and $w_j\in \+{DW}^{v_{i-1}}$ or $w_{j-1}\in\+{DW}^{v_i}$ or $w_{j-1}\in \+{DW}^{v_{i-1}}$. 
	
	Finally, for any $i\in [n]$ and $j\in [m]$, we insert $v_i$ into $\+{DW}^{w_j}$ if and only if $w_j\in \+{DW}^{v_i}$.


	\vspace{4pt}

	\noindent\textbf{Matching construction.} Suppose that $\tilde{d}_F(\tau, \sigma)\le \delta$. We use the following lemma to compute a matching $\+M$ between $\tau$ and $\sigma$ satisfying $d_{\+M}(\tau, \sigma)\le \delta$ explicitly. Note that $\+{M}$ matches each vertex of $\tau$ to a vertex of $\sigma$ and vice verse. We can construct such a matching by backtracking the output of {\sc DisWave}$(\tau, \sigma, \delta, \+S, \+S')$, where $\+S=\{v_1\}$ and $\+S'=\{w_1\}$. For a vertex $p$ of $\tau$ or $\sigma$, let $\+M(p)$ be a vertex matched to $p$ by $\+M$.

	\begin{lemma}~\label{lem: discrete matching}
		Given two polygonal curves $\tau, \sigma$ in $\mathbb{R}^d$, suppose $\tilde{d}_F(\tau, \sigma) \leq \delta$. There is an $O(nm)$-time algorithm for computing a matching $\+M$ between $\tau, \sigma$ such that $d_{\+M}(\tau, \sigma) \leq \delta$. $\+M$ can be stored in $O(n+m)$ space such that for any vertex $p$ of $ \tau $ or $ \sigma$, we can retrieve a vertex $\+{M}(p)$ in $O(1)$ time.
	\end{lemma}

	\begin{proof}	
		We can backtrack the output $\bigl(\+{DW}^{v_i}\bigr)_{i\in[n]}$ of {\sc DisWave}$(\tau, \sigma, \delta, \+S, \+S')$ to construct $\+M$, where $\+S=\{v_1\}$ and $\+S'=\{w_1\}$. For every vertex $p$ of $\tau$ or $\sigma$, we compute $\+{M}(p)$ and store it. Let $\+{L}$ be an empty sequence.
		Since $\tilde{d}_F(\tau, \sigma) \leq \delta$, we have $w_m\in \+{DW}^{v_n}$, we add $(v_n, w_m)$ to $\+{L}$. According to the procedure of {\sc DisWave}, it must be true that $w_m\in \+{DW}^{v_{n-1}}$ or $w_{m-1}\in \+{DR}^{v_n}$ or $w_{m-1}\in \+{DW}^{v_{n-1}}$. Suppose $w_m\in \+{DW}^{v_{n-1}}$, we add $(v_{n-1}, w_m)$ to $\+{L}$. Then it must happen that $w_m\in \+{DW}^{v_{n-2}}$ or $w_{m-1}\in \+{DW}^{v_{n-1}}$ or $w_{m-1}\in \+{DW}^{v_{n-2}}$. We repeat this process until we add $(v_1, w_1)$ to $\+{L}$. Finally, we get a sequences for pairs of vertices from $(v_n, w_m)$ to $(v_1, w_1)$.

		Traverse $\+L$ in order, for each pair $(v_i, w_j)$, if $\+{M}(v_i)$ doesn't exist, we set $\+{M}(v_i) = w_j$; if $\+{M}(w_j)$ doesn't exist, we set $\+{M}(w_j) = v_i$; if neither $\+{M}(v_i)$ nor $\+{M}(w_j)$ exists, we set $\+{M}(v_i) = w_j$ and $\+{M}(w_j) = v_i$. According to the procedure of {\sc DisWave}$(\tau, \sigma, \delta,\+S, \+S')$, we have $d(v_i, w_j) \leq \delta$ for all pairs $(v_i, w_j) \in \+L$. So, the matching $\+M$ satisfies $d_{\+{M}}(\tau, \sigma) \leq \delta$. We store $\+M$ in a linear space, then for any vertex $p \in \tau \cup \sigma$, we can retrieve a vertex $\+{M}(p)$ in $O(1)$ time.
	\end{proof}

	\noindent\textbf{Curve simplification.} We will need to simplify a curve $\tau$ under the discrete Fr\'echet distance. We employ an approximate algorithm developed in \cite[Lemma 23]{filtser2023approximate}

	\begin{lemma}[Lemma 23~\cite{filtser2023approximate}]~\label{lem: discrete simplification}
		Let $\tau$ be a curve of $n$ vertices in $\mathbb{R}^d$. Given $\delta > 0$, and $\epsilon \in (0, 1)$, there is an algorithm that returns a curve $\tau'$ of size at most $k^*$ in $O(\frac{d \cdot n \log n}{\epsilon} + n \cdot \epsilon^{-4.5} \log \frac{1}{\epsilon})$ time such that 
		$\tilde{d}_{F}(\tau, \tau') \leq (1+\epsilon)\delta$, where $k^*=\min\{|\tau''|:\tilde{d}_F(\tau, \tau'')\le \delta\}$. 
	\end{lemma}	

	As with the Fr\'echet distance, we divide $\tau$ and $\sigma$ into short subcurves and implement a discrete counterpart of {\sc Reach} called {\sc DisReach}. Let $\mu_1$ and $\mu_2$ be two integers with $\mu_1>\mu_2$. Define $a_k=(k-1)\mu_1+1$ for $k\in [(n-1)/\mu_1]$. We divide $\tau$ into a collection $(\tau_1, \tau_2,\ldots, \tau_{(n-1)/\mu_1})$ of $(n-1)/\mu_1$ subsequences such that $\tau_k=\tau[v_{a_k}, v_{a_{k+1}}]$. Note that every subsequence $\tau_k$ contains $\mu_1+1$ vertices, and $\tau_{k-1}\cap \tau_k=v_{a_k}$. 

	We divide $\sigma$ into shorter subsequences. Define $b_l=(l-1)\mu_2+1$ for $l\in [(m-1)/\mu_2]$. We divide $\sigma$ into a collection $(\sigma_1, \sigma_2,\ldots, \sigma_{(m-1)/\mu_2})$ of $(m-1)/\mu_2$ subsequences such that $\sigma_l=\sigma[w_{b_l}, w_{b_{l+1}}]$. Every subsequence $\sigma_l$ contains $\mu_2+1$ vertices, and $\sigma_{l-1}\cap \sigma_l=w_{b_l}$.

	For any $v_i$ and $\alpha > 1$, we call a set of $\sigma$'s vertices \emph{$\alpha$-approximate reachability set} if all vertices of $\sigma$ that can form $\delta$-reachable pairs with $v_i$ are in this set, and every vertex $w_j$ in this set satisfies that $(v_i, w_j)$ is an $(\alpha\delta)$-reachable pair, and $d(v_i, w_j) \leq \delta$. For any subcurve $\sigma[w_j, w_{j_1}]$, a set $\+{S}$ of vertices of $\sigma[w_j, w_{j_1}]$ is \emph{$\alpha$-approximate reachable} for $v_i$ if $\+S$ is an intersection between the set of all vertices of $\sigma[w_j, w_{j_1}]$ and some $\alpha$-approximate reachability set of $v_i$.
	We can define the approximate reachability set for a vertex $w_j$ of $\sigma$ similarly. 



	We implement the following subroutine {\sc DisReach} in $O(\mu_1\mu_2)^{1-c}$ time for some constant $c \in (0,1)$. We manage to get the value of $\alpha$ to be $7+\varepsilon$ for some $\varepsilon\in (0,1)$.

	\vspace{4pt}
	
	\noindent\fbox{\parbox{\dimexpr\linewidth-2\fboxsep-2\fboxrule\relax}{\textbf{Procedure {\sc DisReach}}

	\vspace{4pt}

	\textbf{Input:} a subcurve $\tau_k$, a subcurve $\sigma_l$, an set $\+{DA}^{a_k}_l$ that is $\alpha$-approximate reachable for $v_{a_k}$, an array $\+{DA}^{b_l}_{k}$ induced by $\tau_k$ that is $\alpha$-approximate reachable for $w_{b_l}$. 
	
	\vspace{4pt}
	
	\textbf{Output:} an array $\+{DA}^{a_{k+1}}_{l}$ induced by $\sigma_l$ that is $\alpha$-approximate reachable for $v_{a_{k+1}}$ and an array $\+{DA}^{b_{l+1}}_{k}$ induced by $\tau_k$ that is $\alpha$-approximate reachable for $w_{b_{l+1}}$.
 	}
 	}

	\vspace{4pt}
	
	As with the design of {\sc Reach}, we classify vertices of $\tau_k$ and $\sigma_l$ that we aim to cover into four types and deal with each type separately. 


	\vspace{4pt}

	\noindent\fbox{\parbox{\dimexpr\linewidth-2\fboxsep-2\fboxrule\relax}{
	\noindent\textbf{T1.} $w_j$: there is a vertex $w_{j'}\in \+{DA}^{a_k}_{l}$ such that $(v_{a_{k+1}}, w_j)$ is $\delta$-reachable from $(v_{a_k}, w_{j'})$.

	\noindent\textbf{T2.} $w_j$: there is a vertex $v_{i'}\in \+{DA}_{k}^{b_l}$ such that $(v_{a_{k+1}}, w_j)$ is $\delta$-reachable from $(v_{i'}, w_{b_l})$. 

	\noindent\textbf{T3.} $v_i$: there is a vertex $w_{j'}\in \+{DA}^{a_k}_{l}$ such that $(v_i, w_{b_{l+1}})$ is $\delta$-reachable from $(v_{a_k}, w_{j'})$. 
	

	\noindent\textbf{T4.} $v_i$: there is a vertex $v_{i'}\in \+{DA}_{k}^{b_l}$ such that $(v_i, w_{b_{l+1}})$ is $\delta$-reachable from $(v_{i'}, w_{b_l})$. 
	
	}}

	\vspace{4pt}

	Via a similar analysis in the proof of Lemma~\ref{lem:cover1}, we can show that every vertex of $\sigma_l$ that can form a $\delta$-reachable pair with $v_{a_{k+1}}$ must belongs to type 1 or 2, and every vertex of $\tau_k$ that can form a $\delta$-reachable pair with $w_{b_{l+1}}$ must belongs to type 3 or 4.

	\vspace{4pt}


	Our decision algorithm for computing the discrete Fr\'echet distance also consists of a preprocessing phase, and a dynamic programming that uses {\sc DisReach}. To implement {\sc DisReach}, we construct the sets $\+I^1-\+I^4$ to covers vertices of T1-T4, respectively.

	\vspace{4pt}

	\noindent\textbf{Preprocessing for $\+I^1$, $\+I^2$ and $\+I^3$.} We first generate $\zeta_k$, $\zeta_k^{\text{suf}}$ and $\zeta_k^{\text{pre}}$. Set $\epsilon = \varepsilon / 10$ in \Cref{lem: discrete simplification}, we first call the algorithm in \Cref{lem: discrete simplification}  with $\tau_k$ and $\delta$. If the algorithm returns a curve of at most $\mu_2 + 1$ vertices, we set $\zeta_k$ to be the curve returned; otherwise, set $\zeta_k$ to be null. It takes $O(d \cdot \mu_1 \log \mu_1/\epsilon + \mu_1\cdot\epsilon^{-4.5}\log (1/\epsilon))$ time. To generate $\zeta_k^{\text{suf}}$, we invoke the curve simplification algorithm with $\tau[v_i, v_{a_{k+1}}]$ and $\delta$ for all $i\in [a_k+1, a_{k+1}]$. We then identify the minimum $i$ such that the algorithm returns a curve of at most $\mu_2+1$ vertices and set $\zeta_k^{\text{suf}}$ to be the output. It takes $O(d \cdot \mu_1^2 \log \mu_1/\epsilon + \mu_1^2\cdot\epsilon^{-4.5}\log (1/\epsilon))$ time as there are $O(\mu_1)$ suffixes to try. We can generate $\zeta_k^{\text{pre}}$ by trying all prefixes of $\tau_k$ in $O(d \cdot \mu_1^2 \log \mu_1/\epsilon + \mu_1^2\cdot\epsilon^{-4.5}\log (1/\epsilon))$ time in the same way. Let $\tau[v_{a_k},v_{i_{\text{pre}}}]$ and $\tau[v_{i_{\text{suf}}}, v_{a_{k+1}}]$ be the corresponding prefix and suffix of $\zeta_k^{\text{pre}}$ and $\zeta_k^{\text{suf}}$, respectively.

	We also need to store a matching $\Mpre$ such that $d_{\Mpre}(\tau[v_{a_k}, v_{i_{\text{pre}}}], \zeta_k^{\text{pre}})\le (1+\epsilon)\delta$ and a matching $\Msuf$ such that $d_{\Msuf}(\tau[v_{i_{\text{suf}}}, v_{a_{k+1}}], \zeta_k^{\text{suf}})\le (1+\epsilon)\delta$. 
	We use the algorithm in \Cref{lem: discrete matching} to finish it in $O(\mu_1\mu_2)$ time.

	\vspace{4pt}

	\noindent\textbf{Preprocessing for $\+I^4$.} 
	For any curve $\zeta$, a vertex $v_i$ of $\tau_k$ is defined to be \emph{marked} by $\zeta$ if there is a subcurve $\tau'$ of $\tau_k$ that contains $v_i$ such that $\tilde{d}_F(\tau', \zeta) \leq \delta$. We call $\zeta$ \emph{$\omega$-dense} if it marks at least $\omega$ vertices of $\tau_k$.
	We first preprocess $\tau_k$ so that for any $l\in [(m-1)/\mu_2]$ and any subcurve $\sigma'=\sigma[w_{j_1}, w_{j_2}]$ of $\sigma_l$, given a vertex marked by $\sigma'$, we can find a subcurve of $\tau_k$ that is close to $\sigma'$ efficiently. 
	
	For every vertex $v_i$ of $\tau_k$, we identify the longest suffix of $\tau[v_{a_k}, v_i]$ and the longest prefix of $\tau[v_i, v_{a_{k+1}}]$ such that running the algorithm in \Cref{lem: discrete simplification} on them returns a curve of at most $\mu_2+1$ vertices with respect to $\delta$. It takes $O\left(d \cdot \mu_1^2 \log \mu_1/\epsilon + \mu_1^2\cdot\epsilon^{-4.5}\log (1/\epsilon) \right)$ time as there are $O(\mu_1)$ prefixes and suffixes to try. Let $\tau[\bar{v}_i, v_i]$ and $\tau[v_i, \tilde{v}_i]$ be the corresponding suffix and prefix, respectively. Let $\bar{\zeta}_i$ and $\tilde{\zeta}_i$ be the simplified curves for $\tau[\bar{v}_i, v_i]$ and $\tau[v_i, \tilde{v}_i]$, respectively. We have $d_F(\tau[\bar{v}_i, v_i], \bar{\zeta}_i)\le (1+\epsilon)\delta$ and $d_F(\tau[v_i, \tilde{v}_i], \tilde{\zeta}_i)\le (1+\epsilon)\delta$. By \Cref{lem: discrete matching}, we further construct a matching $\bar{\+M}_i$ between $\tau[\bar{v}_i, v_i]$ and $\bar{\zeta}_i$ and a matching $\tilde{\+M}_i$ between $\tau[v_i, \tilde{v}_i]$ and $\tilde{\zeta}_i$ in $O(\mu_1\mu_2)$ time. It takes $O\left(d \cdot \mu_1^3 \log \mu_1/\epsilon + \mu_1^3\cdot\epsilon^{-4.5}\log (1/\epsilon) \right)$ time to process all $v_i$'s.

	Next, we present how to find a subcurve of $\tau_k$ that is close to $\sigma'$ based on the above preprocessing. Suppose that $v_i$ is marked by $\sigma'$. We observe that there must be a subcurve $\tau'$ of $\tau[\bar{v}_i, \tilde{v}_i]$ with $\tilde{d}_F(\tau', \sigma')\le \delta$. By definition, there is a subcurve $\tau[v_x,v_y]$ of $\tau_k$ such that $\tilde{d}_F(\tau[v_x, v_y], \sigma')\le \delta$ and $v_i \in \tau[v_x, v_y]$. It holds that $\bar{v}_i\le_\tau v_x$. Otherwise, there is smaller $v_{i'} <_{\tau} \bar{v}_i$ such that calling algorithm in \Cref{lem: discrete simplification} with $\tau[v_{i'}, v_i]$ returns a curve of at most $\mu_2+1$ vertices, which is a contradiction. We have $v_y\le_\tau \tilde{v}_i$ via the similar analysis.  Hence, $\tau[v_x,v_y]$ is a subcurve of $\tau[\bar{v}_i, \tilde{v}_i]$. Intuitively, we can join the last vertex of $\bar{\zeta}_i$ and  the first vertex of $\tilde{\zeta}_i$ to get a new curve $\zeta'$ of at most $2\mu_2+2$ vertices. Obviously, $\tilde{d}_F(\tau[\bar{v}_i, \tilde{v}_i], \zeta') \leq (1+\epsilon)\delta$. Since there is a subcurve $\tau'$ of $\tau[\bar{v}_i, \tilde{v}_i]$ such that $\tilde{d}_F(\tau', \sigma') \leq \delta$, there is a subcurve $\zeta''$ of $\zeta'$ such that $\tilde{d}_F(\zeta'', \sigma') \leq (2+\epsilon)\delta$. 

	We aim to find a $\zeta''$. For every vertex of $\zeta'$ if it  belongs to the $\+{B}(w_j, (2+\epsilon)\delta)$, we insert it into a set $X$. If this vertex belongs to the ball $\+{B}(w_{j_1}, (2+\epsilon)\delta)$, we insert it into another set $Y$.  The existence of $\zeta''$ implies the existence of some subcurve that starts from a point in $X$, ends at a point in $Y$, and locates within a discrete Fr\'echet distance $(2+\epsilon)\delta$ of $\sigma'$. We test all subcurves of $\zeta'$ starting from some point in $X$ and ending at some point in $Y$. There are $O(|\zeta'|^2)=O(\mu_2^2)$ subcurves to be tested. For every subcurve, we check whether the discrete Fr\'echet distance between it and $\sigma'$ is at most $(2+\epsilon)\delta$. If so, we return it as $\zeta''$. It takes $O(\mu_2^4)$ time.

	Suppose $\zeta'' = \zeta'[x, y]$. We finally find a subcurve of $\tau_k$ within a discrete Fr\'echet distance $(1+\epsilon)\delta$ to $\zeta''$ based on $\bar{\+M}_i$ and $\tilde{\+M}_i$. Since $\zeta'$ is a concatenation of $\bar{\zeta}_i$ and $\tilde{\zeta}_i$. If $x \in \bar{\zeta}_i$, we set $x' = \bar{\+{M}}_i(x)$; otherwise, we set $x' = \tilde{\+M}_i(x)$. We can define $y'$ for $y$ similarly. And we have $\tilde{d}_F(\tau[x', y'], \zeta'') \leq (1+\epsilon)\delta$. Hence, $\tilde{d}_F(\tau[x', y'], \sigma') \leq (3+2\epsilon)\delta$ by the triangle inequality. 

	\begin{lemma}\label{lem: marked-vertex}
    	We can preprocess $\tau_k$ in $O\left(d \cdot \mu_1^3 \log \mu_1 /\epsilon + \mu_1^3\cdot\epsilon^{-4.5}\log(1/\epsilon) \right)$ time such that for any subcurve $\sigma'$ of $\sigma_l$ and a vertex of $\tau_k$, there is an $O(\mu_2^4)$-time algorithm that returns a subcurve $\tau'$ of $\tau_k$ with $d_F(\tau', \sigma')\le (3+2\epsilon)\delta$ or null. If the algorithm returns null, the vertex is not marked by $\sigma'$. 
	\end{lemma} 

	The last part of the preprocessing is a data structure for answering the query {\bf{\sc DisCover}} with a vertex-to-vertex subcurve $\tau'$ of $\tau_k$, $\delta'>0$, and a set $\+S$ of vertices of $\tau_k$. {\sc DisCover} is a discrete version of {\sc Cover}. Its definition is as follows.

	\vspace{2pt}

	\noindent\pmb{{\sc DisCover}$(\tau', \delta', \+S)$.} The query output is another set $\bar{\+S}$ of $\tau_k$'s vertices such that a vertex $v_i$ belongs to $\bar{\+S}$ if and only if there is a vertex $v_{i'}\in \+{S}$ such that $v_{i'}\leq_{\tau} v_i$ and $\tilde{d}_F(\tau[v_{i'}, v_i], \tau') \leq \delta'$. 

	Suppose that $\tau'=\tau[v_{i_1}, v_{i_2}]$. Note that $\bar{\+S}$ is equal to the output array $\+{DW}^{v_{i_2}}$ returned by the invocation {\sc DisWaveFront}$(\tau_k, \tau', \delta', \+S, \+S')$ for the last vertex $v_{i_2}$ of $\tau'$, where $\+S'$ is an emptyset.

	\vspace{4pt}

	\noindent\underline{\bf Data structure.} For every vertex-to-vertex subcurve $\tau[v_{i_1}, v_{i_2}]$ of $\tau_k$, take a vertex $v_i$ of $\tau_k$. We aim to identify all vertices $v_{i'}$ in $\tau_k$ such that $v_{i}\le_\tau v_{i'}$ and $\tilde{d}_F(\tau[v_i, v_{i'}], \tau[v_{i_1}, v_{i_2}])\le \delta'$. We organize all such $v_{i'}$'s into a set and assign it to an element $\+D[i_1, i_2, i]$ in a 3D array. This array $\+D$ is our data structure. Every element of $\+D$ can be constructed by an invocation of {\sc DisWave}. 
	


	We also maintain a 3D array {\sc Max}. For any $i_1, i_2, i$, let $i^*$ be the maximum integer such that $v_{i^*}\in \+{D}[i_1,i_2, i]$, we set {\sc Max}$[i_1, i_2, i] = i^*$. If $\+{D}[i_1,i_2, i]$ is empty, we set {\sc Max}$[i_1, i_2, i] = -1$.

	For every subcurve $\tau[v_{i_1}, v_{i_2}]$ and every vertex $v_{i}$, it takes $O\left(\mu_1(i_2 - i_1) \right)$ time to determine every element of $\+{D}[i_1, i_2, i]$ and {\sc Max}$[i_1, i_2, i]$. Since there are $O(\mu_1^3)$ distinct combinations of $i_1, i_2$ and $i$. It takes $O(\mu_1^5)$ time to construct $\+{D}$ and {\sc Max} as the data structures.

	\noindent\underline{\bf Query algorithm.} Given an arbitrary subcurve $\tau' = \tau[v_{i_1}, v_{i_2}]$ of $\tau_k$ and an arbitrary set $\+S$ of $\tau_k$'s vertices, we present how to answer the query {\sc DisCover} in $O(\mu_1)$ time by using the data structures. Basically, the feasible solution $\bar{\+S}$ is the union of $\+D[i_1, i_2, i]$ for all $v_i\in \+S$.

	Initialize $\bar{\+S}=\emptyset$. We traverse $\+S$ to update $\bar{\+{S}}$ progressively. Assume that all vertices in $\+S=\{v_{a_1}, v_{a_2},\ldots, v_{a_{|\+S|}}\}$ are sorted with $a_1<a_2<\ldots<a_{|\+S|}$. We use $i$ to index the current element in $\+{S}$ within the traversal, and use $i^*$ to store the maximum integer such that $v_{i^*}\in \bar{\+S}$. First, set $\bar{\+S}=\+D[i_1, i_2, a_1]$, set $i^*=${\sc Max}$[i_1, i_2, a_1]$, and increase $i$ to $a_2$. Suppose that we have updated $\bar{\+S}$ to be the union $\bigcup_{r\in [b]} \+D[i_1, i_2, a_r]$, and $i=a_{b+1}$. If {\sc Max}$[i_1, i_2, i]\le i^*$, increase $i$ to $a_{b+1}$. Otherwise, insert all vertices in $\+D[i_1, i_2, i]$ with indices larger than $i^*$ into $\bar{\+S}$, and update $i^*$ to {\sc Max}$[i_1, i_2, i]$. We stop after all vertices in $\+S$ have been processed. It takes $O(\mu_1)$ time.
	


	\begin{lemma}\label{lem:dis-cover}
    	Fix $\tau_k$ and $\delta'$. There is a data structure of size $O(\mu_1^4)$ and preprocessing time $O(\mu_1^5)$ that answers the query {\sc DisCover} for any $\tau'$ and $\+S$ in $O(\mu_1)$ time . 
	\end{lemma}

	Finally, we present how to construct $\+I^1-\+I^4$, respectively. Assume that $\+{DA}^{a_{k}}_{l}$ is $(7+\varepsilon)$-reachable for $v_{a_k}$ and $\+{DA}_{k}^{b_l}$ is $(7+\varepsilon)$-reachable for $w_{b_l}$. 
	
	\vspace{4pt}

	\noindent\textbf{Construction of $\+I^1$.}  
	If $\zeta_k$ is not null,  let $\+{S}^k$ be an emptyset. We invoke {\sc DisWave}$(\zeta_k, \sigma_l, (2+\epsilon)\delta,  \+{S}^k, \+{DA}^{a_k}_l)$. Let $\+A^1$ be the output set for the last vertex of $\zeta_k$. Set $\+I^1 = \+{A}^1$. If $\zeta_k$ is null, set $\+I^1$ to be $\emptyset$. It runs in $O(\mu_2^2)$ time. 

	\begin{lemma}~\label{lem: Dis I_1}
		We can construct $\+I^1$ in $O(\mu_2^2)$ time such that $\+I^1$ contains all vertices of type 1, and every vertex $w_j \in \+I^4$ forms a $((7+\varepsilon)\delta)$-reachable pair with $v_{a_{k+1}}$. 
	\end{lemma}
	\begin{proof}
		The running time follows the construction procedure. We focus on proving $\+I^1$'s property. 
  
		If $\zeta_k$ is null, it means that $\tau_k$ at a discrete Fr\'echet distance more than $\delta$ to all curves of at most $\mu_2+1$ vertices, which implies that vertices of type 1 do not exist. According to the above procedure, $\+I^1=\emptyset$. There is nothing to prove. 

		If $\zeta_k$ is not null, we have $\+I^1 = \+A^1$. Suppose $w_j$ belongs to type 1. There is a vertex $w_{j'} \in \sigma_l\in \+{DA}^{a_k}_l$ such that $w_{j'} \leq_\sigma w_j$ and $\tilde{d}_F(\tau_k, \sigma[w_{j'}, w_j]) \leq \delta$ by definition. According to \Cref{lem: discrete simplification}, $\tilde{d}_F(\tau_k, \zeta_k) \leq (1+\epsilon)\delta$. We have $\tilde{d}_F(\zeta_k, \sigma[w_{j'}, w_j]) \leq (2 + \epsilon) \delta$ via the triangle inequality. Recall that $\+{A}^1$ is the output array for the last vertex of $\zeta_k$. By definition of {\sc DisWave}, $w_j\in \+A^1$. Hence, $w_j$ belongs to $\+I^1$.
		

		Next, we prove that every vertex $w_j \in \+{I}^1$ can form a $((7+\varepsilon)\delta)$-reachable pair with $v_{a_{k+1}}$. 
		Since $\+{S}^k$ is empty, by the definition of {\sc DisWave}, there exists $w_{j'} \leq_\sigma w_j$ such that $\tilde{d}_F(\zeta_k, \sigma[w_{j'}, w_j]) \leq (2+\epsilon)\delta$ and $w_{j'}\in \+{DA}^{a_k}_l$. Note that $\tilde{d}_{F}(\tau[v_1, v_{a_k}], \sigma[w_1, w_{j'}]) \leq (7 + \varepsilon) \delta$. 
		Since $\epsilon = \varepsilon/10$, by triangle inequality, we have $\tilde{d}_F(\tau_k, \sigma[w_{j'}, w_j]) \leq (3+2\epsilon)\delta < (7+\varepsilon)\delta$. It implies that we can construct a matching between $\tau[v_1, v_{a_{k+1}}]$ and $\sigma[w_1, w_j]$ by concatenating the discrete Fr\'echet matching between $\tau_k$ and $\sigma[w_{j'},w_j]$ to that between $\tau[v_1, v_{a_k}]$ and $\sigma[w_1, w_{j'}]$. The matching realizes a distance at most $(7+\varepsilon)\delta$. Thus, $(v_{a_{k+1}}, w_j)$ is a $(7 + \varepsilon)$-reachable pair. 
 	\end{proof}

	\noindent\textbf{Construction of $\+I^2$.}
	We first construct a set $\+S$ of vertices of $\suf$ based on $\+{DA}_k^{b_l}$. Recall that $\suf$ is a simplification of the suffix $\tau[v_{i_{\text{suf}}}, v_{a_{k+1}}]$. We also have a matching $\Msuf$ with $d_{\Msuf}(\tau[v_{i_{\text{suf}}}, v_{a_{k+1}}], \zeta_k^{\text{suf}})\le (1+\epsilon)\delta$. According to \Cref{lem: discrete matching}, for every $v_i \in [v_{i_{\text{suf}}}, v_{a_{k+1}}]$, we can retrieve $\Msuf(v_i)$ of $\zeta_k^{\text{suf}}$ in $O(1)$ time. Initialize $\+S$ to be empty. For every $v_i\in \+{DA}_k^{b_l}$, we insert $\Msuf(v_i)$ into $\+S$. We finish constructing $\+{S}$ in $O(\mu_1)$ time.


	Let $\+{S}'$ be an emptyset, we invoke {\sc DisWave}($\suf$, $\sigma_l$, $(2+\epsilon)\delta$, $\+{S}$, $\+{S}'$). Let $\+{DW}^{\bar{v}_a}$ be the output set for the last vertex $\bar{v}_a$ of $\suf$. Set $\+{I}^2 = \+{DW}^{\bar{v}_a}$.  We construct $\+{I}^2$ in $O(\mu_2^2 + \mu_1)$ time. 

	\begin{lemma}~\label{lem: Dis I_2}
    	We can construct $\+I^2$ in $O(\mu_1 + \mu_2^2)$ time such that $\+I^2$ contains all vertices of type 2, and every vertex $w_j \in \+I^2$ covered forms a $((7+\varepsilon)\delta)$-reachable pair with $v_{a_{k+1}}$. 
	\end{lemma}

	\begin{proof}
		The running time follows the construction procedure. We focus on proving $\+I^2$'s property. 

		By definition, if $w_j$ belongs to type 2, there exists a vertex $v_{i'}\in \+{DA}^{b_l}_k$ such that $\tilde{d}_F(\tau[v_{i'}, v_{a_{k+1}}],$ $\sigma[w_{b_l}, w_j]) \leq \delta$.
		Let the last vertex of $\suf$ be $\bar{v}_a$, we have $\Msuf(v_{a_{k+1}}) = \bar{v}_a$.
		According to \Cref{lem: discrete simplification}, $\tilde{d}_F(\tau[v_{i'}, v_{a_{k+1}}], \suf[\Msuf(v_{i'}), \bar{v}_a]) \leq (1+\epsilon)\delta$. Then $\tilde{d}_F(\sigma[w_{b_l}, w_j],\suf[\Msuf(v_{i'}), \bar{v}_a]) \leq (2+\epsilon)\delta$ by triangle inequality. According to the construction of $\+{S}$, $\Msuf(v_{i'})$ is covered by $\+{S}$. Thus, $w_j\in \+{I}^2$ by the definition of {\sc DisWave}.

		Next, we prove that every vertex $w_j\in \+{I}^2$ can form a $((7 + \varepsilon) \delta)$-reachable pair with $v_{a_{k+1}}$. Since $\+{S}'$ is empty, there exists a vertex $p\in \+{S}$ such that $\tilde{d}_F(\sigma[w_{b_l}, w_j],\suf[p, \bar{v}_a]) \leq (2+\epsilon)\delta$ by the definition of {\sc DisWave}. 
		According to \Cref{lem: discrete simplification}, we have $\tilde{d}_F(\tau[\Msuf(p), v_{a_{k+1}}], \zeta_k^{\text{suf}}[p, \bar{v}_a]) \leq (1 + \epsilon)\delta$. By triangle inequality, we have $\tilde{d}_F(\tau[\Msuf(p), v_{a_{k+1}}], \sigma[w_{b_{l}}, w_j]) \leq (3 + 2 \epsilon)\delta < (7+\varepsilon)\delta$. According to the procedure of constructing $\+{S}$, $\Msuf(p)$ is covered by $\+{DA}^{b_l}_{k}$, which means $\tilde{d}_F(\tau[v_1, \Msuf(p)], \sigma[w_1, w_{b_l}]) \leq (7+\varepsilon)\delta$. We can construct a matching between $\tau[v_1, v_{a_{k+1}}]$ and $\sigma[w_1, w_j]$ by concatenating the discrete Fr\'echet matching between $\tau[\Msuf(p), v_{a_{k+1}}]$ and $\sigma[w_{b_l},w_j]$ to the discrete Fr\'echet matching between $\tau[v_1, \Msuf(p)]$ and $\sigma[w_1, w_{b_l}]$. The matching realizes a distance at most $(7+\varepsilon)\delta$. Thus, we have $(v_{a_{k+1}}, w_j)$ is a $(7 + \varepsilon)$-reachable pair.  
	\end{proof}

	\noindent\textbf{Construction of $\+I^3$.} Let $\+{S}$ be an emptyset. We invoke {\sc DisWave}($\pre$, $\sigma_l$, $(2+\epsilon)\delta$, $\+{S}$, $\+{DA}^{a_k}_{l}$). Let $\+{DW}^{w_{b_{l+1}}}$ be the output set for $w_{b_{l+1}}$. For all $i \in [a_{k}, i_{\text{pre}}]$, insert $v_i$ into $\+I^3$ if $\Mpre(v_i)\in \+{DW}^{w_{b_{l+1}}}$. It takes $O(\mu_2^2 + \mu_1)$ time. 

	\begin{lemma}~\label{lem: Dis I_3}
    	We can construct $\+I^3$ in $O(\mu_2^2+\mu_1)$ time such that $\+I^3$ contains all vertices of type 3, and every vertex $v_i$ covered by  $\+I^3$ forms a $((7+\varepsilon)\delta)$-reachable pair with $w_{b_{l+1}}$. 
	\end{lemma}

	\begin{proof}
		The running time follows the construction procedure. We focus on proving $\+I^3$'s property. 

		Take a vertex $v_i$ of type 3, there is a vertex $w_{j'}\in \+{DA}^{a_k}_l$ such that $\tilde{d}_F(\tau[v_{a_k}, v_i], \sigma[w_{j'}, w_{b_{l+1}}]) \leq \delta$. According to \Cref{lem: discrete simplification}, we have $\tilde{d}_F(\tau[v_{a_k}, v_i], \pre[\Mpre(v_{a_k}), \Mpre(v_i)]) \leq (1+\epsilon)\delta$. By triangle inequality, we have $\tilde{d}_F(\zeta_k^{\text{pre}}[\Mpre(v_{a_k}), \Mpre(v_{i})], \sigma[w_{j'}, w_{b_{l+1}}]) \leq (2+\epsilon)\delta$. By the definition of {\sc DisWave}, $\Mpre(v_i)$ must belong to $\+{DW}^{w_{b_{l+1}}}$. Hence, $v_i\in \+I^3$.
		

		Next, we prove that every vertex $v_i\in \+{I}^3$ can form a  $((7+\varepsilon)\delta)$-reachable pair with $w_{b_{l+1}}$. By the contruction procedure, we have $\Mpre(v_i)\in \+{DW}^{w_{b_{l+1}}}$. Since $\+S$ is empty, there exists a vertex $w_{j'}\in \+{DA}^{a_k}_l$ such that $\tilde{d}_{F}(\pre[\Mpre(v_{a_k}), \Mpre(v_i)], \sigma[w_{j'}, w_{b_{l+1}}]) \leq (2+\epsilon)\delta$ by the definition of {\sc DisWave}. According to \Cref{lem: discrete simplification}, we have $\tilde{d}_F(\tau[v_{a_k}, v_i], \pre[\Mpre(v_{a_k}), \Mpre(v_i)]) \leq (1+\epsilon)\delta$. By triangle inequality, we have $\tilde{d}_F(\tau[v_{a_{k}}, v_i], \sigma[w_{j'}, w_{b_{l+1}}]) \leq (3+2\epsilon)\delta < (7+\varepsilon)\delta$. Since $w_{j'}\in \+{DA}^{a_k}_l$, we have $\tilde{d}_F(\tau[v_1, v_{a_k}], \sigma[w_1, w_{j'}]) \leq (7+\varepsilon)\delta$. We can construct a matching between $\tau[v_1, v_i]$ and $\sigma[w_1, w_{b_{l+1}}]$ by concatenating the discrete Fr\'echet matching between $\tau[v_{a_k}, v_{i}]$ and $\sigma[w_{j'}, w_{b_{l+1}}]$ to that between $\tau[v_1, v_{a_k}]$ and $\sigma[w_1, w_{j'}]$. The matching realizes a distance at most $(7+\varepsilon)\delta$. Thus, we have $(v_i, w_{b_{l+1}})$ is a $(7+\varepsilon)$-reachable pair.
	\end{proof}

	\noindent\textbf{Construction of $\+I^4$.} We divide $\sigma_l$ into a collection $(\sigma_{l,1}, \sigma_{l,2},\ldots,\sigma_{l, (\mu_2-1)/\mu_3})$ of subcurves of $\mu_3$ edges such that $\sigma_{l,r}=\sigma[w_{b_{l,r}}, w_{b_{l,r+1}}]$, where $b_{l,r}=b_l+(r-1)\mu_3+1$ for $r\in [(\mu_2-1)/\mu_3]$. For every $\sigma_{l,r}$ and a constant $c>0$ whose value will be specified later, we sample a set of $2c\log n\cdot \mu_1/\omega$ vertices of $\tau_k$ independently with replacement. We then use \Cref{lem: marked-vertex} with $\sigma_{l,r}$ and every sampled vertex. It takes $O(\mu_2^5\mu_1\cdot\log n/(\omega\mu_3))$ time. In the case that we succeed in finding $\tau'_{l,r}$ for all $\sigma_{l,r}$ via sampling, we get $(\tau'_{l,r})_{r\in [(\mu_2-1)/\mu_3]}$ as a surrogate of $\sigma_l$.. 

	Otherwise, take some $\sigma_{l,r}$ for which we fail to get a $\tau'_{l,r}$. It implies that no sampled vertex is marked by $\sigma_{l,r}$. By Lemma~\ref{lem:Chernoff}, $\sigma_{l,r}$ is not $\omega$-dense with probability at least $1-n^{-10}$ for sufficiently large $c$. Conditioned on that $\sigma_{l,r}$ is not $\omega$-dense. Let $\+S$ be a set that contains all $\tau_k$'s vertices $v_i$ with $v_i \in \+{B}(w_{b_{l, r}}, \delta)$.
	Let $\+{S}'$ be an emptyset.
	We invoke {\sc DisWave}$(\tau_k, \sigma_{l,r}, \delta, \+S, \+{S}')$ to get the output set $\+{DW}^{w_{b_{l, r+1}}}$ for the vertex $w_{b_{l, r+1}}$ in $O(\mu_1\mu_3)$ time. 
	By definition, every vertex in $\+{DW}^{w_{b_{l, r+1}}}$ is marked by $\sigma_{l,r}$. Hence, its size is at most $\omega$. Moreover, $\+{DW}^{w_{b_{l, r+1}}}$ must contain a vertex marked by $\sigma_l$ because the discrete Fr\'echet matching between $\sigma_l$ and any subcurve of $\tau_k$ must matches $\sigma_{l,r}$ to some subcurve of $\tau_k$. We run the algorithm in \Cref{lem: marked-vertex} on $\sigma_l$ and every vertex in $\+{DW}^{w_{b_{l, r+1}}}$ to find $\tau'$. It takes $O(\omega\mu_2^4)$ time.

	With probability at least $1-n^{-10}$, we either find a subcurve $\tau'$ of $\tau_k$ with $\tilde{d}_F(\tau', \sigma_l)\le (3+2\epsilon)\delta$ or subcurves $\tau'_{l,r}$ of $\tau_k$ for all $r\in [(\mu_2-1)/\mu_3]$ such that $\tilde{d}_F(\tau'_{l,r}, \sigma_{l,r})\le (3+2\epsilon)\delta$. When we get $\tau'$, we execute a query {\sc DisCover}$(\tau', (4+2\epsilon)\delta, \+{DA}_k^{b_l})$ and assign the answer to $\+I^4$ in $O(\mu_1)$ time. 

	When we get $\tau'_{l,r}$ for all $\sigma_{l,r}$'s, we execute {\sc DisCover} queries with $\tau'_{l,1}, \tau'_{l,2}, \ldots, \tau'_{l, (\mu_2-1)/\mu_3}$ progressively. Specifically, we first carry out a query {\sc DisCover}$(\tau'_{l,1}, (4+2\epsilon)\delta, \+{DA}_k^{wb_l})$ in $O(\mu_1)$ time to get an array $\+S^1$. For any $r\in [2, (\mu_2-1)/\mu_3]$, suppose that we have gotten $\+S^{r-1}$, we proceed to execute a query {\sc DisCover}$(\tau'_{l,r}, (4+2\epsilon)\delta, \+S^{r-1})$ to get $\+S^r$. In the end, we set $\+I^4=\+S^{(\mu_2-1)/\mu_3}$. It takes $O(\mu_1\mu_2/\mu_3)$ time. 

	\begin{lemma}\label{lem: Dis I_4}
    	We can construct $\+I^4$ in  $O(\mu_1(\mu_3+\mu_2/\mu_3+\mu_2^5\log n/(\omega\mu_3))+\omega\mu_2^4)$  time such that $\+I^4$ contains all vertices of type 4 with probability at least $1-n^{-10}$, and every vertex $v_i\in \+I^4$ forms a $((7+\varepsilon)\delta)$-reachable pair with $w_{b_{l+1}}$.
	\end{lemma}

	\begin{proof}
		The running time follows the construction procedure. We focus on proving $\+I^4$'s property. 

		Suppose that there is a vertex $v_i$ of type 4. It implies that $\sigma_l$ is within a discrete Fr\'echet distance $\delta$ of some subcurve of $\tau_k$. Hence, with probability at least $1-n^{-10}$, we can either find a single subcurve $\tau'$ of $\tau_k$ such that $\tilde{d}_F(\tau', \sigma_l)\le (3+2\epsilon)\delta$ or subcurves $\tau'_{l,r}$ of $\tau_k$ for every $r\in [(\mu_2-1)/\mu_3]$ such that $\tilde{d}_F(\tau'_{l,r}, \sigma_{l,r})\le (3+2\epsilon)\delta$.

		When we get a single subcurve $\tau'$ of $\tau_k$ such that $\tilde{d}(\tau', \sigma_l) \leq (3+2\epsilon)\delta$. We execute a query {\sc DisCover}$(\tau', (4+2\epsilon)\delta, \+{DA}_k^{b_l})$ and assign the solution to $\+I^4$. Suppose $v_i$ belongs to type 4, according to the definition, there is a vertex $v_{i'} \in \tau_k$ covered by $\+{DA}^{b_l}_k$ such that $v_{i'} \leq_\tau v_i$ and $\tilde{d}_F(\tau[v_{i'}, v_i], \sigma_l) \leq \delta$. Since $\tilde{d}_F(\tau', \sigma_l) \leq (3 + 2\epsilon)\delta$, we have $\tilde{d}_F(\tau', \tau[v_{i'}, v_i]) \leq (4+2\epsilon)\delta$ by triangle inequality. Thus, $v_i\in \+{I}^4$ by the definition of {\sc DisCover}. 
    
		When we get subcurves $\tau'_{l, r}$ of $\tau_k$ for all $r \in [(\mu_2 - 1) / \mu_3]$, we will construct a sequence $\+S^1, \+S^2,\ldots, $ $\+S^{(\mu_2-1)/\mu_3}$ and set $\+I^4=\+S^{(\mu_2-1)/\mu_3}$.
		We prove that $\+S^r$ contains all vertices $v_y\in \tau_k$ satisfying that there is a vertex $v_x\in \+{DA}^{b_l}_{k}$ with $v_x\le_\tau v_y$ and $d_F(\tau[v_x, v_y], \sigma[w_{b_l}, w_{b_l, r+1}])\le \delta$ by induction on $r$. 

		When $r=1$, the analysis is the same as the analysis in the case where we find $\tau'$. When $r\ge 2$, assume that $\+S^{r-1}$ satisfies the property. Take any vertex $v_y\in \tau_k$ satisfying that there is a vertex $v_x\in \+{DA}^{b_l}_k$ with $v_x\le_\tau v_y$ and $\tilde{d}_F(\tau[v_x, v_y], \sigma[w_{b_l}, w_{b_l, r+1}])\le \delta$.
		The discrete Fr\'echet matching between $\tau[v_x, v_y]$ and $\sigma[w_{b_l}, w_{b_l, r+1}]$ must matching the vertex $w_{b_l, r}$ to some vertex $v_z\in \tau[v_x, v_y]$. It implies that $\tilde{d}_F(\tau[v_x,v_z], \sigma[w_{b_l}, w_{b_l,r}])\le \delta$ and $\tilde{d}_F(\tau[v_z, v_y], \sigma[w_{b_l, r}, w_{b_l, r+1}])\le \delta$. By the induction hypothesis, the vertex $v_z\in \+S^{r-1}$. By the construction procedure, $\+S^r$ is the answer for {\sc DisCover}$(\tau'_{l,r}, (4+2\epsilon)\delta, \+S^{r-1})$. Since $\sigma_{l,r}=\sigma[w_{b_l,r}, w_{b_l, r+1}]$ is within a discrete Fr\'echet distance $(3+2\epsilon)\delta$ to $\tau'_{l,r}$, by the triangle inequality, $\tilde{d}_F(\tau'_{l,r}, \sigma_{l,r})\le (4+2\epsilon)\delta$. Hence, $v_y$ must belong to $\+S^r$ by the definition of the query {\sc DisCover} as $v_z$ is covered by $\+S^{r-1}$. We finish proving the property for all $\+S^r$'s. Since $\+I^4=\+S^{(\mu_2-1)/\mu_3}$ and $w_{b_l, (\mu_2-1)/\mu_3+1}=w_{b_{l+1}}$, all vertices belongs to type 4 are contained by $\+I^4$.

		Next, we prove that all points covered by $\+I^4$ can form $((7+\varepsilon)\delta)$-reachable pairs with $w_{b_{l+1}}$.
		In the case that we find $\tau'$, $\+I^4$ is the answer of the query {\sc DisCover}$(\tau', (4+2\epsilon)\delta, \+{DA}^{b_l}_{k})$.
		For all vertex $v_i\in \+{I}^4$, there is a vertex $v_{i'} \leq_\tau v_i$ covered by $\+{DA}_{k}^{b_l}$ such that $\tilde{d}_F(\tau', \tau[v_{i'}, v_i]) \leq (4+2\epsilon)\delta$. Since $\epsilon = \varepsilon/10$, by triangle inequality, we have $\tilde{d}_F(\tau[v_{i'}, v_i], \sigma_l) \leq (7+4\epsilon)\delta < (7+\varepsilon)\delta$.
		Since $v_{i'}\in \+{DA}_{k}^{b_l}$, we have $\tilde{d}_F(\tau[v_1, v_{i'}], \sigma[w_1, w_{b_l}]) \leq (7+\varepsilon)\delta$. We can construct a matching between $\tau[v_1, v_i]$ and $\sigma[w_1, w_{b_{l+1}}]$ by concatenating the discrete Fr\'echet matching between $\tau[v_{i'}, v_i]$ and $\sigma_l$ to the discrete Fr\'echet matching between $\tau[v_1, v_{i'}]$ and $\sigma[w_1, w_{b_l}]$. The matching realizes a distance at most $(7+\varepsilon)\delta$. Thus, we have $(v_{i}, w_{b_{l+1}})$ is a $((7 + \varepsilon)\delta)$-reachable pair.

		When we get $\tau'_{l,1}, \tau'_{l,2}, \ldots, \tau'_{l, (\mu_2-1)/\mu_3}$, we will construct a sequence $\+S^1, \+S^2,\ldots, \+S^{(\mu_2-1)/\mu_3}$ and set $\+I_4=\+S^{(\mu_2-1)/\mu_3}$. We prove that  all vertices $v_y$ in $\+S^r$ can form $((7+\varepsilon)\delta)$-reachable pairs with $w_{b_l, r+1}$ by induction on $r$.

		When $r=1$, the analysis is the same as the analysis in the case where we find $\tau'$. When $r\ge 2$, assume that $\+S^{r-1}$ satisfies the property. Take any vertex $v_y\in \+S^r$. Since $\+S^r$ is the answer for {\sc DisCover}$(\tau'_{l,r},(4+2\epsilon)\delta,\+S^{r-1})$, by the definition of the query {\sc DisCover}, we can find a vertex $v_x\in \+S^{r-1}$ such that $v_x\le_\tau v_y$ and $\tilde{d}_F(\tau[v_x,v_y], \tau'_{l,r}) \le (4+2\epsilon)\delta$. Given that $\tilde{d}_F(\tau'_{l,r}, \sigma_{l,r})\le (3+2\epsilon)\delta$, we have $\tilde{d}_F(\tau[v_x,v_y], \sigma_{l,r})\le (7+4\epsilon)\delta$ by the triangle inequality. Note that $\epsilon=\varepsilon/10$. Hence, $\tilde{d}_F(\tau[v_x,v_y], \sigma_{l,r})\le (7+\varepsilon)\delta$. By the induction hypothesis, it holds that $\tilde{d}_F(\tau[v_1, v_x], \sigma[w_1, w_{b_l,r}])\le (7+\varepsilon)\delta$. Hence, we can construct a matching between $\tau[v_1, v_y]$ and $\sigma[w_1, w_{b_l, r+1}]$ by concatenating the discrete Fr\'echet matching between $\tau[v_1, v_x]$ and $\sigma[w_1, w_{b_l, r}]$ and the discrete Fr\'echet matching between $\tau[v_x,v_y]$ and $\sigma_{l,r}$. The matching realizes a distance at most $(7+\varepsilon)\delta$. Therefore, $\tilde{d}_F(\tau[v_1, v_y], \sigma[w_1, w_{b_l, r+1}])\le (7+\varepsilon)\delta$. We finish proving the property for all $\+S^r$'s.  Since $\+I^4=\+S^{(\mu_2-1)/\mu_3}$ and $w_{b_l, (\mu_2-1)/\mu_3+1}=w_{b_{l+1}}$, all vertices in $\+I^4$ can form $((7+\varepsilon)\delta)$-reachable pairs with $w_{b_{l+1}}$.
	\end{proof}

	Finally, we can set $\+{DA}^{a_{k+1}}_{l}$ to the union of $\+I^1$ and $\+I^2$, and set $\+{DA}_{k}^{b_{l+1}}$ to be the union of $\+I^3$ and $\+I^4$.
	It takes $O(\mu_1)$ time in total. We finish implementing {\sc DisReach}. 
	
	Putting Lemma~\ref{lem: marked-vertex},~\ref{lem:dis-cover},~\ref{lem: Dis I_1},~\ref{lem: Dis I_2},~\ref{lem: Dis I_3} and~\ref{lem: Dis I_4} together, we have the following lemma. We treat both $\varepsilon$ and $d$ as fixed constants.

	\begin{lemma}\label{lem: Disreach}
    	There is an algorithm that preprocesses every $\tau_k$ in $O(\mu_1^5)$time to implement the procedure {\sc DisReach} in $O(\mu_1(\mu_3+\mu_2/\mu_3+\mu_2^5\log n/(\omega\mu_3))+\omega\mu_2^4)$ time with success probability at least $1-n^{-10}$.
	\end{lemma}

	We can set $\mu_1=m^{0.24}$, $\mu_2=m^{0.02}$, $\mu_3=m^{0.01}$, and $\omega=m^{0.12}$ to get the running time of {\sc DisReach} to be $o(\mu_1\mu_2)$.
	
	\noindent {\bf Dynamic programming for using \pmb{\sc DisReach.}} We show how {\sc DisReach} helps us realize a subquadratic decision procedure.
	We first compute all reachable pairs of $v_1$ and $w_1$ in $O(n+m)$ time. We then generate $\+{DA}^{1}_{l}$ and $\+{DA}^{1}_{k}$. For every $l\in [(m-1)/\mu_2]$ and every $j\in [b_l, b_{l+1}-1]$, set $\+{DA}^{1}_{w_j} = \+{DA}^{1}_{w_j}$. It takes $O(m)$ time in total. We can generate $\+{DA}^{1}_{k}$ for $w_1$ and all $k\in [(n-1)/\mu_1]$ in $O(n)$ time in the same way.

	We are now ready to invoke {\sc DisReach}$(\tau_1, \sigma_1, \+{DA}^{a_1}_{1}, \+{DA}_{1}^{b_1})$ to get $\+{DA}^{a_2}_{1}$ and $\+{DA}_{1}^{b_2}$. We proceed to invoke the procedure {\sc DisReach}$(\tau_2, \sigma_1, \+{DA}^{a_2}_{1}, \+{DA}_{2}^{b_1})$.  We can repeat the process for all $k\in [(n-1)/\mu_1]$ to get $\+{DA}_{k}^{b_2}$ of $w_{b_2}$ for all $k\in [(n-1)/\mu_1]$.

	We are now ready to invoke {\sc Reach}$(\tau_1, \sigma_2, \+{DA}^{a_1}_{2}, \+{DA}_{1}^{b_2})$. We can repeat the process above to get $\+{DA}_{k}^{b_3}$ of $w_{b_3}$ for all $k\in [(n-1)/\mu_1]$. In this way, we can finally get a $(7+\varepsilon)$-approximate reachability set for $w_m$ after calling {\sc DisReach} for $mn/(\mu_1\mu_2)$ times. Finally, we finish the decision by checking whether $\+{DA}^{m}_{v_n} = 1$ in $O(1)$ time. If so, return yes; otherwise, return no.

	As shown above, we can call {\sc DisReach} for $mn/(\mu_1\mu_2)$ times to get a decision procedure. It succeeds if all these invocations of {\sc DisReach} succeed. The total running time is $O(nm^{0.99})$.

	\begin{theorem}\label{thm:Dis Frechet}
    	Given two polygonal curves $\tau$ and $\sigma$ in $\mathbb{R}^d$ for some fixed $d$, there is a randomized $(7+\varepsilon)$-approximate decision procedure for determining $\tilde{d}_F(\tau, \sigma)$ in $O(nm^{0.99})$ time with success probability as least $1-n^{-7}$. 
	\end{theorem}

We use the approach developed in~\cite{bringmann2016approximability} to get an approximation algorithm. It gets an $\alpha$-approximate algorithm for computing $\tilde{d}_F(\tau, \sigma)$ by carrying out $O(\log n)$ instances of any $\alpha$-approximate decision procedure. In our case, the approximate algorithm succeeds if all these instances succeed.

\begin{theorem}\label{thm:approx_Dis_Frechet}
    Given two sequences $\tau$ and $\sigma$ in $\mathbb{R}^d$ for some fixed $d$, there is a randomized $(7+\varepsilon)$-approximate algorithm for computing $\tilde{d}_F(\tau, \sigma)$ in $O(nm^{0.99}\log n)$ time with success probability at least $1-n^{-6}$.
\end{theorem}
	
	\section{Discussion and future works}
	We present the first constant approximation algorithms for the Fr\'echet distance and the discrete Fr\'echet distance that run in strongly subquadratic time. Our algorithms are randomized and achieve approximation ratios of $(7+\varepsilon)$. While we follow a conventional approach that divides the input curves $\tau$ and $\sigma$ into short subcurves $\tau_k$'s and $\sigma_l$'s and accelerates the propagation of reachablity information involving every pair of $\tau_k$ and $\sigma_l$, we present a new framework for achieving the speedup. Basically, we classify the points in reachablity intervals on $\tau_k$ and $\sigma_l$ into four types and present fast algorithms to cover them approximately and fast. While the first three types can be easily dealt with by curve simplification, the type 4 requires new ideas. The handling of type 4 is the bottlenecks of both approximation ratio and the running time. The root of the ratio $7+\varepsilon$ is the ratios of $(3+2\epsilon)$ in Lemma~\ref{lem: marked-edge} and~\ref{lem: marked-vertex}. Currently, if we have to use a subcurve $\tau'$ of $\tau_k$ as the surrogate of some subcurve $\sigma'$ of $\sigma_l$, we can only guarantee an upper bound of $(3+2\epsilon)\delta$ on the (discrete) Fr\'echet distance between $\tau'$ and $\sigma'$. It causes an approximation ratio of $(7+O(\epsilon))$ inevitably due to the triangle inequality. An improvement of Lemma~\ref{lem: marked-edge} and~\ref{lem: marked-vertex} will lead to an improvement of approximation ratio directly. As for the running time, though the current setting of parameters provides us with an exponent $0.99$, which is not best possible, it is unlikely that the running time can be improved significantly without a more efficient handling of type 4. Because of the high time compleixty in finding the surrogates of subcurves of $\sigma_l$ and answering the query {\sc (Dis)Wave} (Lemma~\ref{lem: marked-edge},~\ref{lem:cover},~\ref{lem: marked-vertex}, and~\ref{lem:dis-cover}), the current setting of parameters is near the optimal.
	
	Improving the approximation performance and the running time can be a natural problem in the near furture. Since our algorithms are randomized, whether there is a \emph{deterministic} constant approximation algorithm in strongly subquadratic time for the (discrete) Fr\'echet distance is sill a big open problem.

	\appendix

		\section{Data structure for answering the query {\sc Cover}}\label{sec:cover}
		
		We restate the query {\bf{\sc Cover}} with a subcurve $\tau'$ of $\tau_k$, an array $\+S$ induced by $\tau_k$, and $\delta'>0$. The start and end of $\tau'$ may not be vertices. The value $\epsilon\in(0,1)$ is known during preprocessing.
		
		\vspace{2pt}
		
		\noindent\pmb{{\sc Cover}$(\tau', \delta',\+S)$.} The answer is another array $\bar{\+S}$ induced by $\tau_k$ satisfing that:
		
		\begin{itemize}
			\item for any point $x\in\tau_k$, if there is a point $y$ covered by $\+S$ such that $y\le_{\tau} x$ and $d_F(\tau[y,x], \tau')\le \delta'$, then $x$ is covered by $\bar{\+S}$;
		\item for any point $x$ covered by $\bar{\+S}$, there is a point $y$ covered by $\+S$ such that $y\le_{\tau} x$ and $d_F(\tau[y,x], \tau')\le (1+\epsilon)\delta'$.
		\end{itemize}

		The underlying idea of the data structure is to use the planarity of the free space diagram, especially the observation in Lemma~\ref{lem:planarity}. Suppose that $\tau'$ is a vertex-to-vertex subcurve of $\tau_k$ from $v_{i_1}$ to $v_{i_2}$. For every edge $v_iv_{i+1}$ of $\tau_k$, let $\ell_i$ be the start of $\+S_{v_i}\cap \+B(v_{i_1},\delta')$ if $\+S_{v_i}$ is not null. Consider the free space diagram induced by $\tau'$ and $\tau_k$. The first useful observation is that for any point $y\in \tau_k$, if there is a point $x\in \+S_{v_i}$ such that $(y, v_{i_2})$ is $\delta'$-reachable from $(x, v_{i_1})$, then $(y, v_{i_2})$ is $\delta'$-reachable from $(\ell_i, v_{i_1})$ as well. Because $(x, v_{i_1})$ is $\delta'$-reachable from $(\ell_i, v_{i_1})$. It means that it is sufficient to store all points $y\in \tau_k$ such that $(y, v_{i_2})$ is $\delta'$-reachable from some $(\ell_i, v_{i_1})$ to answer the query {\sc Cover}$(\tau', \delta', \+S)$. Moreover, we can answer the query in $O(\mu_1)$ time due to Lemma~\ref{lem:planarity}. 

		Suppose that for every $\ell_i$, we have stored an array $\+R^i$ induced by $\tau_k$ that covers all points $y$ satisfying that $(y, v_{i_2})$ is $\delta'$-reachable from $(\ell_i, v_{i_1})$. The solution to {\sc Cover} is the union of points covered by $\+R^i$ for all $\ell_i$. We can construct such a union greedily. Basically, assume that we have already constructed the array $\+S'$, which is the union of $\+R^1-\+R^{i-1}$, and we know the maximum integer $i^*$ such that $\+S'_{v_{i^*}}$ is not empty. To merge $\+R^i$, we only need to replace $\+S'_{v_{i^*+1}}, \+S'_{v_{i^*+2}},\ldots, \+S'_{v_{a_{k+1}}-1}$ by $\+R^i_{v_{i^*+1}}, \+R^i_{v_{i^*+2}},\ldots, \+R^i_{v_{a_{k+1}}-1}$. The reason is as follows. By Lemma~\ref{lem:planarity}, if some point $y\in v_{i'}v_{i'+1}$ with $i'\le i^*$ is covered by $\+R^i$, then $y$ must be covered by some $\+R^{\bar{i}}$ with $\bar{i}<i$ and therefore covered by $\+S'$. We can use the above greedy algorithm to answer {\sc Cover}$(\tau', \delta', \+S)$ in $O(\mu_1)$ time. 
		
		But we do not know $\ell_i$ during preprocessing given that $\+S$ is arbitrary, and the curve $\tau'$ may not be a vertex-to-vertex subcurve of $\tau_k$. We use discretization and a decomposition of $\tau'$ to resolve the issues. The discretization is the reason for introducing approximation in the definition of {\sc Cover}.

		\vspace{4pt}

		\noindent\underline{\bf Data structure.} 
		Our data structure consists of arrays induced by $\tau_k$ for every vertex-to-vertex subcurve of $\tau_k$.
		For every subcurve $\tau[v_{i_1}, v_{i_2}]$ and every edge $v_{i}v_{i+1}$ of $\tau_k$. Test whether $\+B(v_i, \delta')\cap v_{i}v_{i+1}\not=\emptyset$. If it is not empty, we discretize it to get a sequence of points $(p_1, p_2,\ldots, p_a)$ such that $p_1$ is the start of intersection, $d(p_1, p_{a'})=(a'-1)\epsilon\delta'$ for all $a'\in[a]$, and $a=\lfloor\frac{L}{\epsilon\delta'}\rfloor+1$, where $L$ is the length of the intersection. Note that $L$ is at most $2\delta'$, which implies that $a\le 2/\epsilon+1$. For an arbitrary array $\+S$, the discretization allows us to approximate $\ell_i$, which is the start of $\+S_{v_i}\cap \+B(v_{i_1}, \delta')$. 
		
		For every point $p_b$ in the sequence derived from the discretization of $\+B(v_{i_1}, \delta')\cap v_iv_{i+1}$, we store all points $y$ with $(y, v_{i_2})$ being $\delta'$-reachable from $(p_b, v_{i_1})$. Specially, we generate an array $\+A$ induced by $\tau_k$ such that $\+A_{v_i}=p_b$ and all the other elements in $\+A_{v_i}$ are empty. We also take an array $\+A'$ induced by $\tau[v_{i_1}, v_{i_2}]$ such that all elements in $\+A'$ are empty. We then invoke {\sc WaveFront}$(\tau_k, \tau[v_{i_1}, v_{i_2}], \delta', \+A, \+A')$. Take the output array $\+{W}^{v_{i_2}}$ for the vertex $v_{i_2}$. By the definition of {\sc WaveFront}, a point $y$ is covered by $\+{W}^{v_{i_2}}$ if and only if $(y, v_{i_2})$ is $\delta'$-reachable from $(x, v_{i_1})$. We then store $\+{W}^{v_{i_2}}$ in our data structure, which is a 4D array $\+D$ such that the element $\+D[i_1, i_2, i, b]$ stores $\+{W}^{v_{i_2}}$. In the case where $\+B(v_{i_1}, \delta')\cap v_{i}v_{i+1}$ is empty and $p_b$ is not well-defined, we set $\+D[i_1, i_2, i, b]$ to be null. 

		We also maintain a 4D array {\sc Max}. Specifically, let $\bar{\+A}$ denote the array stored in $\+D[i_1, i_2, i, b]$ for ease of notation. Take the maximum $i^*\in [a_k, a_{k+1}]$ such that $\bar{\+A}_{v_{i^*}}$ is not empty, we set {\sc Max}$[i_1, i_2, i, b]$ to be $i^*$. We set {\sc Max}$[i, i_1, i',\alpha]$ to be $-1$ if such an $i^*$ does not exist.
  

		For every subcurve $\tau[v_{i_1}, v_{i_2}]$ and every edge $v_{i}v_{i+1}$, it takes $O(\mu_1\cdot(i_2-i_1+1))=O(\mu_1^2)$ time to determine $\+D[i_1, i_2, i,b]$ and {\sc Max}$[i_1, i_2, i,b]$ for every $b\in [2/\epsilon+1]$. Given that there are $O(\mu_1^3)$ distinct combinations of $i_1, i_2$ and $i$, it takes $O(\mu_1^5/\epsilon)$ time to construct $\+D$ and {\sc Max}. 

		\vspace{2pt}
		
		\noindent\underline{\bf Query algorithm.} Given an arbitrary subcurve $\tau'$ of $\tau_k$ and an arbitrary array $\+S$ induced by $\tau_k$, we present how to answer {\sc Cover}$(\tau',\delta',\+S)$ in $O(\mu_1)$ time based on the arrays $\+D$ and {\sc Max}.
		
		Suppose that $\tau'=\tau[x,y]$. Note that $x$ and $y$ may not be vertices of $\tau_k$. In the case where $\tau[x,y]$ has at most 2 edges, 
		let $\+S'$ be an array induced by $\tau[x, y]$ with all elements being empty. We call {\sc WaveFront}$(\tau_k, \tau[x,y], \delta', \+S, \+S')$. Let $\bar{\+S}$ be the output array for $y$. By the definition of {\sc WaveFront}, $\bar{\+S}$ is a feasible answer for the query. It takes $O(\mu_1)$ time.
		
		
		When $\tau[x,y]$ has more than 2 edges, it must contain at least 2 vertices of $\tau_k$. That is, $\tau[x,y]=(x, v_{i_1}, v_{i_1+1},\ldots, v_{i_2}, y)$ with $i_2>i_1$. We process $xv_{i_1}$, the vertex-to-vertex subcurve $\tau[v_{i_1}, v_{i_2}]$, and $v_{i_2}y$ progressively. Let $\+S'$ be an array induced by $xv_i$ that contains an empty set. We first call {\sc WaveFront}$(\tau_k, xv_{i_1}, \delta', \+S, \+S')$ to get the output array $\+{W}^{v_{i_1}}$ for $v_{i_1}$. Set $\+S^1=\+{W}^{v_{i_1}}$. By the definition of {\sc WaveFront}, any point $y'$ in $\tau_k$ is covered by $\+S^1$ if and only if there is a point $x'$ covered by $\+S$ such that $(y', v_{i_1})$ is $\delta'$-reachable from $(x',x)$. 
		
		
		We then deal with $\tau[v_{i_1}, v_{i_2}]$ to construct another intermediate array $\+S^2$ based on $\+S^1$, $\+D$ and {\sc Max}. 
		We first process $\+S^1$ to facilitate the use of $\+D$. Initialize a new array $\+{NS}$ induced by $\tau_k$. For every edge $v_{i}v_{i+1}$ of $\tau_k$, check $\+S^1_{v_{i}}\cap\+B(v_{i_1}, \delta')$. If the intersection is empty, set $\+{NS}_{v_{i}}$ to be empty; otherwise, we access the sequence $(p_1, p_2,\ldots, p_a)$ derived from the discretization of $\+B(v_{i_1}, \delta')\cap v_{i}v_{i+1}$ to find the maximum $b$ such that $p_{b}v_{i+1}$ includes $\+S^1_{v_{i}}\cap \+B(v_{i_1}, \delta')$. Note that the distance between $p_{b}$ and the start of $\+S^1_{v_{i}}\cap\+B(v_{i_1}, \delta')$ is at most $\epsilon\delta'$. We set $\+{NS}_{v_{i}}$ to be $p_{b}v_{i+1}\cap \+B(v_{i_1},\delta')$.

        We construct $\+S^2$ such that a point $y'\in \tau_k$ is covered by $\+S^2$ if and only if there is a point $x'$ covered by $\+{NS}$ such that $(y', v_{i_1})$ is $\delta'$-reachable from $(x', v_{i_1})$. Note that a point $y'$ is covered by $\+S^2$ if and only if there is a pair $(i, b)$ such that $\+{NS}_{v_{i}}=p_b v_{i'+1}\cap \+B(v_i, \delta')$ and $y'$ is covered by $\+D[i_1, i_2, i, b]$.  The construction of $\+S^2$ follows the greedy method before the data structure construction.
        
		
        Initialize all elements in $\+S^2$ to be empty. We traverse $\+{NS}$ to update $\+S^2$ progressively. We use $i$ to index the current element in $\+{NS}$ within the traversal, and use the $i'$ to index the element being updated in $\+S^2$. Initialize both $i$ and $i'$ to be $a_k$. For $\+{NS}_{v_{i}}$, if it is empty, we increase $i$ by 1. If $\+{NS}_{v_{i}}$ is not empty and $i'<i$, repeat setting $\+S^2_{v_{i'}}$ to be empty and increasing $i'$ by 1 until $i'$ equals to $i$. Let $p_b$ be the start of $\+{NS}_{v_{i}}$. Let $\bar{\+A}$ denote the array $\+D[i_1, i_2, i, b]$ for ease of notation. If {\sc Max}$[i_1, i_2, i,b]$ is equal to $-1$, it means that no point in $\tau_k$ is covered by $\bar{\+A}$. We increase $i$ by 1. Otherwise, some point(s) is covered by $\bar{\+A}$. There are two cases. If {\sc Max}$[i_1, i_2, i, b]\le i'$, by Lemma~\ref{lem:planarity}, all points covered by $\bar{\+A}$ are covered by some $\+D[i_1, i_2, i'', b']$ with $i''<i$, and $\bar{\+A}$ can not introduce any new point to $\+S^2$. We increase $i$ by one. Otherwise, we repeat setting $\+S^2_{v_{i'}}$ to be $\bar{\+A}_{v_{i'}}$ and increasing $i'$ by 1 until $i'$ equals to {\sc Max}$[i_1, i_2, i, b]+1$. We stop updating $\+S^2$ until $i'$ becomes $a_{k+1}$ or $i=a_{k+1}+1$. Since every element in $\+S^2$ is updated once and each update takes $O(1)$ time. We finish constructing $\+S^2$ in $O(\mu_1)$ time.
		
		
		In the end, we invoke {\sc WaveFront}$(\tau_k, v_{i_2}y, \delta', \+S^2, \+S')$ in $O(\mu_1)$ time, where $\+S'$ is an array induced by $v_{i_2}y$ that contains an empty set. Take the output array $\+{W}^y$ for $y$. We set $\bar{\+S}=\+{W}^y$. The array $\bar{\+S}$ is a feasible answer for the query {\sc Cover}$(\tau', \delta',\+S)$ via the following analysis. 
		
		\vspace{4pt}
		
		\noindent\textbf{Lemma~\ref{lem:cover}.}~\emph{
				Fix $\tau_k$ and $\delta'$. For any $\epsilon\in (0,1)$, there is a data structure of size $O(\mu_1^4/\epsilon)$ and preprocessing time $O(\mu_1^5/\epsilon)$ that answers in $O(\mu_1)$ time the query {\sc Cover} for any $\tau'$ and $\+S$.}
		
		\begin{proof}
			 The preprocessing time, data structure size and the query time follows the above procedure. We focus on proving the feasibility of $\bar{\+S}$.
			
			When $\tau'$ has at most 2 edges. The answer $\bar{\+S}$ is  returned by {\sc WaveFront}$(\tau_k, \tau', \delta', \+S, \+S')$ for the last vertex of $\tau'$, where $\+S$ is an array induced by $\tau'$ with all elements being empty. By the definition of {\sc WaveFront}, a point $y'$ is covered by $\bar{\+S}$ if and only if there is a point $x'$ covered by $\+S$ with $x'\le_\tau y'$ and $d_F(\tau[x',y'], \tau')\le \delta'$ as all elements in $\+S'$ are empty. It establishes that $\bar{\+S}$ is a feasible solution.
			
			The rest of the proof focuses on the case where $\tau'$ has more than 2 edges. As discussed in the query algorithm, we can divide $\tau'$ into 3 parts: a line segment $xv_{i_1}$, a vertex-to-vertex subcurve $\tau[v_{i_1}, v_{i_2}]$ of $\tau_k$, and a line segment $v_{i_2}y$. Take any point $y'$ for which we can find a point $x'$ covered by $\+S$ such that $x'\le_\tau y'$ and $d_F(\tau[x', y'], \tau')\le \delta'$. Let $F$ be the Fr\'echet matching between $\tau[x', y']$ and $\tau'$. We prove step by step that $F(v_{i_1})$ is covered by $\+S^1$, $F(v_{i_1})$ is covered by $\+S^2$, and $x'$ is covered by $\bar{\+S}$.
			
			Given that $d_F(\tau[x', F(v_{i_1})], xv_{i_1})\le \delta'$, $F(v_{i_1})$ must be covered by $\+S^1$ by the definition of {\sc Cover}. Suppose that $F(v_{i_1})$ belongs to the edge $v_{i}v_{i+1}$, and $F(v_{i_2})$ belongs to the edge $v_{i'}v_{i'+1}$. Note that $d(F(v_{i_1}), v_{i_1})\le \delta'$. Within the construction of $\+S^2$, we have $\+S^1_{v_{i}}\cap \+B(v_{i_1}, \delta')\not=\emptyset$. It implies that $F(v_{i_1})$ belongs to $\+{NS}_{v_{i}}$. Since a point $y''\in \tau_k$ is covered by $\+S^2$ if and only if there is a point $x''$ covered by $\+{NS}$ such that $x''\le_\tau y''$ and $d_F(\tau[x'', y''], \tau[v_{i_1}, v_{i_2}])\le \delta'$. Hence, $F(v_{i_2})$ belongs to $\+S^2_{v_{i'}}$.
   
			
            Given that $d_F(\tau[F(v_{i_2}), y'], v_{i_2}y)\le \delta'$ and $\bar{\+S}$ is returned by {\sc WaveFront}$(\tau_k, v_{i_2}y, \delta', \+S^2,\+S')$, $y'$ must be covered by $\bar{\+S}$. We have finished the first part of the proof that a point $y'\in \tau_k$ must be covered by $\bar{\+S}$ if there is a point $x'$ covered by $\+S$ with $x'\le_\tau y'$ and $d_F(\tau[x',y'], \tau')\le \delta'$.
			
			Next, we prove that for every point $y'$ covered by $\bar{\+S}$, there is a point $x'$ covered by $\+S$ with $x'\le_\tau y'$ and $d_F(\tau[x', y'], \tau')\le (1+\epsilon)\delta$. By the definition of {\sc Cover}, there is a point $z$ covered by $\+S^2$ such that $z\le_\tau y'$ and $d_F(\tau[z, y'], v_{i_2}y)\le \delta'$. 
			
			We proceed to prove that all points $z$ covered by $\+S^2$ satisfy that there is a point $z'$ covered by $\+S^1$ such that $z'\le_\tau z$ and $d_F(\tau[z', z], \tau[v_{i_1}, v_{i_2}])\le (1+\epsilon)\delta$. Suppose that $z$ belongs to the edge $v_{i}v_{i+1}$. By the construction of $\+S^2$, supposed that $\+S^2_{v_{i}}$ is updated by $\+{NS}_{v_{i'}}$. Let $p_{b}$ be the start of $\+{NS}_{v_{i'}}$. It satisfies that $\+S^1_{v_{i'}}\cap \+B(v_{i_1}, \delta')$ is not empty. Let $\ell$ be the start of this intersection. We also have $p_{b}\le_\tau \ell$ and $d(p_{b}, \ell)\le \epsilon\delta$. By definition, the point $z$ satisfies that there is a point $z''\in \+{NR}_{v_{i'}}$ with $z''\le_\tau z$ and $d_F(\tau[z'', z], \tau[v_{i_1}, v_{i_2}])\le \delta'$. Since we can further extend the Fr\'echet matching between $\tau[z'', z]$ and $\tau[v_{i_1}, v_{i_2}]$ to a matching between $\tau[p_{b}, z]$ and $\tau[v_{i_1}, v_{i_2}]$ by matching the line segment $p_{b}z''$ to $v_i$ and the line segment $p_{b}z''$ is inside the ball $\+B(v_i, \delta')$, the matching realizes a distance at most $\delta'$. Hence, $d_F(\tau[p_{b}, z], \tau[v_{i_1}, v_{i_2}])\le \delta'$. Provided that $d_F(\tau[p_{b}, z], \tau[\ell, z])\le \epsilon\delta$, we have $d_F(\tau[\ell, z], \tau[v_{i_1}, v_{i_2}])\le (1+\epsilon)\delta'$ by the triangle inequality. Note that $\ell$ is covered by $\+S^1$. This finishes proving that all points $z$ covered by $\+S^2$ satisfy that there is a point $z'$ covered by $\+S^1$ such that $z'\le_\tau z$ and $d_F(\tau[z', z], \tau[v_{i_1}, v_{i_2}])\le (1+\epsilon)\delta$.
			
			By the above analysis, we know that all points $y'$ covered by $\bar{\+S}$ satisfy that there is a point $z'$ covered by $\+S^1$ such that $z'\le_\tau y'$ and $d_F(\tau[z', y'], \tau[v_{i_1}, y])\le (1+\epsilon)\delta$. Since $z'$ is covered by $\+S^1$ and satisfies that there is point $x'$ covered by $\+S$ with $x'\le_\tau z'$ and $d_F(\tau[x', z'], xv_{i_1})\le \delta'$, we have $d_F(\tau[x', y'], \tau')\le (1+\epsilon)\delta$. We finish proving that $\bar{\+S}$ satisfies the second constraint. This establishes the feasibility of $\bar{\+S}$.
		\end{proof}

	\bibliographystyle{alpha}
	\bibliography{ref}

\newcommand{\etalchar}[1]{$^{#1}$}
\begin{thebibliography}{vdHvKOS23}

\bibitem[AAKS13]{AAKS2013}
Pankaj~K. Agarwal, Rinat~Ben Avraham, Haim Kaplan, and Micha Sharir.
\newblock Computing the discrete {F}r\'{e}chet distance in subquadratic time.
\newblock In {\em Proceedings of the ACM-SIAM Symposium on Discrete
  Algorithms}, pages 156--167, 2013.

\bibitem[AB18]{abboud_et_al:LIPIcs.ICALP.2018.8}
Amir Abboud and Karl Bringmann.
\newblock {Tighter Connections Between Formula-SAT and Shaving Logs}.
\newblock In {\em 45th International Colloquium on Automata, Languages, and
  Programming (ICALP 2018)}, volume 107, pages 8:1--8:18, 2018.

\bibitem[AG95]{AG1995}
Helmut Alt and Michael Godau.
\newblock Computing the {F}r\'{e}chet distance between two polygonal curves.
\newblock {\em International Journal of Computational Geometry and
  Applications}, 5:75--91, 1995.

\bibitem[AHPK{\etalchar{+}}06]{aronov2006frechet}
Boris Aronov, Sariel Har-Peled, Christian Knauer, Yusu Wang, and Carola Wenk.
\newblock Fr{\'e}chet distance for curves, revisited.
\newblock In {\em Proceedings of the 14th Annual European Symposium on
  Algorithms}, pages 52--63. Springer, 2006.

\bibitem[AHPMW05]{agarwal2005near}
P.K. Agarwal, S.~Har-Peled, N.H. Mustafa, and Y.~Wang.
\newblock Near-linear time approximation algorithms for curve simplification.
\newblock {\em Algorithmica}, 42(3):203--219, 2005.

\bibitem[AKW04]{AKW2004}
Helmut Alt, Christian Knauer, and Carola Wenk.
\newblock Comparsion of distance measures for planar curves.
\newblock {\em Algorithmica}, 38:45--58, 2004.

\bibitem[Alt09]{Alt2009}
Helmut Alt.
\newblock {\em The Computational Geometry of Comparing Shapes}, pages 235--248.
\newblock Springer Berlin Heidelberg, Berlin, Heidelberg, 2009.

\bibitem[BBG{\etalchar{+}}08]{Buchin2008DetectingCP}
Kevin Buchin, Maike Buchin, Joachim Gudmundsson, Maarten L{\"o}ffler, and Jun
  Luo.
\newblock Detecting commuting patterns by clustering subtrajectories.
\newblock {\em Int. J. Comput. Geom. Appl.}, 21:253--282, 2008.

\bibitem[BBG{\etalchar{+}}20]{10.1145/3423334.3431451}
Kevin Buchin, Maike Buchin, Joachim Gudmundsson, Jorren Hendriks,
  Erfan~Hosseini Sereshgi, Vera Sacrist\'{a}n, Rodrigo~I. Silveira, Jorrick
  Sleijster, Frank Staals, and Carola Wenk.
\newblock Improved map construction using subtrajectory clustering.
\newblock In {\em Proceedings of the 4th ACM SIGSPATIAL Workshop on
  Location-Based Recommendations, Geosocial Networks, and Geoadvertising},
  2020.

\bibitem[BBMM14]{buchin2014four}
Kevin Buchin, Maike Buchin, Wouter Meulemans, and Wolfgang Mulzer.
\newblock Four soviets walk the dog—with an application to {A}lt's
  conjecture.
\newblock In {\em Proceedings of the Annual ACM-SIAM Symposium on Discrete
  algorithms}, pages 1399--1413. SIAM, 2014.

\bibitem[BC19]{bringmann2019polyline}
K.~Bringmann and B.R. Chaudhury.
\newblock Polyline simplification has cubic complexity.
\newblock In {\em Proceedings of the International Symposium on Computational
  Geometry}, pages 18:1--18:16, 2019.

\bibitem[BD24]{blank2024faster}
Lotte Blank and Anne Driemel.
\newblock A faster algorithm for the fr{\'e}chet distance in 1d for the
  imbalanced case.
\newblock In {\em Proceedings of the Annual European Symposium on Algorithms},
  pages 28--1. Schloss Dagstuhl--Leibniz-Zentrum f{\"u}r Informatik, 2024.

\bibitem[BDNP22]{BDNP2022}
K.~Bringmann, A.~Driemel, A.~Nusser, and I.~Psarros.
\newblock Tight bounds for approximate near neighbor searching for time series
  under {F}r\'{e}chet distance.
\newblock In {\em Proceedings of the ACM-SIAM Symposium on Discrete
  Algorithms}, pages 517--550, 2022.

\bibitem[BK15]{bringmann2015improved}
Karl Bringmann and Marvin K{\"u}nnemann.
\newblock Improved approximation for {F}r{\'e}chet distance on c-packed curves
  matching conditional lower bounds.
\newblock In {\em Proceedings of International Symposium on Algorithms and
  Computation}, pages 517--528. Springer, 2015.

\bibitem[BM16]{bringmann2016approximability}
Karl Bringmann and Wolfgang Mulzer.
\newblock Approximability of the discrete {F}r{\'e}chet distance.
\newblock {\em Journal of Computational Geometry}, 7(2):46--76, 2016.

\bibitem[BOS19]{buchin2019seth}
Kevin Buchin, Tim Ophelders, and Bettina Speckmann.
\newblock {SETH} says: Weak {F}r{\'e}chet distance is faster, but only if it is
  continuous and in one dimension.
\newblock In {\em Proceedings of the Annual ACM-SIAM Symposium on Discrete
  Algorithms}, pages 2887--2901. SIAM, 2019.

\bibitem[Bri14]{bringmann2014walking}
Karl Bringmann.
\newblock Why walking the dog takes time: {F}r\'echet distance has no strongly
  subquadratic algorithms unless {SETH} fails.
\newblock In {\em Proceedings of the IEEE 55th Annual Symposium on Foundations
  of Computer Science}, pages 661--670. IEEE, 2014.

\bibitem[CF21]{colombe2021approximating}
Connor Colombe and Kyle Fox.
\newblock Approximating the (continuous) {F}r{\'e}chet distance.
\newblock In {\em Proceedings of International Symposium on Computational
  Geometry}. Schloss Dagstuhl-Leibniz-Zentrum f{\"u}r Informatik, 2021.

\bibitem[CH23a]{CH2023}
Siu-Wing Cheng and H.~Huang.
\newblock Approximate nearest neighbor for polygonal curves under
  {F}r{\'{e}}chet distance.
\newblock In {\em Proceedings of the 50th International Colloquium on Automata,
  Languages and Programming}, 2023.

\bibitem[CH23b]{cheng2022curve}
Siu-Wing Cheng and Haoqiang Huang.
\newblock Curve simplification and clustering under {F}r\'{e}chet distance.
\newblock In {\em Proceedings of the ACM-SIAM Symposium on Discrete
  Algorithms}, pages 1414--1432, 2023.

\bibitem[CH24]{cheng2023solving}
Siu-Wing Cheng and Haoqiang Huang.
\newblock Solving {F}r\'echet distance problems by algebraic geometric methods.
\newblock In {\em Proceedings of the ACM-SIAM Symposium on Discrete
  Algorithms}, 2024.

\bibitem[CH25]{Cheng2024FrchetDI}
Siu-Wing Cheng and Haoqiang Huang.
\newblock Fr{\'e}chet distance in subquadratic time.
\newblock In {\em Proceedings of the Annual ACM-SIAM Symposium on Discrete
  Algorithms}, pages 5100--5113. SIAM, 2025.

\bibitem[CHJ25]{SHJ}
Siu-Wing Cheng, Haoqiang Huang, and Le~Jiang.
\newblock Simplification of trajectory streams.
\newblock In {\em Proceedings of International Symposium on Computational
  Geoemtry}, 2025.

\bibitem[CR18]{chan2018improved}
Timothy~M Chan and Zahed Rahmati.
\newblock An improved approximation algorithm for the discrete {F}r{\'e}chet
  distance.
\newblock {\em Information Processing Letters}, 138:72--74, 2018.

\bibitem[DHPW12]{driemel2012approximating}
Anne Driemel, Sariel Har-Peled, and Carola Wenk.
\newblock Approximating the {F}r{\'e}chet distance for realistic curves in near
  linear time.
\newblock {\em Discrete \& Computational Geometry}, 48:94--127, 2012.

\bibitem[EM94]{eiter1994computing}
Thomas Eiter and Heikki Mannila.
\newblock Computing the discrete {F}r{\'e}chet distance.
\newblock Technical report, Vienna University of Technology, 1994.

\bibitem[FFK23]{filtser2023approximate}
Arnold Filtser, Omrit Filtser, and Matthew~J Katz.
\newblock Approximate nearest neighbor for curves: simple, efficient, and
  deterministic.
\newblock {\em Algorithmica}, 85(5):1490--1519, 2023.

\bibitem[GHMS93]{guibas1993approximating}
L.J. Guibas, J.E. Hershberger, J.S.B. Mitchell, and J.S. Snoeyink.
\newblock Approximating polygons and subdivisions with minimum-link paths.
\newblock {\em International Journal of Computational Geometry \&
  Applications}, 3(4):383--415, 1993.

\bibitem[GMMW18]{gudmundsson2018fast}
Joachim Gudmundsson, Majid Mirzanezhad, Ali Mohades, and Carola Wenk.
\newblock Fast {F}r{\'e}chet distance between curves with long edges.
\newblock In {\em Proceedings of the 3rd International Workshop on Interactive
  and Spatial Computing}, pages 52--58, 2018.

\bibitem[God91]{Godau1991ANM}
Michael Godau.
\newblock A natural metric for curves - computing the distance for polygonal
  chains and approximation algorithms.
\newblock In {\em Proceedings of the Symposium on Theoretical Aspects of
  Computer Science}, 1991.

\bibitem[M20]{mirzanezhad2020approximate}
Mirzanezhad M.
\newblock On the approximate nearest neighbor queries among curves under the
  {F}r\'{e}chet distance.
\newblock {\em arXiv preprint arXiv:2004.08444}, 2020.

\bibitem[vdHO24]{vanderhorst_et_al:LIPIcs.SoCG.2024.63}
Thijs van~der Horst and Tim Ophelders.
\newblock {Faster Fr\'{e}chet Distance Approximation Through Truncated
  Smoothing}.
\newblock In {\em Proceedings of the 40th International Symposium on
  Computational Geometry}, pages 63:1--63:15, 2024.

\bibitem[vdHvKOS23]{van2023subquadratic}
Thijs van~der Horst, Marc van Kreveld, Tim Ophelders, and Bettina Speckmann.
\newblock A subquadratic $n^\varepsilon$-approximation for the continuous
  {F}r{\'e}chet distance.
\newblock In {\em Proceedings of the Annual ACM-SIAM Symposium on Discrete
  Algorithms}, pages 1759--1776. SIAM, 2023.

\bibitem[vdKKL{\etalchar{+}}19]{van2019global}
M.~van~de Kerkhof, I.~Kostitsyna, M.~L{\"o}ffler, M.~Mirzanezhad, and C.~Wenk.
\newblock Global curve simplification.
\newblock In {\em Proceedings of the European Symposium on Algorithms}, pages
  67:1--67:14, 2019.

\bibitem[vKLW18]{van2018optimal}
M.~van Kreveld, M.~L{\"o}ffler, and L.~Wiratma.
\newblock On optimal polyline simplification using the {H}ausdorff and
  {F}r{\'e}chet distance.
\newblock In {\em Proceedings of the International Symposium on Computational
  Geometry}, pages 56:1--56:14, 2018.

\end{thebibliography}
	
\end{document}